\documentclass[10pt]{article}
\usepackage{anysize}
\usepackage{comment}
\usepackage{booktabs} 
\usepackage{amsmath,amsthm,amsfonts} 
\usepackage{graphicx}
\usepackage{textcomp}
\usepackage{url}
\usepackage{mathtools}
\usepackage{xcolor}
\usepackage{multirow}
\DeclarePairedDelimiter\ceil{\lceil}{\rceil}
\DeclarePairedDelimiter\floor{\lfloor}{\rfloor}

\DeclareMathOperator{\len}{len}
\DeclareMathOperator{\sspan}{span}
\DeclareMathOperator{\optone}{OPT_1}
\DeclareMathOperator{\opttwo}{OPT_{BSHM}}
\DeclareMathOperator{\cn}{{\it v}}

\allowdisplaybreaks[4]
\marginsize{2cm}{2cm}{2.5cm}{2.5cm}

\usepackage[ruled,linesnumbered]{algorithm2e} 

\theoremstyle{assumption}

\theoremstyle{definition}
\newtheorem{definition}{Definition}[section]

\theoremstyle{proposition}
\newtheorem{proposition}{Proposition}[subsection]

\newtheorem{theorem}{Theorem}[section]

\theoremstyle{property}
\newtheorem{property}{Property}[section]

\theoremstyle{remark}
\newtheorem{remark}{Remark}[subsection]
\theoremstyle{corollary}
\newtheorem{corollary}[theorem]{Corollary}
\theoremstyle{lemma}
\newtheorem{lemma}[theorem]{Lemma}

\begin{document}
\title{Analysis of Busy-Time Scheduling on Heterogeneous Machines}

\author{
        Mozhengfu Liu and Xueyan Tang\\
        School of Computer Science and Engineering\\
        Nanyang Technological University\\
        Singapore\\
        \{mozhengfu.liu,~asxytang\}@ntu.edu.sg
}

\maketitle

\begin{abstract}
This paper studies a generalized busy-time scheduling model on heterogeneous machines. The input to the model includes a set of jobs and a set of machine types. Each job has a size and a time interval during which it should be processed. Each job is to be placed on a machine for execution. Different types of machines have distinct capacities and cost rates. The total size of the jobs running on a machine must always be kept within the machine's capacity, giving rise to placement restrictions for jobs of various sizes among the machine types. Each machine used is charged according to the time duration in which it is busy, i.e., it is processing jobs. The objective is to schedule the jobs onto machines to minimize the total cost of all the machines used. We develop an $O(1)$-approximation algorithm in the offline setting and an $O(\mu)$-competitive algorithm in the online setting (where $\mu$ is the max/min job length ratio), both of which are asymptotically optimal.
\end{abstract}

\section{Introduction}

In this paper, we study generalized busy-time scheduling on heterogeneous machines. In this model, each job is specified by a size and a time interval of execution. The jobs are to be scheduled onto machines nonpreemptively. At any time, the total size of the jobs running on a machine cannot exceed the machine’s capacity. Each machine used is charged proportional to its busy time which is the total length of the time periods in which it is processing jobs. Multiple types of machines with different capacities and cost rates are available. The goal is to schedule the jobs onto machines to minimize the accumulated cost of all the machines used. We focus on the algorithmic aspects of the above model and aim to develop effective solutions
in both the offline and online settings.

Our busy-time scheduling model has useful applications.
Major cloud providers such as Amazon EC2 \cite{amazonec2}, Google Cloud \cite{googlecloud2} and Microsoft Azure \cite{microsoftazure2} provide different types of predefined server instances (virtual machines) for customers to rent at different rates. Jobs with various resource demands have placement constraints among the server types accordingly. The servers rented from the clouds are charged according to their working hours. It is critical for cloud users to decide the types and numbers of servers to rent in order to minimize the total rental cost for processing their jobs. Our model elegantly captures the ``pay-as-you-go'' billing feature of the clouds and the goal of optimizing the monetary expenses for cloud users.

There have been quite a few studies on busy-time scheduling, but almost all of them assumed homogeneous machines only. Earlier work has investigated scheduling interval jobs of uniform sizes, so that each machine can run at most a fixed number of jobs simultaneously \cite{winkler2003wavelength, alicherry2003line, kumar2005approximation, flammini2010minimizing, shalom2014online, mertzios2015optimizing, chang2017lp}. This problem was termed interval scheduling with bounded parallelism and is NP-hard. More recent work has addressed scheduling interval jobs of non-uniform sizes, where the number of jobs that can share a machine is not fixed \cite{spaa2014, tpds2016, ipdps2016, ren2016competitiveness, ren2016clairvoyant, azar2017tight}. This problem was termed MinUsageTime dynamic bin packing. For both problems, the objective is to minimize the total machine busy time for processing a given set of jobs. In the offline setting where all the jobs are known, there exist $O(1)$-approximation algorithms for both problems \cite{alicherry2003line, kumar2005approximation, flammini2010minimizing, khandekar2015real, ren2016clairvoyant, chang2017lp}.
For jobs of uniform sizes, Flammini {\em et al.} \cite{flammini2010minimizing} presented a 4-approximation First Fit algorithm that schedules jobs in descending order of length. Chang {\em et al.} \cite{chang2017lp} proposed a 3-approximation GreedyTracking algorithm. They also showed that the work of Alicherry and Bhatia \cite{alicherry2003line} and that of Kumar and Rudra \cite{kumar2005approximation} imply 2-approximation algorithms. For jobs of non-uniform sizes, Khandekar {\em et al.} \cite{khandekar2015real} gave a 5-approximation algorithm by extending the work of
\cite{flammini2010minimizing}. Ren and Tang  \cite{ren2016clairvoyant} presented a 4-approximation dual coloring algorithm by extending the work of
\cite{kumar2005approximation}. Very recently, Buchbinder {\em et al.} \cite{buchbinder2021prediction} presented algorithms with improved asymptotic approximation ratios.
In the online setting where jobs are released when they are to start execution and the job lengths are not known until they complete execution, the competitiveness of scheduling is bounded from below by the variation of job lengths for both problems \cite{spaa2014, tpds2016}. That is, the competitive ratio of any online algorithm is $\Omega(\mu)$, where $\mu$ is the max/min job length ratio among all the jobs to schedule. It has been shown that the First Fit algorithm achieves a competitive ratio of $\mu+3$ for scheduling jobs of non-uniform sizes, closely matching the lower bound \cite{ren2016competitiveness}. In the case that the length of a job is revealed when it is released, the competitiveness of scheduling has a tight bound of $\Theta(\sqrt{\log\mu})$  \cite{azar2017tight}. In addition, recent work has also considered discrete charging unit \cite{tpds2019}, machine launch cost \cite{tpds2020}, and load predictions \cite{buchbinder2021prediction} in busy-time scheduling. However, none of the above work has studied multiple types of machines.

With heterogeneous machines, jobs have different restrictions on which types of machines can process them. In addition, various machine types can differ in the normalized cost rate per capacity unit. As a result, the cost of scheduling each job depends on not only other overlapping jobs scheduled on the same machine but also the machine type. To the best of our knowledge, the only work that considered heterogeneous machines was \cite{ren2020heterogeneous}, which investigated just two simple cases in which the normalized cost rate per capacity unit increases or decreases monotonically with the machine capacity. In both cases, it was shown that there exist $O(1)$-approximation and $O(\mu)$-competitive algorithms. The authors of \cite{ren2020heterogeneous} conjectured that in the general case, the asymptotic approximability and competitiveness of scheduling would be dependent on the cost and capacity profiles of the machine types. In this paper, we close this open problem by developing $O(1)$-approximation and $O(\mu)$-competitive algorithms in the offline and online settings respectively for {\em any} set of machine types and {\em any} set of jobs, when there are plenty of machines available for each type.

\section{Problem Definition}

Formally, the input to the Busy-time Scheduling on Heterogeneous Machines (BSHM) model includes a set of jobs $\mathcal{J}$ and a set of machine types $\mathcal{M}$.

Each job $J\in{\mathcal{J}}$ has its size $s(J)$ which represents the resource demand for processing $J$, and its time interval of execution  $I(J):=[I(J)^-,I(J)^+)$.
We refer to $I(J)$ as $J$'s {\em active interval} and say that $J$ is {\em active} during $I(J)$. We also refer to the two endpoints $I(J)^-$ and $I(J)^+$ of $I(J)$ as $J$'s start and end times respectively. We denote the length of $I(J)$ by $\len(J):=I(J)^+-I(J)^-$ and refer to it as the {\em job length}.
Let $\mu:=\max_{J\in{\mathcal{J}}} \len(J) / \min_{J\in{\mathcal{J}}} \len(J)$ denote the max/min job length ratio. Without loss of generality, we assume that the maximum and minimum job lengths are $\mu$ and $1$ respectively.

Each job needs to be scheduled to run on a single machine.
Let $\mathcal{M} = \{1,2,\ldots,|\mathcal{M}|\}$ be the set of the indices of all machine types available, where every machine type indexed by $z\in{\{1,2,\ldots,|\mathcal{M}|\}}$ has a cost rate $r_z$ (per unit time) and a resource capacity $g_z$.
At any time, the total size of the jobs running on a machine cannot exceed the machine capacity. Each machine used is charged at its cost rate for the time duration in which it is processing at least one job. There are sufficient machines of each type available. The objective of BSHM is to minimize the total cost of machine usage for processing all the jobs $\mathcal{J}$.

Note that if two different machine types satisfy $g_i \leq g_j$ and $r_i \geq r_j$, then type-$i$ machines will not be needed for processing jobs because any type-$i$ machine used can be replaced by a type-$j$ machine that has at least the same capacity but lower or equal cost. Thus, without loss of generality, we assume that the machine types have distinct capacities and it holds that $g_1 < g_2 < \dots < g_{|\mathcal{M}|}$ and $r_1 < r_2 < \dots < r_{|\mathcal{M}|}$.

We study both the offline and online settings of BSHM. The difference between the two settings lies in how much information about $\mathcal{J}$ can be used for scheduling each job. In the offline setting, all the information about $\mathcal{J}$ can be used, while in the online setting, only the information available before each job $J$ starts can be used for scheduling $J$, i.e., this includes the start times and sizes of the jobs started before $I(J)^-$ and the end times of the jobs ended before $I(J)^-$.

The performance of an offline algorithm or an online algorithm is often characterized by its {\em approximation ratio} or {\em competitive ratio}, i.e., the worst-case ratio between a solution constructed by the algorithm and an optimal solution over all instances of the problem \cite{borodin1998,vazirani2013approximation}. To facilitate algorithm analysis, we assume that the cost rate of each machine type $z$ is a power of $8$, i.e., $r_z\in{}\{8^{n}:n\in{}\mathbb{Z}\}$, where $\mathbb{Z}$ denotes the set of all integers.
This assumption will cause us to lose at most a factor of $8 = O(1)$ in deriving the approximation or competitive ratio of any algorithm.
Specifically, for each $z\in \mathcal{M}$, suppose $c_z$ is the real cost rate of machine type $z$, which can be any positive value, and $r_z\in{}\{8^{n}:n\in{}\mathbb{Z}\}$ is the power of $8$ integer such that $\frac{1}{8} \cdot r_z < c_z \leq{} r_z$, which is the assumed cost rate for machine type $z$ to be used throughout the rest of this paper.
For any scheduling algorithm $ALG$, let $N(z,t)$ be the number of type-$z$ machines used by $ALG$ at time instant $t$ for a given set of jobs $\mathcal{J}$. For the two sets of cost rates of machine types $\{c_z: z\in \mathcal{M}\}$ and $\{r_z: z\in \mathcal{M}\}$, the optimal scheduling for $\mathcal{J}$ can be different. Let $O_c(z,t)$ be the number of type-$z$ machines used by the optimal scheduling for $\{c_z: z\in \mathcal{M}\}$ at time instant $t$, and $O_r(z,t)$ be the number of type-$z$ machines used by the optimal scheduling for $\{r_z: z\in \mathcal{M}\}$ at time instant $t$. If it is shown that $\int_t \sum_{z\in \mathcal{M}} N(z,t)\cdot r_z \,\mathrm{d}t \leq \alpha\cdot \int_t \sum_{z\in \mathcal{M}} O_r(z,t)\cdot r_z \,\mathrm{d}t$, we have
\begin{equation*}
\begin{aligned}
\int_t \sum_{z\in \mathcal{M}} N(z,t)\cdot c_z \,\mathrm{d}t
&\leq \int_t \sum_{z\in \mathcal{M}} N(z,t)\cdot r_z \,\mathrm{d}t\\
&\leq \alpha\cdot \int_t \sum_{z\in \mathcal{M}} O_r(z,t)\cdot r_z \,\mathrm{d}t\\
&\leq \alpha\cdot \int_t \sum_{z\in \mathcal{M}} O_c(z,t)\cdot r_z \,\mathrm{d}t\\
&\leq 8 \alpha\cdot \int_t \sum_{z\in \mathcal{M}} O_c(z,t)\cdot c_z \,\mathrm{d}t,\\
\end{aligned}
\end{equation*}
where the third inequality is because $O_r(z,t)$ is the optimal scheduling for the set of cost rates $\{r_z: z\in \mathcal{M}\}$. Thus, for the purpose of studying the asymptotic approximability and competitiveness of scheduling, it suffices to assume that the cost rates are powers of $8$.

To facilitate presentation, we further define some notations.
We denote by $\opttwo(\mathcal{X})$ the optimal cost of scheduling any given set of jobs $\mathcal{X}$ for the BSHM problem. For any job $J$, $m(J)$ denotes the machine type in $\{1,2,\ldots,|\mathcal{M}|\}$ such that $s(J)\in{}(g_{m(J)-1},g_{m(J)}]$, i.e., $m(J)$ is the lowest-indexed machine type that can accommodate job $J$. We refer to $m(J)$ as the {\em exact machine type} of $J$. For any set of jobs $\mathcal{X}$, $S(\mathcal{X}) := \sum_{J\in{}\mathcal{X}}{s(J)}$ denotes the total size of these jobs, and $\sspan (\mathcal{X}) := \cup_{J\in{}\mathcal{X}} I(J)$ denotes the time interval(s) in which at least one job in $\mathcal{X}$ is active. For any set of jobs $\mathcal{X}$ and any time instant $t$, $\mathcal{X}(t) :=  \{J\in{}\mathcal{X}:t\in{}I(J)\}$ denotes the active jobs in $\mathcal{X}$ at time $t$, and $S(\mathcal{X},t) := S(\mathcal{X}(t))$ denotes the total size of the active jobs at time $t$.

Details of all the missing proofs in the analysis are given in the corresponding sections of appendices.

\section{Preliminaries}

\subsection{Cost-per-capacity graph}

The main challenge for the general BSHM problem comes from the arbitrary order of the normalized cost rates per capacity unit among different machine types. We construct a directed graph to describe the relationships among the machine types in terms of their normalized cost rates per capacity unit. The graph is referred to as the {\em cost-per-capacity graph}. This graph was also indicated in \cite{ren2020heterogeneous}.

\begin{definition}
In the cost-per-capacity graph, each node $i$ represents a machine type $i\in{} \mathcal{M}$.
Each node $i$ has a directed edge pointing to node $p(i):=\min \{j: j>i \land r_j/g_j < r_i/g_i \}$ if such $p(i)$ exists (i.e., $p(i)$ is the lowest-indexed machine type above $i$ that has a lower normalized cost rate per capacity unit than $i$).
\end{definition}

\begin{proposition} \label{pro:forest}
The cost-per-capacity graph is a forest.
\end{proposition}

For simplicity, we shall use the terms ``node'' and ``machine type'' (or ``type'') interchangeably. We say that $p(i)$ is the parent of $i$, and $i$ is a child of $p(i)$. Let $f(i)$ denote the set of children of node $i$: $f(i):=\{z:p(z)=i\}$. Let $P(i)$ denote the set of nodes including node $i$ and all its ancestors $p(i)$, $p(p(i))$, $p(p(p(i)))$, $\ldots$.
Let $A_i$ denote the set of nodes in the tree rooted at node $i$: $A_i:=\{z: i\in P(z)\}$.
Let $v(i)$ denote the lowest-indexed node in $A_i$: $v(i):=\min{A_i}$.

Let $y(i)$ denote the set of younger siblings of node $i$: $y(i):=\{z:z<i \land p(z)=p(i)\}$, and let $e(i)$ denote the set of elder siblings of node $i$: $e(i):=\{z:z>i \land p(z)=p(i)\}$. Furthermore, let $T(i)$ denote the set of nodes $\left(\cup_{z\in{}P(i)} y(z) \right) \cup\{i\}$.

Table \ref{tab:forest} and Figure \ref{fig:forest} show an example of a set of machine types and the cost-per-capacity graph constructed accordingly. By the above definitions,
$p(10)=11$, $P(10)=\{10,11,13\}$,
$f(11)=\{9, 10\}$, $A_{11}=\{8,9,10,11\}$,
$y(11)=\{7\}$, and $e(11)=\{12\}$. In Figure \ref{fig:forest}, the nodes in grey constitute $T(10)=\{3,5,7,9,10\}$.

{
\begin{table}[h]
  \centering
  \caption{An example of a set of machine types}
  \label{tab:forest}
  \begin{tabular}{cccc}
    \toprule
    machine type $i$ & cost rate $r_i$ & capacity $g_i$ & ratio $r_i/g_i$\\
    \midrule
	13 & $8^6=262144$ & $100000$ & $2.62$\\
	12 & $8^5=32768$ & $3000$ & $10.92$\\
	11 & $8^4=4096$ & $1000$ & $4.10$\\
	10 & $8^3=512$ & $50$ & $10.24$\\
	9 & $8^2=64$ & $12$ & $5.33$\\
	8 & $8^1=8$ & $1$ & $8$\\
	7 & $8^0=1$ & $1/3$ & $3$\\
	6 & $8^{-1}=1/8$ & $1/40$  & $5$\\
	5 & $8^{-2}=1/64$ & $1/65$ & $1.02$\\
	4 & $8^{-3}=1/512$  & $1/1024$ & $2$\\
	3 & $8^{-4}=1/4096$ & $1/4096$ & $1$\\
	2 & $8^{-5}=1/32768$ & $1/100000$ & $3.05$\\
	1 & $8^{-6}=1/262144$ & $1/300000$ & $1.14$\\
  \bottomrule
\end{tabular}
\end{table}
}

\begin{figure}[h]
  \centering
  \includegraphics[width=7cm]{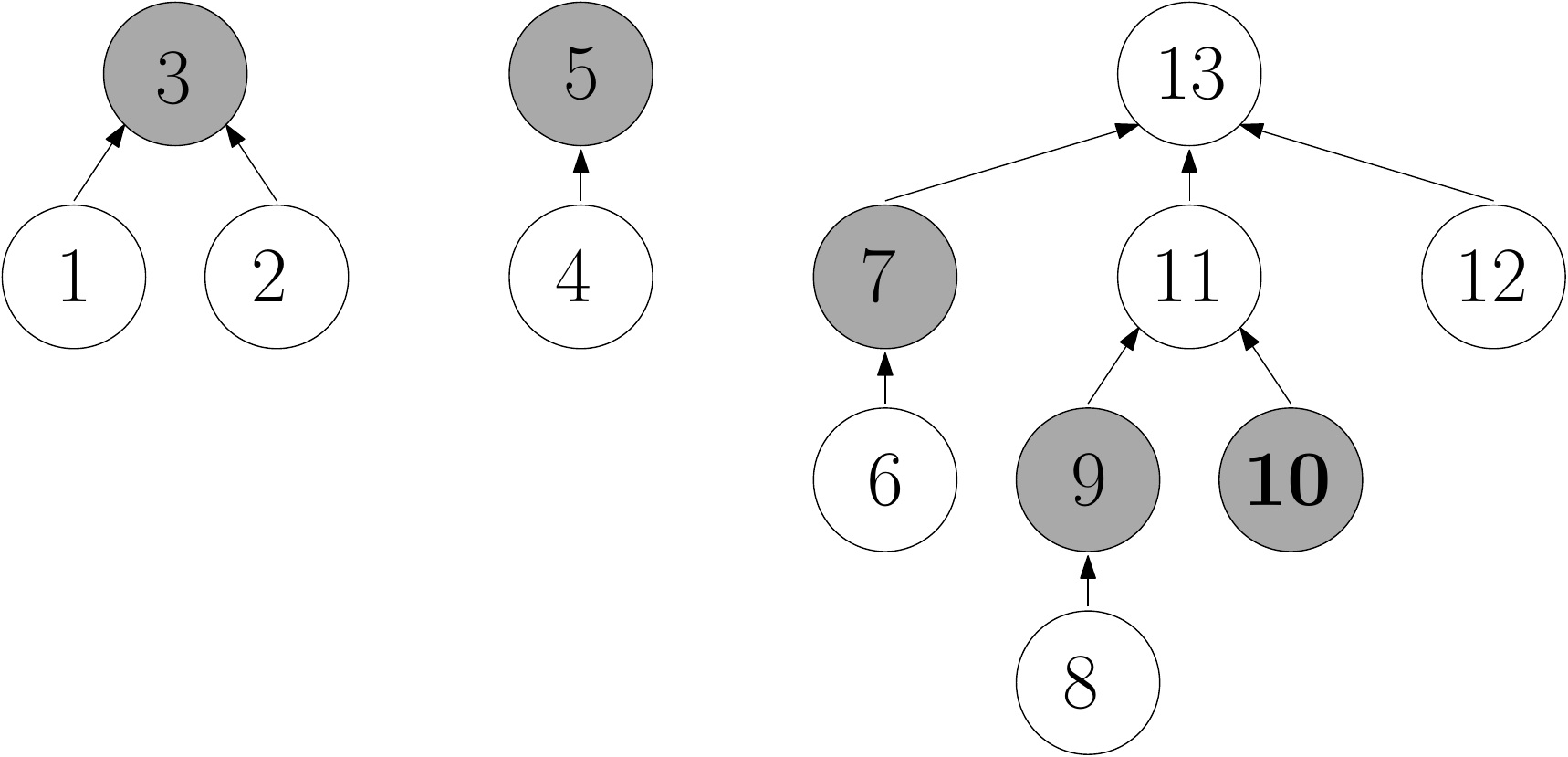}
  \caption{The cost-per-capacity graph for Table \ref{tab:forest}}
  \label{fig:forest}
\end{figure}

\begin{proposition} \label{pro:consecutive}
The node set of each tree has consecutive indexes, i.e., for any  $k\in \mathcal{M}$, $A_k=\{v(k),v(k)+1,\ldots{},k\}$.
\end{proposition}

\begin{proposition} \label{pro:T_ratio}
For any $k\in \mathcal{M}$ and any $y_1,y_2\in{}T(k)$, if $y_1\leq{}y_2$, then  $r_{y_1}/g_{y_1} \leq{} r_{y_2}/g_{y_2}$. That is, for the nodes in $T(k)$, their normalized cost rates per capacity unit is non-decreasing with indexes.
\end{proposition}

\begin{proposition} \label{pro:T_union}
For any $k\in{} \mathcal{M}$, the trees rooted at the nodes of $T(k)$, i.e., $\{A_z: z\in{}T(k)\}$, is a partitioning of $\{1,2,\ldots{},k\}$. Furthermore, suppose that $T(k)=\{i_1,i_2,\ldots,i_n\}$ where $i_1<i_2<\ldots<i_n$. Then $i_q=v(i_{q+1})-1$ for each $q=1,2,\ldots,n-1$.
\end{proposition}

\begin{proposition} \label{pro:T_two}
For any $k_0\in \mathcal{M}$ and any $k_1\in{}P(k_0)$, denoting $\min {T(k_0)\cap{}A_{k_1}}$ by $z_0$, we have:

(1) $T(k_0)\cap{}A_{k_1}=\left(\cup_{z\in{}P(k_0)\setminus{}P(k_1)} y(z)\right) \cup \{k_0\} =\{z:z\in{}T(k_0) \land z\geq{}z_0\}$;

(2) $T(k_0)\setminus{}A_{k_1}=\cup_{z\in{}P(k_1)} y(z)=\{z:z\in{}T(k_0) \land z<z_0\}$.
\end{proposition}

\subsection{One-shot job scheduling} \label{sec:optone}

To understand the optimal cost of BSHM, we start by considering scheduling jobs on heterogeneous machines at a single time instant and refer to this problem as {\em one-shot scheduling}. In the one-shot  scheduling problem, we relax the constraint that each job must be scheduled onto a single machine and allow a job to be divided into multiple pieces along its size dimension and each piece to be scheduled onto a distinct machine. However, we retain the restriction that all the machines onto which a job is scheduled must have capacities no less than the original size of the job. The goal of one-shot scheduling is to minimize the total cost rate of all the machines used for accommodating the jobs.

Note that BSHM does not allow jobs to be divided and it also enforces each job to be scheduled onto the same machine throughout its active interval. Therefore, the optimal cost $\opttwo(\mathcal{J})$ of BSHM for a set of jobs $\mathcal{J}$ is bounded from below by the accumulated cost of optimal one-shot scheduling for the active jobs $\mathcal{J}(t)$ at each time instant $t$, i.e.,
\begin{equation} \label{eq:opt2toopt1}
    \opttwo(\mathcal{J}) \geq \int_{t} \optone (\mathcal{J}(t)) \,\mathrm{d}t,
\end{equation}
where $\optone (\mathcal{J}(t))$ denotes the optimal cost rate of one-shot scheduling for the jobs $\mathcal{J}(t)$. Note that in one-shot scheduling, only the sizes of the jobs $\mathcal{J}(t)$ matter while the active intervals of the jobs are irrelevant. We shall use the above lower bound in the analysis of algorithm performance with respect to $\opttwo(\mathcal{J})$.

We define a {\em machine configuration} $w$ as a set of numbers $\{w(z):  z\in \mathcal{M}\}$, each representing the number of machines for a machine type.
Given a set of jobs $\mathcal{J}^{1d}$, the one-shot scheduling problem essentially seeks a minimum-cost machine configuration described by the following integer linear program:\footnote{Since only the sizes of jobs matter in one-shot scheduling, we use the notation $\mathcal{J}^{1d}$ for the input to one-shot scheduling, differentiating it from the input $\mathcal{J}$ to BSHM.} \begin{equation}\label{opt:1d}
\begin{aligned}
\min \quad
& \sum_{z\in \mathcal{M}} w(z)r_z\\
\textrm{s.t.} \quad
& S(\{J\in{}\mathcal{J}^{1d}:m(J)\geq{}i\}) \leq{}\sum_{z\geq i} w(z)\cdot g_z,
\forall{} i\in \mathcal{M};\\
& w(z)\in{}\mathbb{Z}_{\geq{}0},\forall z\in \mathcal{M};\\
\end{aligned}
\end{equation}

The first constraint above says that the total capacity of the machines of types at least $i$ must be no less than the total size of the jobs whose exact machine types are at least $i$.

Since integer programming is NP-complete in general, it is not easy to
derive a closed-form optimal machine configuration for one-shot scheduling. Moreover, the optimal machine configuration may not be unique. For the purpose of analysis, we present several properties of a particular optimal machine configuration. Let $k_0:=\max \{m(J):J\in{}\mathcal{J}^{1d}\}$ be the exact machine type of the highest index among all the jobs.
Then, any feasible machine configuration must include at least one machine of type no less than $k_0$. For each machine type $z\in \mathcal{M}$, let $H_z = \{J\in\mathcal{J}^{1d}:m(J)\in{}A_z \}$ denote all the jobs whose exact machine types are in the tree rooted at type $z$.

\begin{theorem} \label{thm:optmachine}
There exists an optimal machine configuration $w^*$ 
satisfying all the following properties: \\
(1) $k_{opt}:=\max \{z: w^* (z)>0\}\in{}P(k_0)$, i.e., the highest-indexed machine type $k_{opt}$ used is either $k_0$ or an ancestor of $k_0$ in the cost-per-capacity graph; \\
(2) $\sum_{z\in{A_{z_0}\setminus{}\{z_0\}}} w^* (z) r_z < r_{z_0}$ for each $z_0 \in \mathcal{M}$, i.e., the total cost of the machines in the tree rooted at each machine type $z_0$ (except type $z_0$ itself) is less than the cost of one type-$z_0$ machine; \\
(3) $\floor*{\frac{S(H_{k_{opt}} )}{g_{k_{opt}}}} \leq w^* (k_{opt}) \leq \ceil*{ \frac{S(H_{k_{opt}})}{g_{k_{opt}}} }$, i.e., the number of type-$k_{opt}$ machines can almost host all the jobs whose exact machine type is in the tree rooted at type $k_{opt}$.
\end{theorem}

In the rest of this paper, the optimal machine configuration discussed shall always refer to one with the above properties.

\section{The Offline Setting}

\subsection{The offline algorithm: $ALG_{offline}$}

We now discuss the offline BSHM problem, which is NP-hard since it is a generalization of interval scheduling with bounded parallelism. Consider a set of jobs $\mathcal{J}$ for BSHM. For each $z\in{} \mathcal{M}$, let $R_z = \{J\in\mathcal{J}:m(J)\in{}A_z \}$ denote the set of jobs whose exact machine types are in the tree rooted at type $z$. Algorithm \ref{alg:offline} shows our offline algorithm $ALG_{offline}$ for BSHM. The algorithm iteratively determines the set of jobs $K_z$ scheduled onto each machine type $z$ in descending order of type indexes (line 1). For each machine type $z$, we consider all the unscheduled jobs in $R_z$, denoted by $\mathcal{R}_z$ (line 2). Note that $\mathcal{R}_z$ includes jobs whose exact machine types are $z$ (denoted by $\mathcal{R}^h_z$) and jobs whose exact machine types are $z$'s descendants in the cost-per-capacity graph (given by $\mathcal{R}_z\setminus{}\mathcal{R}^h_z$). The jobs in $\mathcal{R}^h_z$ must be scheduled onto type-$z$ machines (line 4) since all the machine types indexed higher than $z$ have been considered before. For each job $J$ in $\mathcal{R}_z\setminus{}\mathcal{R}^h_z$,
we check whether it is {\em cost-effective} to open a type-$z$ machine throughout $J$'s active interval. If so, $J$ is scheduled onto type-$z$ machines (lines 9-10). If not, $J$ is left to subsequent iterations and will be scheduled onto a descendant machine type of $z$ (more specifically, a machine type in the tree rooted at a child of $z$ that includes $m(J)$). To decide whether it is cost-effective to open a type-$z$ machine at a time instant $t$, we examine all the active jobs in $\mathcal{R}_z$ at time $t$. If there exists at least one job in $\mathcal{R}^h_z$ active at $t$ (i.e., $t \in \sspan(\mathcal{R}^h_z)$), time instant $t$ is considered cost-effective (line 7). Otherwise, all the jobs active at $t$ are from $\mathcal{R}_z\setminus{}\mathcal{R}^h_z$.
Thus, each job active at $t$ must have its exact machine type in one of the trees rooted at $z$'s children. For each child type $x$ of $z$, we compute the number of type-$x$ machines needed to host all the jobs $F_{x,t}$ whose exact machine types are in the tree rooted at $x$ (lines 5-6). If the total cost of the machines of $z$'s child types calculated in this way exceeds 1/3 of the cost of a type-$z$ machine, time instant $t$ is considered cost-effective (line 7). Note that the set of active jobs does not change between two successive job starts/ends. Thus, the cost-effectiveness only needs to be evaluated once for each interval between two successive job starts/ends. As a result, the algorithm runs in polynomial time.

\begin{algorithm}[h]
\KwIn{A set of jobs $\mathcal{J}$.}
\KwOut{A schedule for $\mathcal{J}$.}

\For{$z=|\mathcal{M}|,|\mathcal{M}|-1,\dots,1$}{
	$\mathcal{R}_z\gets{}R_z \setminus (\cup_{i=|\mathcal{M}|,|\mathcal{M}|-1,\ldots,z+1}K_i)$\;
	$\mathcal{R}^h_z\gets\{J\in\mathcal{R}_z:m(J)=z\}$\;
	$K_z\gets{}\mathcal{R}^h_z$\;
	$F_{x,t}\gets\{J\in\mathcal{R}_z(t):m(J)\in{}A_x\}$, for each $x\in{f(z)}$ and $t$\;
	$c_{z,t}\gets\sum_{x\in{}f(z)} \ceil*{\frac{S(F_{x,t})}{g_x}r_x}$, for each $t$\;
	$\mathcal{T}_z\gets{}\sspan(\mathcal{R}^h_z)\cup\{t:c_{z,t}\geq{}\frac{1}{3}r_z\}$\;
	\For{each $J\in{}\mathcal{R}_z\setminus{}\mathcal{R}^h_z$}{
		\If{$I(J)\subset{}\mathcal{T}_z$}
		{add $J$ into $K_z$\;}
	}	
	schedule jobs in $K_z$ onto type-$z$ machines by using the dual coloring algorithm (see \cite{ren2016clairvoyant})\;
}
\caption{$ALG_{offline}$}
\label{alg:offline}
\end{algorithm}

It is easy to infer that $ALG_{offline}$ eventually assigns each job $J$ to either the exact machine type $m(J)$ or one of $m(J)$'s ancestors. After determining the set of jobs $K_z$ for all machine types, we use an existing dual coloring algorithm \cite{ren2016clairvoyant} to schedule the jobs in each $K_z$ onto the machines of the corresponding type (line 13). Dual coloring is a $4$-approximation algorithm for scheduling jobs onto homogeneous machines.

The output of the $ALG_{offline}$ algorithm has the following properties for any time instant $t\in\sspan(\mathcal{J})$.

\begin{property} \label{property:off_head}
Let $k_{off}:= \max \{z:K_z (t)\neq \emptyset\}$ be the highest-indexed machine type used by $ALG_{offline}$ at time $t$, and $k_0:= \max \{m(J):J\in{}\mathcal{J}(t)\}$ be the highest-indexed exact machine type among the active jobs at time $t$. We have $k_{off}\in{}P(k_0)$, i.e., $k_{off}$ is either $k_0$ or an ancestor of $k_0$.
\end{property}

\begin{property} \label{property:off_upperb}
For any machine type $z \in \mathcal{M}$, we have $\sum_{i\in{A_z\setminus\{z\}}}
\ceil*{\frac{S(K_i,t)}{g_i}} \cdot r_i \leq 2\cdot{}r_z $, i.e., the total cost of the machines needed for hosting all the jobs assigned to $z$'s descendants
is bounded by the cost of a constant number of type-$z$ machines.
\end{property}

\begin{property} \label{property:off_lowerb}
For any machine type $z \in \mathcal{M}$ such that (a) $K_i (t)=\emptyset{}$ for each $i\in{}P(z)\setminus \{z\}$ ($i$ is $z$'s ancestor), (b) $K_z (t)\neq{}\emptyset{}$, and (c) $\mathcal{R}_z^h (t)=\emptyset{}$, we have $\sum_{x\in{}f(z)} S(R_x ,t)\cdot \frac{r_x}{g_x} \geq \frac{4}{21}\cdot{}r_z$, i.e., if no job is assigned to $z$'s ancestors and all the jobs assigned to type $z$ have exact machine types that are $z$'s descendants, then assigning all the jobs whose exact machine types are $z$'s descendants to the corresponding $z$'s children would incur a total cost at least a constant fraction of a type-$z$ machine.
\end{property}

We exploit the properties of the cost-per-capacity graph to analyze the $ALG_{offline}$ algorithm. In particular, Proposition \ref{pro:T_union} of the cost-per-capacity graph indicates that all machine types indexed from $1$ to any $k$ is a disjoint union of subtrees rooted at nodes from the set $T(k)$.
In our analysis,
$T(k)$ plays a critical role to bridge the cost of
$ALG_{offline}$
and the optimal cost of BSHM.
The general idea is as follows. For each time instant, we charge the cost of the optimal machine configuration onto only machine types $T(k_{opt})$ within an $O(1)$ factor (where $k_{opt}$ is the highest-indexed machine type used by the optimal configuration) (Section \ref{sec:first_app}). We also charge the cost of the machines used by the $ALG_{offline}$ algorithm onto only machine types $T(k_{off})$ within an $O(1)$ factor (where $k_{off}$ is the highest-indexed machine type used by $ALG_{offline}$). These charging mechanisms significantly reduce the set of machine types we need to consider. Finally, we establish the connections between the costs of $T(k_{opt})$ and $T(k_{off})$ by carefully analyzing different possible relationships between $k_{opt}$ and $k_{off}$ according to the definition of the $ALG_{offline}$ algorithm (Section \ref{sec:offline_analysis}).

\subsection{An $O(1)$ approximation of optimal one-shot scheduling}
\label{sec:first_app}
The optimal machine configuration $w^*$ for one-shot scheduling discussed in Section \ref{sec:optone} is not concrete enough. For the analysis of $ALG_{offline}$, we define an {\em alternative machine configuration} that is an $O(1)$ approximation of the optimal machine configuration.

Given a set of jobs $\mathcal{J}^{1d}$, let $k_0:=\max \{m(J):J\in{}\mathcal{J}^{1d}\}$. We first define a machine configuration $\cn_{z^*}$ given the highest-indexed machine type $z^*$ used, where $z^*$ is either $k_0$ or an ancestor of $k_0$ (i.e., $z^*\in{P(k_0)}$). The configuration $\cn_{z^*}$ uses only machine types $T(z^*)$. By Proposition \ref{pro:T_union}, $\{A_z: z\in{}T(z^*)\}$ is a partitioning of $\{1,2,\ldots{},z^*\}$. Thus, $\mathcal{J}^{1d}$ can be rewritten as $\cup_{i \in T(z^*)} H_i$, where $H_i$ is all the jobs in $\mathcal{J}^{1d}$ whose exact machine types are in the tree rooted at type $i$. In the configuration $\cn_{z^*}$, the number of type-$z^*$ machines is given by $\cn_{z^*}(z^*) = \ceil*{\frac{S(H_{z^*})}{g_{z^*}}}$, which is sufficient to host all the jobs in $H_{z^*}$. A capacity amount of $\cn_{z^*}(z^*)
\cdot g_{z^*} - S(H_{z^*})$ from type-$z^*$ machines is available for hosting other jobs $\mathcal{J}^{1d} \setminus{} H_{z^*}$. We use this available capacity to host the jobs in $H_i$ where $i \in T(z^*)\setminus{}\{z^*\}$ in decreasing order of $i$. The jobs in $H_i$
that cannot fit into the type-$z^*$ machines are put into type-$i$ machines. The number of type-$i$ machines is allowed to be fractional and is just enough to host all these jobs. Let $i^* \in T(z^*)\setminus{}\{z^*\}$ be a boundary type index such that
\begin{equation} \label{eq:boundary_machine_type1}
    \sum_{z \in T(z^*) \land z > i^*} S(H_z) \leq 
    \cn_{z^*}(z^*) \cdot g_{z^*},
\end{equation}
and
\begin{equation} \label{eq:boundary_machine_type2}
    \sum_{z \in T(z^*) \land z \geq i^*} S(H_z) > 
    \cn_{z^*}(z^*) \cdot g_{z^*}.
\end{equation}
Then, the machine configuration $\cn_{z^*}$ is given by:
\begin{equation*}
\begin{cases}
\cn_{z^*}(z^*) := \ceil*{\frac{S(H_{z^*})}{g_{z^*}}}, \\
\cn_{z^*}(i):=
\begin{cases}
0, \quad \text{~if~} i \in T(z^*)\setminus{}\{z^*\} \text{~and~} i > i^*, \\
\frac{\sum_{z \in T(z^*) \land z \geq i} S(H_z) - \cn(z^*) \cdot g_{z^*}}{g_{i}}, \quad \text{~if~} i = i^*, \\
\frac{S(H_{i})}{g_{i}}, \quad \text{~if~} i \in T(z^*)\setminus{}\{z^*\} \text{~and~} i < i^*, \\
0, \quad \text{~if~} i \notin T(z^*). \\
\end{cases}
\end{cases}
\end{equation*}

The machine configuration $\cn_{z^*}$ has the following properties:

\begin{proposition} \label{pro:cn&r'_total}
For any $z_0,z_1\in{}T(z^*)$ such that $z_0\leq{}z_1$, we have $\left| \sum_{z\in{}T(z^*) \land z_0\leq{}z\leq{}z_1} \left( S(H_z) \frac{r_z}{g_z} - \cn_{z^*} (z) r_z \right) \right| < r_{z^*}$.
\end{proposition}

\begin{proposition} \label{pro:cn&r'_head}
For any $z_0\in{}T(z^*)$, we have
$\sum_{z\in{}T(z^*) \land z\geq{}z_0} \cn_{z^*}(z)  r_z
\geq{} \sum_{z\in{}T(z^*) \land z\geq{}z_0} S(H_z) \frac{r_z}{g_z}$. \end{proposition}

In the machine configuration $\cn_{z^*}$, there may exist a $z^*$'s ancestor type $z^\triangle \in P(z^*)\setminus\{z^*\}$ such that the total cost of the machines in the tree rooted at $z^\triangle$ exceeds the cost of one type-$z^\triangle$ machine, i.e.,
\begin{equation}
    \sum_{z \in{T(z^*) \cap A_{z^\triangle}}} \cn_{z^*}(z) r_z > r_{z^\triangle}. \nonumber
\end{equation}
In this case, we say that the machine configuration $\cn_{z^*}$ is not {\em decent} in that the total cost can be reduced by replacing the machines in the tree rooted at $z^\triangle$ with one type-$z^\triangle$ machine. The {\em alternative machine configuration} is defined as the first decent configuration among $\cn_{z^*}$'s where $z^* \in P(k_0)$, i.e., $\cn_{z^\diamond}$ where $z^\diamond$ is the lowest-indexed type in $P(k_0)$ such that $\cn_{z^\diamond}$ is decent. By definition, the alternative machine configuration $\cn_{z^\diamond}$ has the following property:

\begin{proposition} \label{pro:alternative} In the alternative machine configuration $\cn_{z^\diamond}$, \\
(1) for each $z^\triangle \in P(z^\diamond)\setminus\{z^\diamond\}$, we have  $\sum_{z \in{T(z^\diamond) \cap A_{z^\triangle}}} \cn_{z^\diamond}(z) r_z \leq{} r_{z^\triangle}$; \\
(2) for each $z^*\in \{z\in P(k_0):z<z^\diamond\}$, there exists some $z^\triangle\in P(z^*)\setminus\{z^*\}$ such that $\sum_{z \in{T(z^*) \cap A_{z^\triangle}}} \cn_{z^*}(z) r_z > r_{z^\triangle}$.
\end{proposition}

Next, we prove that the alternative machine configuration $\cn_{z^\diamond}$ is an $O(1)$ approximation of the optimal machine configuration.

\begin{theorem} \label{thm:cn_app}
For any $\mathcal{J}^{1d}$, let $w^*$ be an optimal machine configuration and $\cn_{z^\diamond}$ be the alternative machine configuration. We have  $\frac{7}{15}\cdot{} \sum_{z\in{}T(z^\diamond)} \cn_{z^\diamond}(z) r_z \leq{} \optone (\mathcal{J}^{1d}) = \sum_{z\in \mathcal{M}} w^*(z) r_z \leq{} \frac{8}{7}\cdot{} \sum_{z\in{}T(z^\diamond)} \cn_{z^\diamond}(z) r_z$. \end{theorem}

\begin{proof} [Sketch of proof]

For the left inequality, we are to show that $\sum_{z\in{}T(z^\diamond)}{\cn_{z^\diamond}(z) r_z}
\leq{} \frac{15}{7} \cdot{} \sum_{z\in \mathcal{M}} w^*(z) r_z$.
Let $k_0=\max \{m(J):J\in{}\mathcal{J}^{1d}\}$ be the highest-indexed exact machine type. Denote by $k_{opt} := \max \{z:w^*(z)>0\}$ the highest-indexed machine type used by $w^*$. By the definition of $z^\diamond$ and the choice of $w^*$, both $z^\diamond$ and $k_{opt}$ are in $P(k_0)$. We analyze three cases separately: $z^\diamond = k_{opt}$, $z^\diamond < k_{opt}$, and $z^\diamond > k_{opt}$. For $z^\diamond = k_{opt}$, by Theorem \ref{thm:optmachine}, we study the subcases of $w^*(k_{opt}) = \floor*{\frac{S(H_{k_{opt}})}{g_{k_{opt}}}}$ and $w^*(k_{opt}) = \ceil*{\frac{S(H_{k_{opt}})}{g_{k_{opt}}}}$. The analysis makes use of Proposition \ref{pro:T_ratio}, Proposition \ref{pro:cn&r'_total}, and Proposition \ref{pro:alternative}. The details are given in the appendices.

For the right inequality, consider the machine configuration $w'$ naturally induced by the definition of $\cn_{z^\diamond}$: $w'(z):= \ceil*{\cn_{z^\diamond}(z)}$ for each $z \in \mathcal{M}$.
It is not hard to see that $w'$ is a feasible solution to the optimization problem (\ref{opt:1d}). On the other hand,
\begin{eqnarray}
\sum_{z=\mathcal{M}} w'(z) r_z
&=& \sum_{z\in{}T(z^\diamond)\setminus{}\{z^\diamond\}} w'(z) r_z + \cn_{z^\diamond} (z^\diamond) r_{z^\diamond} \nonumber \\
&\leq{}& \sum_{z\in{}T(z^\diamond)\setminus{}\{z^\diamond\}} (\cn_{z^\diamond}(z) r_z+r_z) + \cn_{z^\diamond} (z^\diamond) r_{z^\diamond} \nonumber \\
&\leq{}& \sum_{z\in{}T(z^\diamond)\setminus{}\{z^\diamond\}} r_z + \sum_{z\in T(z^\diamond)} \cn_{z^\diamond}(z) r_z \nonumber \\
&\leq{}& \frac{1}{7} r_{z^\diamond} + \sum_{z\in T(z^\diamond)} \cn_{z^\diamond}(z) r_z \nonumber \\
&\leq{}& \frac{8}{7}\cdot{} \sum_{z\in T(z^\diamond)} \cn_{z^\diamond}(z) r_z. \nonumber
\end{eqnarray}
Since $w^*$ is an optimal solution to the optimization problem (\ref{opt:1d}), we have $\sum_{z \in \mathcal{M}} w^*(z) r_z \leq{} \sum_{z \in \mathcal{M}} w'(z) r_z\leq{}\frac{8}{7}\cdot{}\sum_{z\in T(z^\diamond)} \cn_{z^\diamond}(z) r_z$.
\end{proof}

\subsection{$ALG_{offline}$ achieves $O(1)$ approximation}
\label{sec:offline_analysis}
Now we show that $ALG_{offline}$ is an $O(1)$-approximation algorithm.
Recall that $ALG_{offline}$ partitions all the jobs $\mathcal{J}$ into $K_z$'s for $z \in \mathcal{M}$.
The jobs in each $K_z$ are scheduled onto type-$z$ machines by the dual coloring algorithm. Take any time instant $t\in{}\sspan(\mathcal{J})$. The dual coloring algorithm \cite{ren2016clairvoyant} guarantees that the total cost rate of type-$z$ machines used at time $t$ is bounded by $4\cdot{}\ceil{S(K_z,t)/g_z} r_z$, where $S(K_z,t)$ is the total size of the active jobs in $K_z$ at time $t$.
Let $\cn_{z^\diamond}$ be the alternative machine configuration for the set of active jobs $\mathcal{J}(t)$ at time $t$. The following theorem shows that the total cost rate of all the machines used by $ALG_{offline}$ at time $t$ is bounded by $O(1)$ times the cost of the alternative machine configuration.

\begin{theorem} \label{thm:off_app}
$\sum_{z \in \mathcal{M}} \ceil*{S(K_z,t)/g_z} r_z \leq{} 21\cdot{}  \sum_{z\in{}T(z^\diamond)} \cn_{z^\diamond}(z) r_z$.
\end{theorem}

\begin{proof}
Let $k_{off}:=\max \{z:K_z(t)\neq{}\emptyset{}\}$ be the highest-indexed machine type used by $ALG_{offline}$ at time $t$. Let $k_0:=\max \{m(J):J\in\mathcal{J}(t)\}$ be the highest-indexed exact machine type among the active jobs at time $t$. Note that both $k_{off}$ and $z^\diamond$ are in $P(k_0)$ by Property \ref{property:off_head} and the definition of the alternative machine configuration for $\mathcal{J}(t)$.

Case 1: $z^\diamond\geq{}k_{off}$.

Note that $\{1,2,\ldots{},k_{off}\}\subset{}\{1,2,\ldots{},z^\diamond\} = \cup_{z\in{}T(z^\diamond)} A_z$, where the equality is by Proposition $\ref{pro:T_union}$. Also note that $R_{z^\diamond} (t)\neq{}\emptyset{}$ since $k_0\in{}A_{z^\diamond}$, where $R_z = \{J \in \mathcal{J}: m(J) \in A_z\}$ is the set of jobs whose exact machine types are in $A_z$.

Step 1:
For each $z\in{}T(z^\diamond)$, if $R_z(t) \neq \emptyset$,
we have
\begin{equation} \nonumber
\begin{aligned}
&\quad \sum_{i\in{}A_z} \ceil*{S(K_i,t)/g_i}\cdot r_i \\
&= \ceil*{S(K_z,t)/g_z}\cdot r_z + \sum_{i\in{}A_z\setminus\{z\}} \ceil*{S(K_i,t)/g_i}\cdot r_i \\
&\leq{} (\ceil*{S(K_z,t)/g_z} + 2) \cdot r_z \quad \text{(by~Property~\ref{property:off_upperb})} \\
&\leq{} 3\cdot{} \ceil*{S(R_z,t)/g_z}\cdot r_z. \quad \text{(since~} K_z(t)\subset{}R_z(t) \text{)}\\
\end{aligned}
\end{equation}
If $R_z(t) = \emptyset$, by Property \ref{property:off_head}, we have $K_i(t) = R_i(t) = \emptyset$ for each $i \in A_z$. Thus, it also holds that
\begin{equation} \nonumber
\sum_{i\in{}A_z} \ceil*{S(K_i,t)/g_i}\cdot r_i
\leq{} 3\cdot{} \ceil*{S(R_z,t)/g_z}\cdot r_z. \end{equation}
Therefore,
\begin{eqnarray}
\sum_{z=1,2,\ldots{},k_{off}} \ceil*{S(K_z,t)/g_z}\cdot r_z
&=& \sum_{z\in{}T(z^\diamond)} \sum_{i\in{}A_z} \ceil*{S(K_i,t)/g_i}\cdot r_i \nonumber \\
&\leq{}& 3\cdot{} \sum_{z\in{}T(z^\diamond)} \ceil*{S(R_z,t)/g_z}\cdot r_z. \label{eq:off_1_part1}
\end{eqnarray}

Step 2:
By the definition of the alternative machine configuration for $\mathcal{J}(t)$, $\cn_{z^\diamond} (z^\diamond) = \ceil*{S(R_{z^\diamond},t)/g_{z^\diamond}}$. Since the cost rate of each machine type is a power of $8$, we have
$\sum_{z\in{}T(z^\diamond)\setminus\{z^\diamond\}} r_z \leq{} \frac{1}{7} r_{z^\diamond}$.
\begin{eqnarray}
& & \sum_{z\in{}T(z^\diamond)} \ceil*{S(R_z,t)/g_z}\cdot r_z \nonumber \\
&\leq{}& \ceil*{S(R_{z^\diamond},t)/g_{z^\diamond}}\cdot r_{z^\diamond}+ \sum_{z\in{}T(z^\diamond)\setminus\{z^\diamond\}} \bigg(r_z + S(R_z,t)\cdot \frac{r_z}{g_z}\bigg) \nonumber \\
&\leq{}& \cn_{z^\diamond} (z^\diamond)\cdot r_{z^\diamond} +\frac{1}{7}\cdot r_{z^\diamond} + \sum_{z\in{}T(z^\diamond)\setminus\{z^\diamond\}} S(R_z,t)  \cdot \frac{r_z}{g_z} \nonumber \\
&\leq{}& \cn_{z^\diamond} (z^\diamond)\cdot r_{z^\diamond} +\frac{1}{7}\cdot r_{z^\diamond} + \sum_{z\in{}T(z^\diamond)} S(R_z,t)  \cdot \frac{r_z}{g_z} \nonumber \\
&\leq{} & \frac{8}{7}\cdot{} \cn_{z^\diamond} (z^\diamond)\cdot r_{z^\diamond} + \sum_{z\in T(z^\diamond)} \cn_{z^\diamond}(z) \cdot r_z \quad \text{(by~Proposition~\ref{pro:cn&r'_head})} \nonumber \\
&\leq{}& \frac{15}{7} \cdot{} \sum_{z\in T(z^\diamond)} \cn_{z^\diamond}(z) \cdot r_z.
\nonumber
\end{eqnarray}

In summary of steps 1 and 2,
$\sum_{z=1,2,\ldots{},k_{off}} \ceil*{S(K_z,t)/g_z} r_z \leq{} 3\cdot{}\sum_{z\in{}T(z^\diamond)} \ceil*{S(R_z,t)/g_z} r_z \leq{} \frac{45}{7}\cdot{} \sum_{z\in T(z^\diamond)} \cn_{z^\diamond}(z) r_z$.

Case 2: $k_{off} > z^\diamond$.

Step 1: By similar arguments to equation (\ref{eq:off_1_part1}), we have
\begin{equation} \label{eq:off_2_part1}
\begin{aligned}
&\quad \sum_{z=1,2,\ldots{},k_{off}} \ceil*{S(K_z,t)/g_z}\cdot r_z\\
&\leq{} 3\cdot{}\sum_{z\in{}T(k_{off})} \ceil*{S(R_z,t)/g_z}\cdot r_z \\
&\leq{} 3\cdot{} \sum_{z\in{}T(k_{off})\setminus{}\{k_{off}\}} \left(S(R_z,t)\cdot \frac{r_z}{g_z} + r_z\right) + 3\cdot
\ceil*{\frac{S(R_{k_{off}},t)}{g_{k_{off}}}}\cdot r_{k_{off}}\\
&\leq{} 3\cdot{} \sum_{z\in{}T(k_{off} )\setminus{}\{k_{off}\}} S(R_z,t)\cdot \frac{r_z}{g_z} + \frac{24}{7}\cdot{}
\ceil*{\frac{S(R_{k_{off}},t)}{g_{k_{off}}}}\cdot r_{k_{off}}.\\
\end{aligned}
\end{equation}

Step 2: Since $k_{off}>z^\diamond$ and both $k_{off}$ and $z^\diamond$ are in $P(k_0)$, $T(k_{off})\setminus{}\{k_{off}\}\subset T(z^\diamond)$ holds by Proposition \ref{pro:T_two}. We have
\begin{eqnarray}
& & \sum_{z\in{}T(k_{off})\setminus{}\{k_{off}\}} S(R_z,t)\cdot \frac{r_z}{g_z} \nonumber \\
&\leq{}& \sum_{z\in{}T(z^\diamond)} S(R_z,t)\cdot{}\frac{r_z}{g_z} \nonumber \\
&\leq{}& \sum_{z\in{}T(z^\diamond)} \cn_{z^\diamond} (z) \cdot{} r_z. \quad \text{(by~Proposition~\ref{pro:cn&r'_head})} \label{eq:off_2_part2}
\end{eqnarray}

Step 3: By claim (1) of Proposition \ref{pro:alternative} and the condition ${k_{off} > z^\diamond}$, $\sum_{z\in{}T(z^\diamond)\cap{}A_{k_{off}}} \cn_{z^\diamond}(z) r_z \leq{} r_{k_{off}}$.
By Propositions \ref{pro:T_two} and \ref{pro:cn&r'_head}, we have $\sum_{z\in{}T(z^\diamond)\cap{}A_{k_{off}}} S(R_z,t)\cdot{} \frac{r_z}{g_z} \leq{} \sum_{z\in{}T(z^\diamond)\cap{}A_{k_{off}}} \cn_{z^\diamond}(z) r_z$.
Furthermore, $S(R_{k_{off}},t)\cdot{}\frac{r_{k_{off}}}{g_{k_{off}}} < \sum_{z\in{}T(z^\diamond)\cap{}A_{k_{off}}} S(R_z,t)\cdot{} \frac{r_z}{g_z}$ because $R_{k_{off}}(t) = \cup_{z\in{}T(z^\diamond)\cap{}A_{k_{off}}} R_z (t)$ and $r_{k_{off}}/g_{k_{off}} < r_z/g_z$ for each $z\in{}T(z^\diamond)\cap{}A_{k_{off}}$.
After combining all the pieces stated above, we have $S(R_{k_{off}},t)\cdot{}\frac{r_{k_{off}}}{g_{k_{off}}} < r_{k_{off}}$. Therefore,
\begin{equation} \label{eq:off_2_part3}
S(R_{k_{off}},t)<g_{k_{off}}.
\end{equation}

Step 4: Now we give an upper bound to $r_{k_{off}}$. Since ${k_{off} > z^\diamond}$ and $z^\diamond \in P(k_0)$, it follows that $k_{off} > k_0$, which implies that $\{J\in{}R_{k_{off}}(t):m(J)=k_{off}\}=\emptyset{}$. By Property \ref{property:off_lowerb}, $\frac{4}{21}\cdot{}r_{k_{off}} \leq{} \sum_{x\in{}f(k_{off})} S(R_x,t)\cdot \frac{r_x}{g_x}$. Furthermore,
\begin{equation*}
\begin{aligned}
\sum_{x\in{}f(k_{off})} S(R_x,t)\cdot \frac{r_x}{g_x}
&\leq{} \sum_{z\in{}T(z^\diamond) \cap{} A_{k_{off}}} S(R_z,t)\cdot \frac{r_z}{g_z}\\
&\leq{} \sum_{z\in{}T(z^\diamond) \cap{} A_{k_{off}}} \cn_{z^\diamond}(z)\cdot r_z\\
&\leq{} \sum_{z\in T(z^\diamond)} \cn_{z^\diamond}(z)\cdot r_z,\\
\end{aligned}
\end{equation*}
where the first inequality is because for each $z\in{}T(v^\diamond)\cap{}A_{k_{off}}$, $z\in{}A_x$ and hence $R_z\subset{}R_x$ for some $x\in{}f(k_{off})$, and the second inequality is due to Proposition \ref{pro:cn&r'_head}. Therefore, we have
\begin{equation} \label{eq:off_2_part4}
r_{k_{off}}\leq{} \frac{21}{4}\cdot{} \sum_{z\in T(z^\diamond)} \cn_{z^\diamond}(z)\cdot r_z.
\end{equation}

In summary of the four steps above,
\begin{equation*}
\begin{aligned}
&\quad \sum_{z=1,2,\ldots{},k_{off}} \ceil*{S(K_z,t)/g_z}\cdot r_z\\
&\leq{} 3\cdot{} \sum_{z\in{}T(k_{off})\setminus{}\{k_{off}\}} S(R_z,t)\cdot \frac{r_z}{g_z} + \frac{24}{7} \cdot{} r_{k_{off}}\\
&\quad \text{~(by~equations~(\ref{eq:off_2_part1}),~(\ref{eq:off_2_part3}))}\\
&\leq{} 3\cdot{} \sum_{z\in T(z^\diamond)} \cn_{z^\diamond}(z)\cdot r_z + \frac{24}{7}\cdot{}\frac{21}{4}\cdot{} \sum_{z\in T(z^\diamond)} \cn_{z^\diamond}(z)\cdot r_z \\
&\quad \text{~(by~equations~(\ref{eq:off_2_part2}),~(\ref{eq:off_2_part4}))}\\
&= 21\cdot{} \sum_{z\in T(z^\diamond)} \cn_{z^\diamond}(z)\cdot r_z.\\
\end{aligned}
\end{equation*}
\end{proof}

By Theorem \ref{thm:cn_app}, we have the following corollary.

\begin{corollary} \label{cor:off_app}
$\sum_{z \in \mathcal{M}} \ceil{S(K_z,t)/g_z} r_z \leq{} 45 \cdot{}  \optone (\mathcal{J}(t))$.
\end{corollary}

It follows that the cost of $ALG_{offline}$ satisfies
\begin{equation*}
\begin{aligned}
ALG_{offline}(\mathcal{J}) &\leq{} \int_{t\in{}\sspan(\mathcal{J})} \left(\sum_{z \in \mathcal{M}}
4\cdot{}\ceil{S(K_z,t)/g_z} r_z \right) \,\mathrm{d}t\\
&\leq{} 180 \cdot{}\int_{t\in{}\sspan(\mathcal{J})} \optone(\mathcal{J}(t)) \,\mathrm{d}t\\
&\leq{} 180 \cdot{} \opttwo(\mathcal{J}). \quad \text{(by~equation~(\ref{eq:opt2toopt1}))}
\end{aligned}
\end{equation*}
Therefore, $ALG_{offline}$ is an $O(1)$-approximation algorithm.

\section{The Online Setting}

\subsection{The online algorithm $ALG_{online}$}

We now discuss the online BSHM problem. We say that a machine is {\em opened} when it receives the first job to process. When all the active jobs end on an open machine, the machine is {\em closed}. In the online setting, jobs are released when they are to start execution.
For simplicity, we assume that jobs are released one at a time.
Algorithm \ref{alg:online} shows our online algorithm $ALG_{online}$ for each new job $J$ released. The algorithm iteratively considers the exact machine type $m(J)$ and its ancestor types for processing $J$ (lines 1 and 9). When a machine type $z$ is considered, if there are one or more type-$z$ machines that are open and have available capacity to host job $J$, $J$ is scheduled onto the machine which was opened earliest among these machines (this is known as the First Fit rule) (lines 3-5). If not, we check whether a new type-$z$ machine should be opened. If opening a new type-$z$ machine does not cause the total cost of the open machines for the types in the tree rooted at each ancestor type $z_a$ (except type $z_a$ itself) to exceed that of one type-$z_a$ machine, a new type-$z$ machine is opened to host job $J$ (lines 6-8). Otherwise, we proceed to consider the parent type $p(z)$ unless $z$ has no parent in the cost-per-capacity graph (line 9).

\begin{algorithm} [t]
\KwIn{A new job $J$ released at time $I(J)^-$}
\KwOut{A machine for processing $J$}
$z\gets{} m(J)$\;
\While{true}
{
	\If{there exist type-$z$ machines open at time $I(J)^-$ with available capacity at least $s(J)$} 
	    {
        among these machines, \textbf{return} the machine which was opened earliest\;}
	\If{p(z) \textrm{does not exist or} $\forall{}z_a\in{}P(z)\setminus{}\{z\}$, $\sum_{x\in{}A_{z_a}\setminus{}\{z_a\}} n_x r_x < r_{z_a}-r_{z}$ where $n_x$ is the number of type-$x$ machines open at time $I(J)^-$}
	    {
	    open and \textbf{return} a \textrm{new} type-$z$ machine\;}
	$z\gets{}p(z)$\;
}
\caption{$ALG_{online}$}
\label{alg:online}
\end{algorithm}

We define some notations for $ALG_{online}$. For each machine type $z\in\mathcal{M}$, let $K_z$ denote the set of jobs scheduled onto machine type $z$. For each machine type $z\in\mathcal{M}$ and each time instant $t\in{}\sspan(\mathcal{J})$, let $N(z,t)$ denote the number of type-$z$ machines being open at time $t$. Let $k(t) := \max\{z:K_z(t)\neq{}\emptyset{}\}=\max \{z:N(z,t)>0\}$ denote the highest-indexed machine type used at time $t$.

To analyze the $ALG_{online}$ algorithm, we create a set of artificial jobs to fill up the unused machine capacities of open machines, in order to establish the relation between the cost of $ALG_{online}$ and the cost of the optimal machine configuration (see Section \ref{sec:artificial}). Recall from Section \ref{sec:first_app} that the alternative machine configuration is an $O(1)$ approximation of the optimal machine configuration. For each time instant, we also invent a mechanism to charge the cost of the alternative machine configuration onto individual jobs within an $O(1)$ factor (see Section \ref{sec:modifiedapprox}). This charging mechanism provides a nice ``monotonic'' property (adding new jobs can only increase the costs charged onto existing jobs, see Theorem \ref{thm:tilder_1}). Based on this property, we show that the cost due to the artificial jobs is bounded by a factor $O(\mu)$ of the cost due to the original jobs (Theorem \ref{thm:hard}), which leads to the $O(\mu)$ competitive ratio of the $ALG_{online}$ algorithm.

\subsection{A set of artificial jobs $\mathcal{R}$}
\label{sec:artificial}

We start by creating some artificial jobs to fill up the machine capacities of open machines by $ALG_{online}$. For each job $J\in{}\mathcal{J}$, we create three artificial jobs
$F_1 (J)$, $F_2 (J)$ and $F_3 (J)$. They have the same sizes as $J$, i.e., $s(F_1 (J))=s(F_2 (J))=s(F_3 (J))=s(J)$. Their active intervals are defined as follows: $I(F_1 (J))=I(J)$, i.e., $F_1 (J)$ has the same active interval as $J$; $I(F_2 (J))=[I(J)^+,I(J)^++\mu)\cap{}\sspan(\mathcal{J})$, i.e., $F_2 (J)$ extends $J$'s active interval by a period $\mu$; $I(F_3 (J))=[I(J)^+,I(J)^++2\mu)\cap{}\sspan(\mathcal{J})$, i.e., $F_3 (J)$ extends $J$'s active interval by a period $2\mu$. Let $\mathcal{R}=\{F_1 (J), F_2 (J), F_3 (J) : J \in \mathcal{J}\}$ denote all the artificial jobs created. In the following, we show that at each time instant $t$, the active jobs $\mathcal{R}(t)$ satisfies some properties.

Lemma \ref{thm:R_case2} says that for each machine type $z\in{}T(k(t))$, if there are multiple type-$z$ machines open at time $t$, the active jobs in these machines together with some artificial jobs active at time $t$
can fill up the capacities of these machines except one.

\begin{lemma} \label{thm:R_case2}
For each $z\in{}T(k(t))$ such that $N(z,t)>1$, we have $S(K_z,t) + S(\{J\in{}\mathcal{R}(t):m(J)\in{}A_z\}) > (N(z,t)-1) g_z$.
\end{lemma}
\begin{proof}

As illustrated in Figure \ref{fig:R_diagram_case2}, suppose $n = N(z,t)>1$ type-$z$ machines being open at time $t$ were opened in the order of $m_1$, $m_2$, $\ldots$, $m_n$. We pick an active job $J_i$ (black rectangle) in each machine $m_i$ at time $t$. For each $i>1$, when $J_i$ was scheduled onto $m_i$ at time $I(J_i)^-$, machine $m_{i-1}$ was also open at that time. Let $\mathcal{K}_{z}^{i-1}$ denote the set of active jobs in machine $m_{i-1}$ at time $I(J_i)^-$ (black rectangles). By the First Fit scheduling rule, we must have $s(J_i) + S(\mathcal{K}_{z}^{i-1}) > g_z$ for each $i=2,\ldots{},n$. Note that all the jobs $J_i$ and $\mathcal{K}_{z}^{i-1}$ have their exact machine types in the tree $A_z$. Each job $J_i$ has an artificial job $F_1(J_i)$ in $\mathcal{R}$ active at $t$ (grey rectangles). In addition, each job $J \in \mathcal{K}_{z}^{i-1}$ has an artificial job $F_2(J)$ in $\mathcal{R}$ which extends $J$ by a period $\mu$ (rectangles in back slash pattern). Either $J$ or $F_2(J)$ is active at $t$, since $I(J)^- \leq I(J_i)^- \leq t$ and $t - I(J)^+ < \len(J_i) \leq \mu$. Therefore, the total size of the active jobs in $K_z \cup \mathcal{R}$ at time $t$ (having their exact machine types in $A_z$) is at least $\sum_{i=2}^n (s(J_i) + S(\mathcal{K}_{z}^{i-1})) > (n-1) g_z$.
\end{proof}

\begin{figure}[t]
  \centering
  \includegraphics[width=7.8cm]{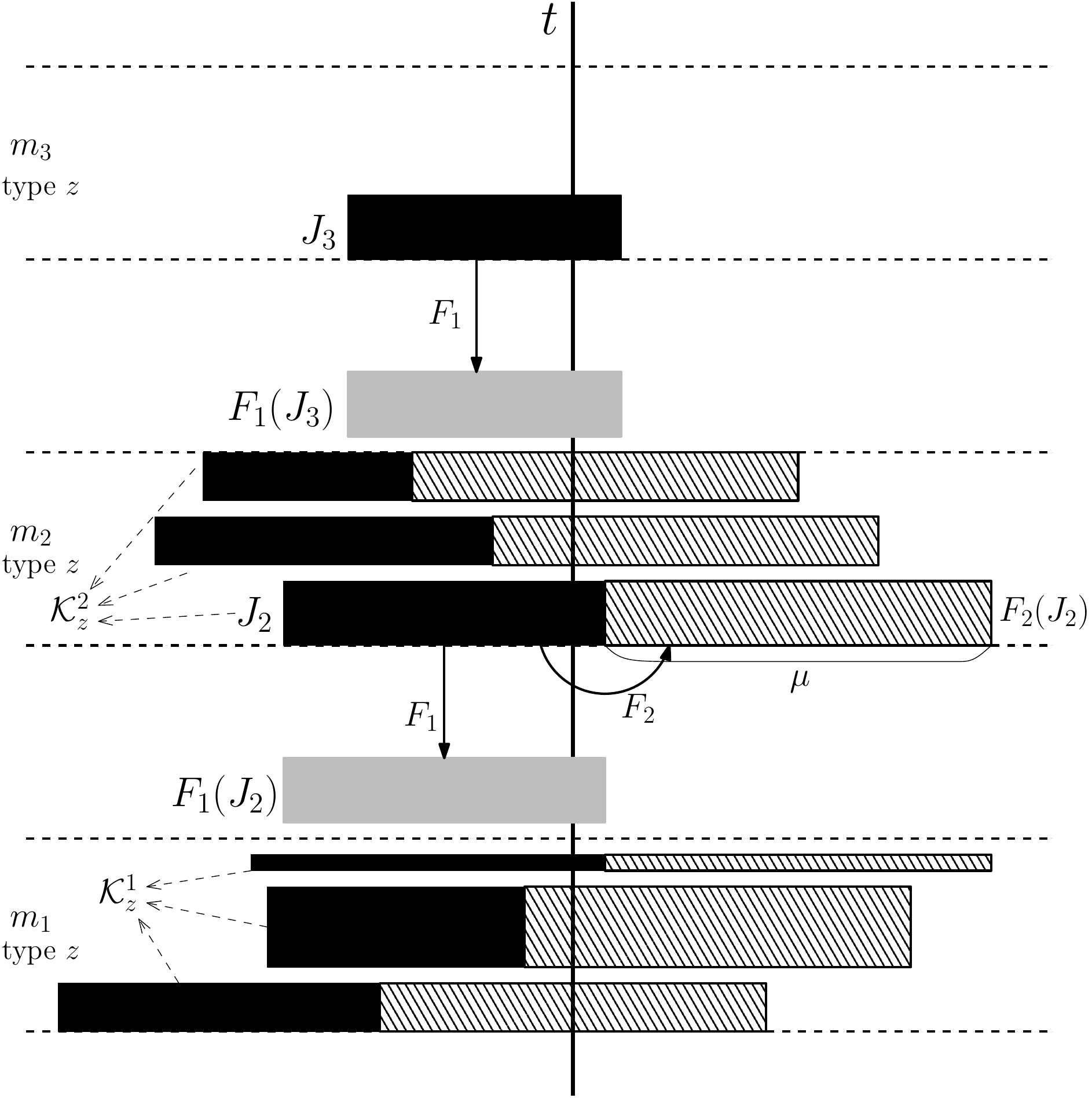}
  \caption{An example of Lemma \ref{thm:R_case2}}
  \label{fig:R_diagram_case2}
\end{figure}

Lemma \ref{thm:R_case1} says that if there is only one $k(t)$-type machine open at time $t$ and all the active jobs in this machine can be placed in some lower-indexed machine type than $k(t)$, take any active job $\hat{J}$ in this machine, then for each $k(t)$'s child type $z$ with multiple type-$z$ machines open at time $I(\hat{J})^-$, the active jobs in these type-$z$ machines together with some artificial jobs active at time $t$
can fill up the capacities of these type-$z$ machines except one.

\begin{figure}[h]
  \centering
  \includegraphics[width=7.8cm]{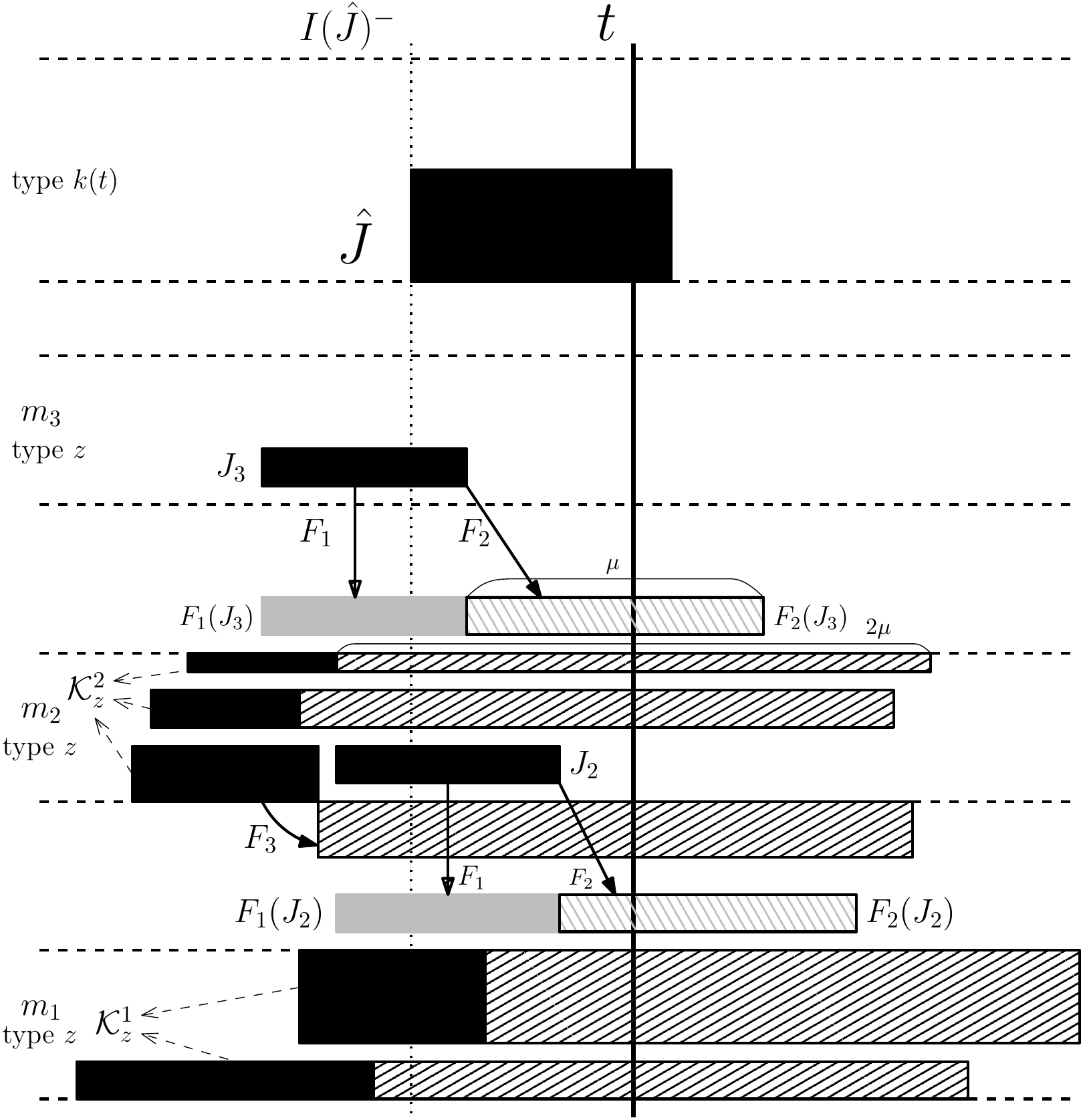}
  \caption{An example of Lemma \ref{thm:R_case1}}
  \label{fig:R_diagram_case1}
\end{figure}

\begin{lemma} \label{thm:R_case1}
If $N(k(t),t)=1$ and $m(J)<k(t)$ for each $J\in{}K_{k(t)} (t)$, taking any job $\hat{J} \in{} K_{k(t)} (t)$, we have $S(K_z,t) + S(\{J\in{}\mathcal{R}(t):m(J)\in{}A_z\}) > (N(z,I(\hat{J} )^- ) - 1) g_z$ for each $z\in{}f(k(t))$ such that $N(z,I(\hat{J})^-)>1$.
\end{lemma}
\begin{proof}

The proof is similar to Lemma \ref{thm:R_case2}. As illustrated in Figure \ref{fig:R_diagram_case1}, take any active job $\hat{J}$ (black rectangle) in the only open type-$k(t)$ machine at time $t$, and we consider all the type-$z$ machines being open at time $I(\hat{J})^-$. Suppose $n = N(z,I(\hat{J})^-)>1$ type-$z$ machines being open at time $I(\hat{J})^-$ were opened in the order of $m_1$, $m_2$, $\ldots$, $m_n$. We pick an active job $J_i$ (black rectangle) in each machine $m_i$ at time $I(\hat{J})^-$. For each $i>1$, when $J_i$ was scheduled into $m_i$ at time $I(J_i)^-$, machine $m_{i-1}$ was also open at that time. Let $\mathcal{K}_{z}^{i-1}$ denote the set of active jobs in machine $m_{i-1}$ at time $I(J_i)^-$ (black rectangles). By the First Fit scheduling rule, we must have $s(J_i) + S(\mathcal{K}_{z}^{i-1}) > g_z$ for each $i=2,\ldots{},n$. Note that all the jobs $J_i$ and $\mathcal{K}_{z}^{i-1}$ have their exact machine types in the tree $A_z$. Each job $J_i$ has an artificial job $F_1(J_i)$ in $\mathcal{R}$ (grey rectangles) and an artificial job $F_2(J_i)$ in $\mathcal{R}$ which extends $J_i$ by a period $\mu$ (rectangles in back slash pattern). One of these two artificial jobs must be active at time $t$, since $I(J_i)^- \leq I(\hat{J})^- \leq t$ and $t - I(J_i)^+ < \len(\hat{J}) \leq \mu$. In addition, each job $J' \in \mathcal{K}_{z}^{i-1}$ has an artificial job $F_3(J')$ in $\mathcal{R}$ which extends $J'$ by a period $2\mu$ (rectangles in slash pattern). Either $J'$ or $F_3(J')$ is active at time $t$, since $I(J')^- \leq I(J_i)^- \leq t$ and $t - I(J')^+ < \len(J_i) + \len(\hat{J}) \leq 2\mu$. Therefore, the total size of the active jobs in $K_z \cup \mathcal{R}$ at time $t$ (having their exact machine types in $A_z$) is at least $\sum_{i=2}^n (s(J_i) + S(\mathcal{K}_{z}^{i-1})) > (n-1) g_z$.
\end{proof}

Lemmas \ref{thm:R_case2} and \ref{thm:R_case1} indicate that the artificial jobs $\mathcal{R}$ can fill up the capacity of (most) machines used by $ALG_{online}$ at any time $t$. Based on this fact,
we can prove that the total cost of the machines used by $ALG_{online}$ at any time $t$ is bounded by $O(1)$ times the cost of optimal one-shot scheduling for the active jobs in $\mathcal{J}\cup{}\mathcal{R}$ at $t$.

\begin{theorem} \label{thm:N_leq_opt}
At each time instant $t$, we have  $\sum_{z=1,\ldots{},m} N(z,t) r_z \leq{} 5\cdot{} \optone(\mathcal{J}(t)\cup{}\mathcal{R}(t))$.
\end{theorem}

\subsection{A sufficient condition}

Let $\mathcal{F}_i=\{F_i (J) : J \in \mathcal{J}\}$ for each $i = 1, 2, 3$. By definition, $\mathcal{J}\cup{}\mathcal{R}=(\mathcal{J}\cup{}\mathcal{F}_3) \cup{} (\mathcal{F}_1 \cup{} \mathcal{F}_2)$. For any time instant $t$, the combination of any optimal machine configuration for the active jobs in $\mathcal{J}\cup{}\mathcal{F}_3$ and any optimal machine configuration for the active jobs in $\mathcal{F}_1  \cup{} \mathcal{F}_2$ is a feasible machine configuration for the active jobs in $\mathcal{J}\cup{}\mathcal{R}$. By optimality,
\begin{eqnarray}
& & \optone (\mathcal{J}(t)\cup{}\mathcal{R}(t)) \nonumber \\
&\leq{}& \optone (\mathcal{J}(t)\cup{}\mathcal{F}_3(t))+\optone (\mathcal{F}_1(t)\cup{}\mathcal{F}_2(t)). \label{eq:split}
\end{eqnarray}
Actually, in order to prove that $ALG_{online}$ is an $O(\mu)$-competitive algorithm, it suffices to show the following theorem.
Define a function $F_d$ with $d\geq{}\mu$ which maps each job $J$ in $\mathcal{J}$ to a new job $F_d (J)$ defined as: $s(F_d (J))=s(J)$ and $I(F_d (J))=[I(J)^+,I(J)^++d)\cap{}\sspan(\mathcal{J})$, i.e., the new job $F_d (J)$ has the same size as $J$ and extends $J$'s active interval by a period $d$.

\begin{theorem} \label{thm:hard}
Let $\mathcal{H}=\{F_d (J) : J \in \mathcal{J}\}$ with $d\geq{}\mu$. We have
$\int_{t\in{}\sspan(\mathcal{J})} \optone (\mathcal{J}(t)\cup{}\mathcal{H}(t))\,\mathrm{d}t \leq{} O(d)\cdot{}\int_{t\in{}\sspan(\mathcal{J})} \optone(\mathcal{J}(t))\,\mathrm{d}t$.
\end{theorem}

Applying Theorem \ref{thm:hard} by letting $d=\mu$ and $2\mu$, the cost of $ALG_{online}$ satisfies
\begin{eqnarray}
\lefteqn{ALG_{online}(\mathcal{J}) = \int_{t\in{}\sspan(\mathcal{J})} \sum_{i\in \mathcal{M}} N(i,t)\cdot r_i \,\mathrm{d}t} \nonumber \\
&\leq{}& O(1)\cdot{} \int_{t\in\sspan(\mathcal{J})} \optone (\mathcal{J}(t)\cup{}\mathcal{R}(t)) \,\mathrm{d}t \quad \text{(by~Theorem~\ref{thm:N_leq_opt})} \nonumber \\
&\leq{}& O(1)\cdot{} \int_{t\in\sspan(\mathcal{J})} \Big(\optone (\mathcal{J}(t)\cup{}\mathcal{F}_3(t)) \nonumber \\
& &\quad \quad + \optone (\mathcal{F}_1(t) \cup{} \mathcal{F}_2(t))\Big) \,\mathrm{d}t \quad \text{(by~equation~(\ref{eq:split}))} \nonumber \\
&\leq{}& O(\mu)\cdot{}\int_{t\in\sspan(\mathcal{J})} \optone (\mathcal{J}(t)) \,\mathrm{d}t \quad \text{(by~Theorem~\ref{thm:hard})} \nonumber \\
&\leq{}& O(\mu)\cdot{} \opttwo (\mathcal{J}). \quad \text{(by~equation~(\ref{eq:opt2toopt1}))} \nonumber
\end{eqnarray}

\subsection{A modified $O(1)$ approximation of optimal one-shot scheduling}
\label{sec:modifiedapprox}

Recall that in Section \ref{sec:first_app}, we defined an alternative machine configuration $\cn_{z^\diamond}$ for one-shot scheduling of a set of jobs $\mathcal{J}^{1d}$, where the highest-indexed machine type $z^\diamond$ used is derived from $\mathcal{J}^{1d}$. In order to prove Theorem \ref{thm:hard}, we shall charge the machine cost of a machine configuration onto individual jobs in $\mathcal{J}^{1d}$ and have a desired ``monotonic'' property that the cost charged on each job is non-increasing as the job set $\mathcal{J}^{1d}$ expands (and hence the machine configuration changes with $\mathcal{J}^{1d}$). To achieve this ``monotonic'' property, the alternative machine configuration $\cn_{z^\diamond}$ defined in Section \ref{sec:first_app} is not adequate. We modify it as follows.

Given a set of jobs $\mathcal{J}^{1d}$, the modified machine configuration
uses the same highest-indexed machine type $z^\diamond$ as the alternative machine configuration $\cn_{z^\diamond}$. For each machine type $i\in \mathcal{M}$, let $H_i = \{J\in\mathcal{J}^{1d}:m(J)\in{}A_i \}$ denote the set of jobs whose exact machine types are in the tree rooted at type $i$. Then, $\{H_i: i \in T(z^\diamond)\}$ is a partitioning of $\mathcal{J}^{1d}$.

For each $i\in{}T(z^\diamond)\setminus{}\{z^\diamond\}$, the jobs in $H_i$ are always accommodated by type-$i$ machines in the modified machine configuration. Hence, we need
$\frac{S(H_i)}{g_i}$ type-$i$ machines with a total cost of $S(H_i)\frac{r_i}{g_i}$. Each job is charged a cost proportional to each size, i.e., each job $J \in H_i$ is charged a cost of $\tilde{r}(\mathcal{J}^{1d})(J) = s(J) \frac{r_i}{g_i}$, so that their total cost matches $S(H_i)\frac{r_i}{g_i}$. Note that we include the job set $\mathcal{J}^{1d}$ in the notation $\tilde{r}(\mathcal{J}^{1d})(J)$ to indicate that the machine configuration and hence the cost charged to each job are dependent on $\mathcal{J}^{1d}$.

For the jobs in $H_{z^\diamond}$, if their total size at least the capacity of one type-$z^\diamond$ machine, i.e., $S(H_{z^\diamond}) \geq g_{z^\diamond}$, all of them are accommodated by type-$z^\diamond$ machines in the modified machine configuration. Hence, we need
$\frac{S(H_{z^\diamond})}{g_{z^\diamond}}$ type-$z^\diamond$ machines with a total cost of $S(H_{z^\diamond})\frac{r_{z^\diamond}}{g_{z^\diamond}}$. Again, each job is charged a cost proportional to each size, i.e., each job $J \in H_{z^\diamond}$ is charged a cost of $\tilde{r}(\mathcal{J}^{1d})(J) = s(J) \frac{r_{z^\diamond}}{g_{z^\diamond}}$, so that their total cost matches $S(H_{z^\diamond})\frac{r_{z^\diamond}}{g_{z^\diamond}}$.

If $S(H_{z^\diamond}) < g_{z^\diamond}$, we aim to use one type-$z^\diamond$ machine to accommodate all the jobs in $H_{z^\diamond}$
with a cost of $r_{z^\diamond}$. The cost is charged onto the jobs in $H_{z^\diamond}$ as follows. Note that
the jobs $H_{z^\diamond}$ can be further partitioned into $H_{z^\diamond}^h := \{J\in\mathcal{J}^{1d}:m(J) = z^\diamond\}$ and $\{H_i: i \in f(z^\diamond)\}$ where $f(z^\diamond)$ is the set of $z^\diamond$'s child types in the cost-per-capacity graph. Let $c:= S(H_{z^\diamond}^h)\frac{r_{z^\diamond}}{g_{z^\diamond}} + \sum_{i\in{}f(z^\diamond)} S(H_i)\frac{r_i}{g_i}$ be the cost of using type-$z^\diamond$ machines to accommodate $H_{z^\diamond}^h$ and type-$i$ machines to accommodate each $H_i$ $(i \in f(z^\diamond))$. If $c > r_{z^\diamond}$, each job $J \in H_{z^\diamond}^h$ is charged a cost of $\tilde{r}(\mathcal{J}^{1d})(J) = s(J) \frac{r_{z^\diamond}}{g_{z^\diamond}}$, and each job $J \in H_i$ is charged a cost of $\tilde{r}(\mathcal{J}^{1d})(J) = s(J) (\frac{r_i}{g_i} - (\frac{r_i}{g_i} - \frac{r_{z^\diamond}}{g_{z^\diamond}}) \cdot{} \alpha^*)$, where $\alpha^*\in{}\left(0,1\right)$ is given by $\alpha^* = (c - r_{z^\diamond}) / \sum_{i \in f(z^\diamond)}\sum_{J \in H_i} s(J)(\frac{r_i}{g_i} - \frac{r_{z^\diamond}}{g_{z^\diamond}})$ to ensure that the total cost charged is $r_{z^\diamond}$. Since $\alpha^*\in{}\left(0,1\right)$, we have $\tilde{r}(\mathcal{J}^{1d})(J)\in{}\left(\frac{r_{z^\diamond}}{g_{z^\diamond}},\frac{r_i}{g_i}\right)$ for each $J \in H_i$.
Otherwise, if $c \leq r_{z^\diamond}$, each job $J \in H_{z^\diamond}^h$ is charged a cost of $\tilde{r}(\mathcal{J}^{1d})(J) = s(J) \frac{r_{z^\diamond }}{g_{z^\diamond}}\cdot{}(1+\beta^*)$ if $H_{z^\diamond}^h \neq \emptyset$, and each job $J \in H_i$ is charged a cost of $\tilde{r}(\mathcal{J}^{1d})(J) = s(J) \frac{r_i}{g_i}$, where $\beta^* \geq 0$ is given by $\beta^* = (r_{z^\diamond} - c) / \sum_{J\in H_{z^\diamond}^h}s(J) \frac{r_{z^\diamond}}{g_{z^\diamond}}$ to ensure that the total cost charged is $r_{z^\diamond}$.

With the costs charged on individual jobs, the total cost of the modified machine configuration is given by $\sum_{J\in{}\mathcal{J}^{1d}} \tilde{r}(\mathcal{J}^{1d})(J)$.
The charging is deliberately designed in the above way to achieve the ``monotonic'' property below. Note that a major challenge in guaranteeing the ``monotonic'' property is that the highest-indexed machine type $z^\diamond$ used is derived from $\mathcal{J}^{1d}$ and it may change as $\mathcal{J}^{1d}$ expands.

\begin{theorem} \label{thm:tilder_1}
For any two sets of jobs $\mathcal{X}\subset{}\mathcal{Y}$, we have $\tilde{r}(\mathcal{X})(J) \geq{} \tilde{r}(\mathcal{Y})(J)$, for each job $J\in{}\mathcal{X}$.
\end{theorem}

In addition, the modified machine configuration is an $O(1)$ approximation of the alternative machine configuration. Hence, it remains an $O(1)$ approximation of optimal one-shot scheduling.

\begin{theorem} \label{thm:tilder_2}
$\frac{1}{2}\cdot{}\sum_{J\in{}\mathcal{J}^{1d}} \tilde{r}(\mathcal{J}^{1d})(J) \leq{} \sum_{z\in T(z^\diamond)} \cn_{z^\diamond} (z) r_z \leq{} \frac{15}{7}\cdot{} \sum_{J\in{}\mathcal{J}^{1d}} \tilde{r}(\mathcal{J}^{1d})(J)$.
\end{theorem}

\subsection{Proof of Theorem \ref{thm:hard}}
Now, we are ready to finish the proof of Theorem \ref{thm:hard} which is the last piece to be discussed.

\begin{proof}

By Theorem \ref{thm:cn_app} and Theorem \ref{thm:tilder_2}, it suffices to show that \begin{eqnarray}
& & \int_{t\in\sspan(\mathcal{J})} \sum_{J\in{}\mathcal{J}(t)\cup{}\mathcal{H}(t)} \tilde{r}(\mathcal{J}(t)\cup{}\mathcal{H}(t))(J) \,\mathrm{d}t \nonumber \\
&\leq{}& O(d)\cdot{}\int_{t\in\sspan(\mathcal{J})} \sum_{J\in{}\mathcal{J}(t)}  \tilde{r}(\mathcal{J}(t))(J) \,\mathrm{d}t. \label{eq:tilder3}
\end{eqnarray}

For each job $J\in{}\mathcal{J}$ and each $t\in{}I(J)$, define $\tilde{r}_1 (J,t)=\tilde{r}(\mathcal{J}(t))(J)$.
For each $t\in\sspan(\mathcal{J})$, let $\mathcal{G}_{t}:=\{J\in\mathcal{J}:t-d<I(J)^+ \leq t\}$ be the set of jobs ending in the period $(t-d,t]$. By definition, $\mathcal{H}(t) = \{F_d (J) : J \in \mathcal{G}_{t}\}$, i.e., $\mathcal{G}_{t}$ is exactly all the jobs whose $d$ extension (to right) covers time $t$.
For each job $J\in\mathcal{J}$, define
$\tilde{r}_2 (J,t):=\tilde{r}_1 \left(J,\frac{t-I(J)^+}{d} \len(J) + I(J)^- \right)$ for each $t\in{}[I(J)^+,I(J)^++d)$. We have
\begin{eqnarray}
\lefteqn{\int_{t\in\sspan(\mathcal{J})} \left(\sum_{J\in\mathcal{G}_t} \tilde{r}_2 (J,t)\right) \,\mathrm{d}t} \nonumber \\
&= & \sum_{J\in\mathcal{J}} \left(\int_{t\in[I(J)^+,I(J)^++d)\cap{}\sspan(\mathcal{J})} \tilde{r}_2(J,t) \,\mathrm{d}t \right) \nonumber \\
&\leq{}& \sum_{J\in\mathcal{J}} \left(\int_{t\in[I(J)^+,I(J)^++d)} \tilde{r}_2 (J,t) \,\mathrm{d}t \right) \nonumber \\
&= & \sum_{J\in\mathcal{J}} \left(\int_{t\in[I(J)^+,I(J)^++d)} \tilde{r}_1 (J,\frac{t-I(J)^+}{d} \len(J) + I(J)^- ) \,\mathrm{d}t \right) \nonumber \\
&= & \sum_{J\in\mathcal{J}} \left(\int_{\tau\in{}I(J)} \tilde{r}_1(J,\tau) \, \frac{\mathrm{d}t}{\mathrm{d}\tau} \,\mathrm{d}\tau \right) \quad \text{(where~} \tau= \frac{t-I(J)^+}{d} \len(J)+I(J)^-\text{)} \nonumber \\
&= & \sum_{J\in{}\mathcal{J}} \left(\int_{\tau\in{}I(J)} \tilde{r}_1 (J,\tau) \frac{d}{\len(J)} \,\mathrm{d}\tau \right) \nonumber \\
&\leq{}& d\cdot{} \sum_{J\in\mathcal{J}} \left(\int_{\tau\in{}I(J)} \tilde{r}_1 (J,\tau) \,\mathrm{d}\tau \right) \quad \text{(since~} 1 \leq \len(J) \leq \mu \leq d\text{)} \nonumber \\
&= & d\cdot{} \int_{t\in\sspan(\mathcal{J})} \sum_{J\in\mathcal{J}(t)} \tilde{r}(\mathcal{J}(t))(J) \,\mathrm{d}t, \label{eq:tilder2}
\end{eqnarray}
where we have used a swap of two summations in the first equality and the last equality.

\textbf{Claim~1}: For each $t\in[t_0-d,t_0)$, we have $\mathcal{J}(t)\subset{}\mathcal{J}(t_0)\cup{}\mathcal{G}_{t_0}$.

\textbf{Proof~of~Claim~1}.
Each job $J\in \mathcal{J}(t)$ is active at $t$. If $J$ is also active at $t_0$, then $J\in \mathcal{J}(t_0)$. If $J$ is not active at $t_0$, it must end in the period $(t,t_0]$, i.e., $t < I(J)^+ \leq t_0$. It follows from $t\in[t_0-d,t_0)$ that $t_0-d < I(J)^+ \leq t_0$. Hence, $J\in \mathcal{G}_{t_0}$.
\textbf{End~of~Claim~1}

\textbf{Claim~2}: For each $t_0\in\sspan(\mathcal{J})$, we have $\sum_{J\in\mathcal{J}(t_0)} \tilde{r}(\mathcal{J}(t_0))(J) + \sum_{J\in\mathcal{G}_{t_0}} \tilde{r}_2 (J,t_0) \geq{} \sum_{J\in{} \mathcal{J}(t_0) \cup \mathcal{H}(t_0)} \tilde{r}(\mathcal{J}(t_0) \cup \mathcal{H}(t_0))(J)$.

\textbf{Proof~of~Claim~2}. By Theorem \ref{thm:tilder_1}, $\mathcal{J}(t_0)\subset{}\mathcal{J}(t_0) \cup \mathcal{G}_{t_0}$ implies that for each job $J\in{}\mathcal{J}(t_0)$, $\tilde{r}(\mathcal{J}(t_0))(J) \geq{} \tilde{r}(\mathcal{J}(t_0) \cup \mathcal{G}_{t_0})(J)$. Recall that by definition, $\mathcal{H}(t_0) = \{F_d (J) : J \in \mathcal{G}_{t_0}\}$, which implies that $F_d$ is actually a 1-1 correspondence between $\mathcal{G}_{t_0}$ and $\mathcal{H}(t_0)$ such that $s(J)=s(F_d(J))$ for each job $J\in \mathcal{G}_{t_0}$. Therefore,
\begin{equation*}
\tilde{r}(\mathcal{J}(t_0) \cup \mathcal{H}(t_0)) (J) = \tilde{r}(\mathcal{J}(t_0) \cup \mathcal{G}_{t_0})(J) \text{~for~each~}J\in\mathcal{J}(t_0),
\end{equation*}
and
\begin{equation*}
\tilde{r}(\mathcal{J}(t_0) \cup \mathcal{H}(t_0))(F_d(J)) = \tilde{r}(\mathcal{J}(t_0) \cup \mathcal{G}_{t_0})(J) \text{~for~each~}J\in\mathcal{G}_{t_0}.
\end{equation*}

Thus, it remains and suffices to show that for each job $J\in\mathcal{G}_{t_0}$, $\tilde{r}_2 (J,t_0) \geq{} \tilde{r}(\mathcal{J}(t_0) \cup \mathcal{G}_{t_0})(J)$.

Take any job $J\in{}\mathcal{G}_{t_0}$. Let $t_1:=\frac{t_0-I(J)^+}{d} \len(J)+I(J)^- $.
We are to show that $t_0-d \leq{} t_1< t_0$.
The left inequality is equivalent to $\frac{t_0-d-I(J)^-}{\len(J)} \leq{} \frac{t_0-I(J)^+}{d} = \frac{t_0-d-I(J)^- + d-\len(J)}{\len(J)+d-\len(J)}$.
Since $d-\len(J) \geq d-\mu \geq 0$, it suffices to show that $\frac{t_0-d-I(J)^-}{\len(J)} \leq{} 1$.
This is easy to get: $J\in{}\mathcal{G}_{t_0} \Rightarrow{} t_0-d < I(J)^+ \leq t_0 \Rightarrow{} t_0-d < I(J)^-+\len(J) \Rightarrow{} t_0-d-I(J)^- < \len(J)$.
For the right inequality, $t_0-d < I(J)^+$ implies that $\frac{t_0-I(J)^+}{d} \len(J)+ I(J)^- < \len(J)+I(J)^- = I(J)^+ \leq{} t_0$.
Therefore, $t_1\in{} [t_0-d,t_0)$.

By Claim 1 and Theorem \ref{thm:tilder_1}, $\mathcal{J}(t_1)\subset{}\mathcal{J}(t_0) \cup \mathcal{G}_{t_0}$ implies that $\tilde{r}_2 (J,t_0) = \tilde{r}_1 (J,t_1 ) = \tilde{r}(\mathcal{J}(t_1))(J) \geq{} \tilde{r}(\mathcal{J}(t_0) \cup \mathcal{G}_{t_0})(J)$. \textbf{End~of~Claim~2}

Eventually,
\begin{equation*}
\begin{aligned}
&\quad  \int_{t\in{}\sspan(\mathcal{J})} \sum_{J\in{} \mathcal{J}(t) \cup \mathcal{H}(t)} \tilde{r} (\mathcal{J}(t) \cup \mathcal{H}(t))(J) \,\mathrm{d}t\\
&\leq{} \int_{t\in{}\sspan(\mathcal{J})} \left(\sum_{J\in\mathcal{J}(t)} \tilde{r} (\mathcal{J}(t))(J) + \sum_{J\in{}\mathcal{G}_t} \tilde{r}_2 (J,t)\right) \,\mathrm{d}t\\
&\leq{} (d+1)\cdot{} \int_{t\in{}\sspan(\mathcal{J})} \sum_{J\in{}\mathcal{J}(t)} \tilde{r} (\mathcal{J}(t))(J) \,\mathrm{d}t,\\
\end{aligned}
\end{equation*}
where the first inequality is due to Claim 2, and the second is by equation (\ref{eq:tilder2}). Hence, equation (\ref{eq:tilder3}) is proven.
\end{proof}

\section*{Acknowledgments}
This work is supported by Singapore Ministry of Education Academic Research Fund Tier 1 under Grant 2019-T1-002-042.

\newpage
\appendix

{\noindent \LARGE \bf APPENDICES}

\section{Preliminaries}

\subsection{Cost-per-capacity graph} \label{sec:pros_graph}

\begin{proof} [Proof of proposition \ref{pro:forest}]
By definition,  $r_{p(i)}/g_{p(i)} < r_i/g_i$ for each pair of nodes $i$ and $p(i)$. The whole graph is clearly acyclic. This implies that for any two nodes in the same component, there is an unique path connecting them. Take $C[i]$ the component containing node $i$. We want to show that $C[i]$ is a rooted tree. For node $i$, by the finiteness of the graph, the node $i^*:=\max P(i)$. We want to show that $C[i]$ is a tree rooted at $i^*$. Take any node $z\in{}C[i]$, since $z$ and $i^*$ are in the same component, suppose $(z,z_1,z_2,\ldots{},z_l,i^*)$ is the unique path linking $z$ and $i^*$. Since $p(i^*)$ does not exist, $z_l$ must be directed to $i^*$, i.e., $p(z_l )=i^*$. Consequently, $z_{l-1}$ must be directed to $z_l$ because of the uniqueness of $p(z_l)$. After repeating this argument along the path from $i^*$ to $z$, finally we will have $z$ is directed to $z_1$, $z_1$ is directed to $z_2$, $\ldots{}$, $z_{l-1}$ is directed to $z_l$ and $z_l$ is directed to $i^*$. Clearly this means $i^*\in P(z)$. Since $z$ is taken arbitrarily from $C[i]$, $C[i]$ is indeed a tree rooted at node $i^*$. Therefore, the whole graph is a forest.
\end{proof}

\begin{proof} [Proof of proposition \ref{pro:consecutive}]

Step 1: Consider the set of nodes $\{z: P(z)=\emptyset{}\}:=\{i_1,i_2,\ldots{},i_n\}$ with $i_1<i_2<\ldots{}<i_n$. We show that $A_{i_q}$ is consecutive for each $q=1,2,\ldots{},n$.

Observe that $i_q=\max A_{i_q}$ for each $q=1,2,\ldots,n$. By definition, $v(i_q)=\min A_{i_q}$. By proposition \ref{pro:forest}, $\{A_{i_q}: q=1,2,\ldots,n\}$ is a partitioning of $\mathcal{M}:=\{1,2,\ldots,|\mathcal{M}|\}$ which is consecutive by definition.
If we can show that for each $2\leq{}q\leq{}n$, $v(i_q) > i_{q-1}$, then we have proven that $A_{i_q}$ is consecutive for each $q$.
Suppose the contrary, i.e. there exists some $q\in\{2,\ldots{},n\}$ such that $i_{q-1} > v(i_q)$ which implies that $v(i_1) < i_q$.
We have $v(i_q) < i_{q-1} < i_q$ but $r_{v(i_q)} /g_{v(i_q)} > r_{i_q}/g_{i_q} \geq r_{i_{q-1}}/g_{i_{q-1}}$, where $r_{v(i_q)} /g_{v(i_q)} > r_{i_q}/g_{i_q}$ is because $i_q\in P(v(i_q))\setminus \{v(i_q)\}$ and $r_{i_q}/g_{i_q} \geq r_{i_{q-1}}/g_{i_{q-1}}$ is because otherwise $P(i_{q-1})$ must not be empty.
This situation contradicts to the fact that $i_q$ is an ancestor of $v(i_q)$.

Step 2: Suppose that $h$ is any non-negative integer. Assume that for each $z\in \{z: |P(z)|=h\}$, $A_z$ is consecutive. We show that for each $z\in \{z: |P(z)|=h+1\}$, $A_z$ is consecutive.

Take any $i^*\in \{z: |P(z)|=h+1\}$. By the given assumption, $A_{p(i^*)}$ is consecutive. Suppose that $f(p(i^*))=\{z_1,z_2,\ldots,z_x=i^*,\ldots,z_a\}$ with $z_1<z_2<\ldots{}<z_a$. To show that $A_{i^*}$ is consecutive, similarly as step 1, we show that for each $2\leq{}q\leq{}a$, $z_{q-1} < v(i_q)$. Again, similarly as step 1, suppose the contrary, i.e., there exists some $q\in\{2,\ldots{},a\}$ such that $z_{q-1} > v(z_q)$. We have $v(z_q) < z_{q-1} < z_q$ but $r_{v(z_q)} / g_{v(z_q)}  > r_{z_q}/g_{z_q} \geq{} r_{z_{q-1}}/g_{z_{q-1}}$. This contradicts to the fact that $z_q$ is an ancestor of $v(z_q)$. Therefore, we have shown that $A_{i^*}$ is consecutive.

Step 3: By the finiteness of graph, the combination of the results of step1 and step2 have proved that, $A_k$ must be consecutive, for each $k\in \mathcal{M}$ eventually.
\end{proof}

\begin{lemma} \label{lem:maxT}
For any $k\in \mathcal{M}$, $k=\max T(k)$.
\end{lemma}

\begin{proof} [Proof of lemma \ref{lem:maxT}]
Take any $i\in{}T(k)\setminus{}\{k\}$. By definition, $T(k)=\{k\}\cup\left(\cup_{z\in P(k)} y(z) \right)$. Therefore, $i\in{}y(z)$ for some $z\in P(k)$. Proposition \ref{pro:consecutive} says that $i< v(z) \leq k$ since $k\in A_z$.
\end{proof}

\begin{proof} [Proof of proposition \ref{pro:T_ratio}]
We separate the discussion into two cases.

Case 1: $k\notin\{y_1,y_2\}$.

Since $y_1,y_2\in{}T(k)=\{k\}\cup \left(\cup_{z\in P(k)} y(z) \right)$ and $k\notin\{y_1,y_2\}$,  we have $y_1\in y(z_1), y_2\in y(z_2)$, where $z_1,z_2\in P(k)$. We show that $z_1 \geq z_2$ must hold. Otherwise, i.e., if $z_1 < z_2$, we have $y_1\in y(z_1)\subset A_{z_2}$ and $y_2\in y(z_2)$. By proposition \ref{pro:consecutive}, $y_2 < y_1$ which contradicts to the condition $y_1\leq{}y_2$. Therefore, $z_1 \geq z_2$. Immediately, $y_2\in A_{p(y_1)}\setminus\{p(y_1)\}$.

Next, we show that $r_{y_1}/g_{y_1} > r_{y_2}/g_{y_2}$ cannot hold. Otherwise, i.e., $r_{y_1}/g_{y_1} > r_{y_2}/g_{y_2}$ and $y_1 < y_2<p(y_1)$, where $y_2 < p(y_1)$ is because $y_2\in A_{p(y_1)}\setminus\{p(y_1)\}$. In this case, by the definition of the cost per capacity graph, $y_1$ would not be directed to $p(y_1)$. Therefore, we have shown that $r_{y_1}/g_{y_1} \leq{} r_{y_2}/g_{y_2}$.

Case 2: one of $y_1$ and $y_2$ is $k$.

$y_2$ must be $k$ by lemma \ref{lem:maxT}, and it suffices to assume that $y_1\in y(z)$ for some $z\in P(k)$. By the same argument in case 1, $r_{y_1}/g_{y_1} > r_{y_2}/g_{y_2}$ cannot hold, because otherwise $y_1$ would not be directed to $p(y_1)$.

\end{proof}

\begin{lemma} \label{lem:V}
Take any $k\in \mathcal{M}$ and let $V_{k}:=\left(\cup_{z\in P(k)} y(z) \cup e(z) \right) \cup \{k\}$. We have $\left( P(k)\setminus \{k\} \right)  \cup \left(\cup{}_{z\in{}V_{k}} A_z \right)= \mathcal{M}$.
\end{lemma}

\begin{proof} [Proof of lemma \ref{lem:V}]

Let $k_0$ taken from $P(k)$ arbitrarily. We show the lemma by running through each $k_0$ in $P(k)$ in the top-down manner. When $P(k_0)=\emptyset$, by proposition \ref{pro:forest}, $\cup_{z\in{}V_{k_0}} A_z = \mathcal{M}$.
Now, take $k_1$ arbitrarily from $P(k)\setminus \{k\}$, and let $k_0\in P(k)$ such that $p(k_0)=k_1$. It suffices to show that $\left( P(k_1)\setminus \{k_1\} \right)  \cup \left(\cup{}_{z\in{}V_{k_1}} A_z \right) \subset \left( P(k_0)\setminus \{k_0\} \right)  \cup \left(\cup{}_{z\in{}V_{k_0}} A_z \right)$ holds.

Let $i\in \left( P(k_1)\setminus \{k_1\} \right)  \cup \left(\cup{}_{z\in{}V_{k_1}} A_z \right)$. Since $P(k_1)\setminus \{k_1\}\subset P(k_0)\setminus \{k_0\}$ and $P(k_1)\subset P(k_0)$, it suffices to look at when $i\in A_{k_1}$. Then, either $i=k_1$ or $i\in A_z$ for some $z\in y(k_0)\cup e(k_0) \cup \{k_0\}$. Therefore, we have shown that $i\in \left( P(k_0)\setminus \{k_0\} \right)  \cup \left(\cup{}_{z\in{}V_{k_0}} A_z \right)$.
\end{proof}

\begin{proof} [Proof of proposition \ref{pro:T_union}]
Lemma \ref{lem:maxT} says that $k=\max T(k)$, while for each $z\in{}T(k)$, $\max A_z=z$.
Therefore, $\cup_{z\in{}T(k)} A_z\subset{}\{1,2,\ldots{},k\}$. For the other direction, take any $i\in \mathcal{M}\setminus{} \cup{}_{z\in{}T(k)} A_z$. It suffices to show that $i>k$. Note that by definition $T(k)=\{k\} \cup \left( \cup_{z\in P(k)} y(z) \right)$.
By lemma \ref{lem:V}, either $i\in P(k)\setminus \{k\}$ or $i\in{}A_z$ for some $z\in{} \left(\cup_{z\in P(k)} e(z) \right) \cup{} \{k\}$. In the former case, clearly $i>k$. In the latter case, suppose that $i\in{}A_z$ for some $z\in{}e(z_0)$ where $z_0\in P(k)$. By proposition \ref{pro:consecutive}, since $k\in{}A_{z_0}$, $k\leq{} z_0 < v(z) \leq i$. Therefore, we have proven that $\cup_{z\in{}T(k)} A_z=\{1,2,\ldots{},k\}$.

For the second part of the proposition, suppose $T(k)=\{i_1,i_2,\ldots{},i_n\}$ with $i_1<i_2<\ldots{}<i_n$. It is not hard to see that by definition of $T(k)$, $A_{i_{q_1}}\cap{}A_{i_{q_2}}=\emptyset{}$ for any two distinct $q_1$ and $q_2$. By proposition \ref{pro:consecutive}, each $A_{i_q}$ is consecutive. Therefore, $i_q<v(i_{q+1})$ must hold for each $q=1,2,\ldots{},n-1$.
Since we have proven that $\cup_{q=1,2,\ldots{},n} A_{i_q}=\{1,2,\ldots{},k\}$ which is a consecutive set, $i_q+1=v(i_{q+1})$ must hold for each $q=1,2,\ldots{},n-1$.
\end{proof}

\begin{proof} [Proof of proposition \ref{pro:T_two}]
$T(k_0)\cap{}A_{k_1}$ and $T(k_0)\setminus{}A_{k_1}$ is a partitioning of $T(k_0)$.
By definition of $T(k_0)$, since $k_1\in P(k_0)$, it is not hard to see that $T(k_0)\cap{}A_{k_1}=\{k_0\}\cup\left(\cup_{z\in{}P(k_0)\setminus{}P(k_1)} y(z)\right)$.
Consequently, $T(k_0)\setminus{}A_{k_1}=\cup_{z\in{}P(k_1)} y(z)$.
For the second equality, it suffices to show that $z_1<z_2$ for any $z_1\in \cup_{z\in{}P(k_1)} y(z)$ and $z_2\in \{k_0\}\cup\left(\cup_{z\in{}P(k_0)\setminus{}P(k_1)} y(z)\right)$.
It suffices to assume that $z_2\neq{}k$ and $k_1 > k_0$.
By proposition \ref{pro:consecutive}, $z_1<v(k_1)\leq{}z_2$. Therefore, we are done.

\end{proof}

\subsection{One-shot job scheduling} \label{sec:thm_opt_machine}

Let $k_0:=\max \{m(J): J\in \mathcal{J}^{1d}\}$. For simplicity, for each machine type $z\in \mathcal{M}$, Let $H_{\geq z}:=\{J\in \mathcal{J}^{1d}: m(J)\geq z\}$ denote the set of jobs whose exact indexed machine type is greater than or equal to $z$. Let $\mathcal{O}$ be the set of all the optimal machine configurations.

\begin{lemma} \label{lem:optmachine_step1}
There exists an optimal machine configuration $w^*\in{}\mathcal{O}$ such that $\max \{z: w^* (z)>0\}\in P(k_0)$.
\end{lemma}

\begin{proof} [Proof of lemma \ref{lem:optmachine_step1}]
Take any $w_1\in{}\mathcal{O}$. Construct another machine configuration $w_2:=G_1 (w_1)$ by the following definition of function $G_1$. (see (\ref{eq:G_1}))

\begin{equation} \label{eq:G_1}
w_2(z):=
\begin {cases}
w_1(z), \text{~if~} z=1,2,\ldots,k_0-1, \\
w_1(z) + \sum_{e\in e(z)} \sum_{i\in A_e} \frac{r_i}{r_z} w_1(i), \text{~if~} z\in P(k_0),\\
0, \text{~otherwise}. \\
\end {cases}
\end{equation}
By lemma \ref{lem:V} and proposition \ref{pro:T_union}, the changes from $w_1$ to $w_2$ are just that the costs from the machine types $\mathcal{M}\setminus \left(\{1,2,\ldots,k_0-1\}\cup P(k_0)\right)$ have been shifted to the machine types $P(k_0)$, and the two machine configurations $w_1$ and $w_2$ have the same cost.

It remains to show that $w_2$ is a feasible solution to (\ref{opt:1d}). By the definition of $k_0$, for each $l>k_0$, $S(H_{\geq{}l})=0$. Consider the following inequalities:
\begin{equation*}
\begin{aligned}
S(H_{\geq{}k_0} )
&\leq{} \sum_{z=k_0}^{|\mathcal{M}|} w_1 (z)g_z = \sum_{z=k_0}^{|\mathcal{M}|} w_1 (z) r_z \frac{g_z}{r_z}\\
&\leq{}_{(*)} \sum_{z=k_0}^{|\mathcal{M}|} w_2 (z) r_z \frac{g_z}{r_z} = \sum_{z=k_0}^{|\mathcal{M}|} w_2 (z) g_z. \\
\end{aligned}
\end{equation*}
The first inequality is due to the feasibility of $w_1$. For the inequality indexed by $(*)$, observe that for each $z\in P(k_0)$,
$w_1(z) r_z + \sum_{e\in e(z)} \sum_{i\in A_e} w_1(i) r_i = w_2(z) r_z + \sum_{e\in e(z)} \sum_{i\in A_e} w_2(i) r_i $, while by proposition \ref{pro:T_ratio}, $g_{z}/r_{z} \geq{} g_i/r_i$, for each $i\in{}A_e$ for $e\in{}e(z)$.
Therefore, $w_1(z) r_z \frac{g_z}{r_z} + \sum_{e\in e(z)} \sum_{i\in A_e} w_1(i) r_i \frac{g_i}{r_i}
\leq{} w_1(z) r_z \frac{g_z}{r_z} + \sum_{e\in e(z)} \sum_{i\in A_e} w_1(i) r_i \frac{g_z}{r_z} = w_2(z) r_z \frac{g_z}{r_z}$. The inequality $\leq_{(*)}$ follows.

For each $l=1,2,\ldots{},k_0-1$, we have $S(H_{\geq{}l}) \leq \sum_{z=l}^{|\mathcal{M}|} w_1 (z) g_z \leq \sum_{z=l}^{|\mathcal{M}|} w_2 (z) g_z$, where the first inequality is due to the feasibility of $w_1$, and the second is because of $\leq_{(*)}$ and $w_2(z) = w_1(z)$ for each $z=1,2,\ldots,k_0-1$. Eventually, we have shown that $w_2$ is feasible and hence optimal. Clearly, $\max \{z: w_2(z)>0\}\in P(k_0)$.

\end{proof}

\begin{lemma} \label{lem:optmachine_step2}
There exists an optimal machine configuration $w^*\in{}\mathcal{O}$ such that $\max \{z: w^* (z)>0\}\in P(k_0)$ and $\sum_{z\in{}A_{z_0}\setminus \{z_0\}} w^*(z) r_z < r_{z_0}$ for each $z_0\in \mathcal{M}$.
\end{lemma}

\begin{proof} [Proof of lemma \ref{lem:optmachine_step2}]
Take $w_1\in{}G_1 (\mathcal{O})$. It suffices to construct an function $G_2$ which maps $w_1$ to some optimal machine configuration $w_2$ such that $\max \{z:w_2 (z)>0\} \in{} P(k_0)$ and \\
$\sum_{z\in{}A_{z_0}\setminus{}\{z_0\}} w_2 (z) r_z < r_{z_0}$ for each $z_0\in \mathcal{M}$. (see Algorithm \ref{alg:G_2} for the definition of $G_2$)

\begin{algorithm}
\SetAlgoLined
\KwIn{$w_1\in G_1 (\mathcal{O})$}
\KwOut{$w_2$}
$w_2 (z)\gets w_1 (z)$ for each $z\in \mathcal{M}$\;
\For {$z_0=1,2,\ldots,|\mathcal{M}|$}
{
	\If {$\sum_{z\in A_{z_0}\setminus \{z_0\}} w_2 (z) r_z \geq r_{z_0}$}
	{
		Take $z^*\in A_{z_0}\setminus \{z_0\}$ and $n_{z^*} \in \{0,1,\ldots,w_2 (z^* )\}$ such that $n_{z^*} r_{z^*} + \sum_{z\in\{z^*+1,\ldots,z_0-1\}}  w_2 (z) r_z = \floor*{\frac{\sum_{z\in A_{z_0}\setminus \{z_0\}} w_2 (z) r_z}{r_{z_0}}}\cdot r_{z_0}$ (see remark \ref{rem:G_2})\;
		$w_2 (z_0) \gets w_2 (z_0) + \floor*{\frac{\sum_{z\in A_{z_0}\setminus \{z_0\}} w_2 (z) r_z}{r_{z_0}}}$\;
		$w_2 (z)\gets 0 $ for each $z=z^*+1,\ldots,z_0-1$\;
		$w_2 (z^*) \gets w_2 (z^* )-n_{z^*}$\;
	}
}
\caption{$G_2$}
\label{alg:G_2}
\end{algorithm}

\begin{remark} \label{rem:G_2}
By proposition \ref{pro:consecutive}, $A_{z_0}\setminus \{z_0\} = \{v(z_0),v(z_0 )+1,\ldots,z_0-1\}$.
Because $\floor*{\frac{\sum_{z\in A_{z_0}\setminus \{z_0\}} w_2 (z) r_z}{r_{z_0}}}\cdot r_{z_0} \leq \sum_{z\in A_{z_0}\setminus \{z_0\}} w_2 (z)\cdot r_z$ and $r_{z_0} > r_z$ for each $z\in A_{z_0}\setminus \{z_0\}$, there must exist $z^*\in \{v(z_0),v(z_0)+1,\ldots,z_0-1\}$ such that
\begin{equation*}
\begin{aligned}
\sum_{z\in \{z^*+1,\ldots,z_0-1\}}w_2 (z)\cdot r_z
&\leq \floor*{\frac{\sum_{z\in A_{z_0}\setminus \{z_0\}} w_2 (z) r_z}{r_{z_0}}}\cdot r_{z_0}\\
&\leq \sum_{z\in \{z^*,\ldots,z_0-1\}}w_2 (z)\cdot r_z, \\
\end{aligned}
\end{equation*}
which implies that
\begin{equation*}
\begin{aligned}
 0 &\leq \floor*{\frac{\sum_{z\in A_{z_0}\setminus \{z_0\}} w_2 (z) r_z}{r_{z_0}}}\cdot r_{z_0} - \sum_{z\in \{z^*+1,\ldots,z_0-1\}} w_2 (z)\cdot r_z \\
 &\leq w_2 (z^* )\cdot r_{z^*}. \\
\end{aligned}
\end{equation*}
Since the cost rate of each machine type is a power of $8$,
$r_{z^*}$ divides $\sum_{z=z^*+1,\ldots,z_0} a_z\cdot r_z$ for integers $a_{z^*+1},\ldots ,a_{z_0}$. The existence of $n_{z^*}$ follows.
\end{remark}

Now, we show that $w_2$ is indeed an optimal machine configuration. In the definition of $G_2$, $\sum_{z\in \mathcal{M}} w_2 (z)\cdot r_z = \sum_{z\in \mathcal{M}} w_1 (z)\cdot r_z$ holds at the end of each iteration. So, the output machine configuration $w_2$ has the same cost as $w_1$.
Next, we show that $w_2$ is feasible to the optimization problem (\ref{opt:1d}), i.e., $S(H_{\geq l}) \leq \sum_{z=l}^{|\mathcal{M}|} w_2 (z) g_z$, for $l=1,2,\ldots, |\mathcal{M}|$. It suffices to assume that the machine configuration represented by $w_2$ is feasible at the beginning of each iteration and we are to show that that is still feasible at the end of the iteration. Take any iteration $z_0\in \{1,2,\ldots,|\mathcal{M}|\}$. Denote the two machine configurations at the beginning and the end of the $z_0$-th iteration by $w_a$ and $w_b$ respectively.
By definition, say $n_{z^*} r_{z^*} + \sum_{z\in \{z^*+1,\ldots,z_0-1\}} w_a (z) r_z = \floor*{\frac{\sum_{z\in A_{z_0}\setminus \{z_0\}} w_a (z) r_z}{r_{z_0}}}\cdot r_{z_0}$. Clearly $\frac{g_z}{r_z} \leq \frac{g_{z_0}}{r_{z_0}}$ for each $z\in A_{z_0}$. Therefore, we have the following result, for each $l=1,2,\ldots,|\mathcal{M}|$:
\begin{equation*}
\begin{aligned}
&\sum_{z\in A_{z_0}\setminus \{z_0\} \land z\geq l} (w^a(z) - w^b(z)) g_z\\
&\leq n_{z^*} g_{z^*} + \sum_{z\in \{z^*+1,\ldots,z_0-1\}} w_a (z) g_z \\
&\leq \floor*{\frac{\sum_{z\in A_{z_0}\setminus \{z_0\}} w_a (z) r_z}{r_{z_0}}} \cdot g_{z_0}\\
&= (w_b (z_0 ) - w_a (z_0 ))\cdot g_{z_0}\\
\end{aligned}
\end{equation*}
Since $z_0=\max A_{z_0}$, we have that for each $l\in \mathcal{M}$, $\sum_{z=l}^{|\mathcal{M}|} w_a (z) g_z \leq \sum_{z=l}^{|\mathcal{M}|} w_b (z) g_z$. Since $w_a$ is feasible, we also have $w_b$ feasible and hence optimal.

It remains to show that the output machine configuration $w_2$ satisfies the two properties: (1). $\max \{z:w_2 (z)>0\}\in P(k_0)$ and (2). $\sum_{z\in A_{z_0}\setminus \{z_0\}} w_2 (z) r_z < r_{z_0}$ for each $z_0=1,2,…,|\mathcal{M}|$.

For property (1): denote $\max \{z:w_2 (z)>0\}$ and $\max \{z: w_1(z)>0\}$ by $k_2$ and $k_1$ respectively.
By the definition of $G_2$, clearly $k_2 \geq k_1$ because in each iteration of algorithm \ref{alg:G_2} the variable $\max \{z:w_2 (z)>0\}$ does not decrease.
By lemma \ref{lem:V}, $(P(k_0)\setminus \{k_0\}) \cup \left( \cup_{z\in V_{k_0}} A_z \right) = \mathcal{M}$, and also note that $k_1\in P(k_0)$. Then, $k_2$ must be in $P(k_1)$ or $A_e$ for some $e\in e(x)$ for some $x\in P(k_0)$. (see proposition \ref{pro:T_union}). We show that the latter case is impossible.
On the one hand, by the definition of $G_1$, $w_1(i)=0$ for each $i\in A_e$. On the other hand, by the definition of $G_2$, $\exists z\in A_{k_2}$ s.t. $w_1 (z) > 0$. Therefore, we have shown $k_2\in P(k_1)\subset P(k_0)$ and hence the property (1) holds.

For property (2), we show it by contradiction. Assume that $\sum_{z\in A_{z_1}\setminus \{z_1\}} w_2 (z) r_z \geq r_{z_1}$ for some $z_1\in \mathcal{M}$. For each $z_0$-th iteration with $z_0\in \mathcal{M}$, let $w^{z_0}$ denote the the machine configurations $w_2$ represents at the end of the $z_0$-th iteration. At the end of the $z_1$-th iteration, we must have $r_{z_1} > \sum_{z\in A_{z_1}\setminus \{z_1\}} w^{z_1} (z) r_z$.
By the definition of algorithm \ref{alg:G_2}, observe that $\sum_{z\in A_{z_1}\setminus \{z_1\}} w^{z_1} (z) r_z \geq \sum_{z\in A_{z_1}\setminus \{z_1\}} w^{z_1 + 1} (z) r_z \geq \ldots \geq \sum_{z\in A_{z_1}\setminus \{z_1\}} w^{|\mathcal{M}|} (z) r_z$.
Consequently, $r_{z_1} > \sum_{z\in A_{z_1}\setminus \{z_1 \}} w_2 (z) r_z$ since $w^{|\mathcal{M}|}$ is actually the output machine configuration $w_2$. Therefore, property (2) is proven.

\end{proof}

\begin{proof} [Proof of theorem \ref{thm:optmachine}]
Take $w_1\in G_2 (G_1 (\mathcal{O}))$. Let $k_{opt}:=\max \{z:w_1 (z)>0\}\in \{k_0,\ldots,k_h\}$. We are to construct an optimal machine configuration $w_2:=G_3 (w_1 )$ which is the optimal machine configuration described in theorem \ref{thm:optmachine}.
If $T(k_{opt})$ is singleton, i.e. $v(k_{opt})=1$, then $w_1 (k_{opt}) \leq \ceil* {\frac{S(H_{k_{opt}}) } {g_{k_{opt}} } }$ clearly holds because anymore type-$k_{opt}$ machine must be redundant. If $T(k_{opt})$ is singleton or $w_1(k_{opt}) \leq \ceil*{\frac{S(H_{k_{opt}})} {g_{k_{opt}}} }$, let $w_2:=w_1$.
Otherwise, i.e., $w_1 (k_{opt}) > \ceil*{ \frac{S(H_{k_{opt}} )} {g_{k_{opt}} }}$ and $v(k_{opt})>1$. Let $s=\max T(k_{opt})\setminus \{k_{opt}\}=v(k_{opt})-1$.
\begin{equation*}
w_2(z):=
\begin{cases}
\ceil* {\frac{S(H_{k_{opt}})} {g_{k_{opt}} }} \geq 1, \text{~if~} z=k_{opt}, \\
w_1 (s) + \frac{r_{k_{opt}}}{r_s} \cdot \left(w_1 (k_{opt}) - \ceil*{\frac{S(H_{k_{opt}} )}{g_{k_{opt}}}} \right), \text{~if~} z=s, \\
w_1 (z), \text{~otherwise}. \\
\end{cases}
\end{equation*}
The above has defined $w_2=G_3(w_1)$ in each of the cases. It is clear that $w_2$ and $w_1$ have the same costs. Now, we show that $w_2$ is feasible.
Since for each $l>k_0$, $S(H_{\geq l})=0$, it remains to consider $l=1,2,\ldots,k_0$. For any $l=s+1,\ldots,k_0$ (Note that $s+1=v(k_{opt})\leq k_0$. ), $S(H_{\geq l}) \leq S(H_{\geq s+1} ) = S(H_{k_{opt}}) \leq \ceil*{\frac{S(H_{k_{opt}})} {g_{k_{opt}}}} g_{k_{opt}} = w_2 (k_{opt}) g_{k_{opt}}\leq \sum_{z=l}^{|\mathcal{M}|} w_2 (z) g_z$. For $l=1,2,\ldots,s$, we have $S(H_{\geq l} ) \leq \sum_{z=l}^{|\mathcal{M}|} w_1 (z) g_z \leq \sum_{z=l}^{|\mathcal{M}|} w_2 (z) g_z$ where the first inequality is due to the feasibility of $w_1$ and the second is due to $\frac{g_s}{r_s} \geq \frac{g_{k_{opt}}} {r_{k_{opt}}}$ and $(w_2 (s) - w_1 (s)) r_s = (w_1 (k_{opt})-w_2 (k_{opt})) r_{k_{opt}} \geq 0$. Finally, we have shown that $w_2$ is feasible and hence optimal.

Next, we show that $w_2$ satisfies the first two properties. Clearly, $\max \{z:w_2 (z)>0\}= \max \{z:w_1 (z)>0\} = k_{opt}\in P(k_0)$. For the second property, since the only machine type whose usage has been increased from $w_1$ to $w_2$ is type $s$, it suffices to check the trees rooted at the nodes $P(s)\setminus \{s\}$. Since $s\in T(k_{opt})$, by definition of $T(k_{opt})$, $P(s)\setminus \{s\}\subset P(k_{opt})\setminus \{k_{opt}\}$. For each $z\in P(s)\setminus \{s\}$, both $k_{opt}$ and $s$ are in $A_z$. This implies that $\sum_{i\in A_z\setminus \{z\}} w_2 (i) r_i = \sum_{i\in A_z\setminus \{z\}} w_1 (i) r_i < r_z$ for each $z\in P(s)\setminus \{s\}$. Therefore, $w_2$ indeed satisfies the first two properties.

At last, we show that $\floor*{\frac{S(H_{k_{opt}} )} {g_{k_{opt}}} } \leq w_2 (k_{opt})$ holds as a consequence of the second property. Note that $\sum_{z\in A_{k_{opt}}\setminus \{k_{opt}\}} w_2 (z) r_z  < r_{k_{opt}}$ implies that
\begin{equation*}
\begin{aligned}
\sum_{z\in A_{k_{opt}}\setminus \{k_{opt}\}} w_2 (z) g_z
&= \sum_{z\in A_{k_{opt}}\setminus \{k_{opt}\}} w_2 (z) r_z \cdot \frac{g_z}{r_z} \\
&< \left( \sum_{z\in A_{k_{opt}}\setminus \{k_{opt}\}} w_2 (z) r_z \right) \cdot \frac{g_{k_{opt}}}{r_{k_{opt}}}\\
&<r_{k_{opt}} \cdot \frac{g_{k_{opt}}}{r_{k_{opt}}} = g_{k_{opt}}.\\
\end{aligned}
\end{equation*}
By the feasibility of $w_2$, $S(H_{k_{opt}} ) = S(H_{\geq v(k_{opt})}) \leq \sum_{z\in A_{k_{opt}}\setminus \{k_{opt}\}} w_2 (z) g_z + w_2 (k_{opt}) g_{k_{opt}} < (w_2 (k_{opt})+1) g_{k_{opt}}$, where the last inequality is the result from the above equation in display. After dividing both sides by $g_{k_{opt}}$, we have $\floor*{\frac{S(H_{k_{opt}} )} {g_{k_{opt} }}} \leq \frac{S(H_{k_{opt}} )} {g_{k_{opt} }} < w_2 (k_{opt})+1$, which implies that $\floor*{\frac{S(H_{k_{opt}} )} {g_{k_{opt} }}} \leq w_2 (k_{opt})$.

\end{proof}

\section{The offline setting} \label{sec:offline}
\subsection{The offline algorithm: $ALG_{offline}$}

\begin{proof} [Proof of property \ref{property:off_head}]
By the definition of $ALG_{offline}$, $k_{off}\notin \cup_{e\in \cup_{z\in P(k_0)} e(z)} A_e$. By lemma \ref{lem:V} and proposition \ref{pro:T_union}, $k_{off}\in \{1,2,\ldots,k_0\} \cup \left( P(k_0)\setminus \{k_0\} \right)$.
On the other hand, $\{J\in \mathcal{J}(t) : m(J)=k_0 \}\neq \emptyset$ implies that $K_z (t)\neq \emptyset$ for some $z\in P(k_0)$. This implies that $k_{off} \geq k_0$. Therefore, $k_{off}\in P(k_0)$.

\end{proof}

\begin{proof} [Proof of property \ref{property:off_upperb}]

We are to show that for any $z\in \mathcal{M}$ and $t\in \sspan(\mathcal{J})$, $\sum_{i\in{A_z\setminus\{z\}}} \ceil*{\frac{S(K_i,t)}{g_i} r_i} \leq 2\cdot{}r_z $ holds.
Use induction on the height $l\geq 1$ of the tree rooted at $z$, where the height of the tree rooted at $z$ is defined as $\max \{|P(i)\cap A_z| : i\in A_z\}$. The base case when $l=1$ is trivially true, so it suffices to consider the inductive case.
Let $\mathbb{T}:=\{t_0: \mathcal{R}_z^h (t_0)=\emptyset \land c_{z,t_0} < \frac{1}{3} r_z \}$. For each job $J\in \mathcal{R}_z$, by definition of $ALG_{offline}$, $J\in K_z$ if and only if $I(J)\cap \mathbb{T}=\emptyset$.(line 9-11) For the fixed time $t$, we have either $t\in \mathbb{T}$ or $t\notin \mathbb{T}$.

Case 1: $t\in \mathbb{T}$.

In this case, $K_z (t)=\emptyset$. For each job $J\in \mathcal{R}_z (t)$, $J\in \mathcal{R}_x (t)$ for some $x\in f(z)$, i.e., all the jobs in $\mathcal{R}_z(t)$ are passed to $z$'s children nodes. Clearly, $\mathcal{R}_z (t) = \cup_{x\in f(z)} \mathcal{R}_x (t)$. By definition, $F_{x,t}=\mathcal{R}_x(t)\supset K_x(t)$ for each $x\in f(z)$. (line 5) Hence, $\sum_{x\in f(z)} \ceil*{\frac{S(K_x,t)}{g_x}} r_x \leq \sum_{x\in f(z)} \ceil*{\frac{S(F_{x,t})} {g_x}} r_x = c_{z,t} < \frac{1}{3} r_z$. Consider the following:
\begin{equation*}
\begin{aligned}
&\quad \sum_{i\in A_z\setminus \{z\}} \ceil*{\frac{S(K_i,t)}{g_i}} r_i\\
&=\sum_{x\in f(z)} \left( \ceil*{\frac{S(K_x,t)}{g_x}} r_x + \sum_{i\in A_x\setminus \{x\}} \ceil*{\frac{S(K_i,t)}{g_i}} r_i \right)\\
&\leq \sum_{x\in f(z)} \ceil*{\frac{S( F_{x,t} ) } {g_x}} r_x + 2\cdot \ceil*{ \frac{ S( F_{x,t} ) }{g_x}} r_x\\
&=3\cdot c_{z,t} < 3\cdot (\frac{1}{3} r_z)=r_z.\\
\end{aligned}
\end{equation*}
Here we explain the above inequalities. By inductive hypothesis, $\sum_{i\in A_x\setminus \{x\}} \ceil*{\frac{S(K_i,t)}{g_i}} r_i \leq 2\cdot r_x$ for each $x\in f(z)$, since the heights of trees rooted at the children nodes of $z$ are all strictly less than the height of the tree rooted at $z$. If $\sum_{i\in A_x\setminus \{x\}} \ceil*{\frac{S(K_i,t)}{g_i}} r_i > 0$, then $\mathcal{R}_x(t)$ is nonempty as well as $F_{x,t}$. Therefore, $\sum_{i\in A_x\setminus \{x\}} \ceil*{\frac{S(K_i,t)}{g_i}} r_i \leq 2\cdot \ceil*{\frac{S(F_{x,t}) }{g_x}} r_x$ always holds.

Case 2: $t\notin \mathbb{T}$.

Denote $\mathcal{R}_z (t_0)\setminus K_z (t_0)$ by $\mathcal{R}_z' (t_0)$ for any $t_0\in \sspan(\mathcal{J})$, i.e., $\mathcal{R}_z' (t_0)$ represents the set of jobs active at $t_0$ which are passed into the children nodes of $z$. Clearly, $\mathcal{R}_z' (t_0)=\cup_{x\in f(z)} \mathcal{R}_x(t_0) $.
By definition, we also have $\mathcal{R}_z' (t) = \{J\in \mathcal{J}(t)\setminus \left(\cup_{i=|\mathcal{M},|\mathcal{M}|-1,\ldots,z+1} K_i (t) \right) : m(J)\in A_z\setminus \{z\} \land I(J)\cap \mathbb{T}\neq \emptyset\}$.
Consider two sets: $S_1:=\sspan(\mathcal{R}_z' (t))\cap \mathbb{T} \cap (-\infty,t)$ and $S_2:=\sspan(\mathcal{R}_z' (t))\cap \mathbb{T} \cap (t,\infty)$.
If both $S_1$ and $S_2$ are empty, clearly $\mathcal{R}_z' (t)=\emptyset$. Then, for each $i\in A_z\setminus \{z\}$, $K_i (t)=\emptyset$. Then $\sum_{i\in A_z\setminus \{z\}} \ceil*{\frac{S(K_i,t)} {g_i}} r_i = 0$. Therefore, it suffices to consider the cases when at least one of $S_1$ and $S_2$ is nonempty.

Case 2.1: both $S_1$ and $S_2$ are nonempty.

Let $t_1:=\sup S_1$ and $t_2:=\min S_2$. The essential reason why $\sup$ is used in the first definition and $\min$ is used in the second definition is that the active intervals of jobs are defined as left closed and right open intervals. Observe that $t_1 \leq t < t_2$. Suppose that $\epsilon > 0$ is infinitesimal. Note that $t_1-\epsilon \in \mathbb{T}$ and $t_2\in \mathbb{T}$. Then, we have for each $J\in \mathcal{R}_z' (t)$, either $t_1-\epsilon$ or $t_2$ is in $I(J)$, otherwise $I(J)\cap \mathbb{T}$ would be empty. So, we have $\mathcal{R}_z' (t)\subset \mathcal{R}_z' (t_1-\epsilon) \cup \mathcal{R}_z' (t_2) = \mathcal{R}_z (t_1-\epsilon) \cup \mathcal{R}_z (t_2)$, where the equality is because $t_1-\epsilon, t_2\in \mathbb{T}$. (see the discussion in case 1)
Consequently, after the partitioning of jobs into children nodes of $z$, we have $\mathcal{R}_x (t) \subset F_{x,t_1-\epsilon} \cup F_{x,t_2}$ for each $x\in f(z)$. Therefore,
\begin{equation} \label{eq:property_off_upperb_case2.1}
\begin{aligned}
&\quad \sum_{x\in f(z)} \ceil*{\frac{S(\mathcal{R}_x,t)} {g_x} } r_x\\
&\leq \sum_{x\in f(z)} \ceil*{\frac{S(F_{x,t_1-\epsilon}\cup F_{x,t_2})} {g_x} } r_x\\
&\leq \sum_{x\in f(z)} \ceil*{\frac{S(F_{x,t_1-\epsilon})} {g_x} } r_x + \sum_{x\in f(z)} \ceil*{\frac{S(F_{x,t_2})} {g_x} } r_x\\
&< \frac{2}{3} r_z, \text{~since~} t_1-\epsilon, t_2 \in \mathbb{T}. \\
\end{aligned}
\end{equation}

Similarly as case 1, we have
\begin{equation*}
\begin{aligned}
&\quad \sum_{i \in A_z\setminus \{z\}} \ceil*{\frac{S(K_i,t)} {g_i}} r_i\\
&=\sum_{x\in f(z)} \left(\ceil*{ \frac{S(K_x,t)} {g_x}} r_x + \sum_{i\in A_x\setminus \{x\}} \ceil*{ \frac{S(K_i,t)} {g_i}} r_i \right)\\
&\leq \sum_{x\in f(z)} \left( \ceil*{\frac{S(\mathcal{R}_x,t)}{g_x}} r_x+2\cdot \ceil*{\frac{S(\mathcal{R}_x,t)}{g_x}} r_x \right)\\
&=3\cdot \sum_{x\in f(z)} \ceil*{\frac{S(\mathcal{R}_x,t)}{g_x} } r_x\\
&<3\cdot (\frac{2}{3} r_z ) = 2 \cdot r_z, \text{~by~equation~(\ref{eq:property_off_upperb_case2.1})},\\
\end{aligned}
\end{equation*}
where the first inequality has used inductive hypothesis on the trees rooted at $x$ for each $x\in f(z)$.

Case 2.2: either $S_1$ or $S_2$ is empty.

Due to symmetry, it suffices to assume $S_1\neq \emptyset$ and $S_2=\emptyset$. Let $t_1:=\sup S_1$. Let $\epsilon>0$ be infinitesimal. Similarly, $\mathcal{R}_z' (t)\subset \mathcal{R}_z (t_1-\epsilon)$. Furthermore, $\mathcal{R}_z' (t) = \cup_{x\in f(z)} \mathcal{R}_x (t)$ and $\mathcal{R}_z (t_1-\epsilon) = \cup_{x\in f(z)} \mathcal{R}_x (t_1-\epsilon) = \cup_{x\in f(z)} F_{x,t_1-\epsilon}$. Therefore, $\mathcal{R}_x (t) \subset F_{x,t_1-\epsilon}$. Therefore,
\begin{equation} \label{eq:property_off_upperb_case2.2}
\begin{aligned}
&\quad \sum_{x\in f(z)} \ceil*{\frac{S(\mathcal{R}_x,t)} {g_x} } r_x\\
&\leq \sum_{x\in f(z)} \ceil*{\frac{S(F_{x,t_1-\epsilon})} {g_x} } r_x\\
&=c_{z,t_1-\epsilon}< \frac{1}{3} r_z.\\
\end{aligned}
\end{equation}

Finally,
\begin{equation*}
\begin{aligned}
&\quad \sum_{i\in A_z\setminus \{z\} } \ceil*{ \frac{S(K_i,t)} {g_i} } r_i\\
&=\sum_{x\in f(z)}  \left( \ceil*{ \frac{S(K_x,t)} {g_x} } r_x + \sum_{i\in A_x \setminus \{x\}} \ceil*{ \frac{S(K_i,t)} {g_i} } r_i \right)\\
&\leq \sum_{x\in f(z)} \left( \ceil*{\frac{S(\mathcal{R}_x,t)} {g_x} } r_x + 2\cdot \ceil*{\frac{S(\mathcal{R}_x,t)} {g_x} } r_x \right)\\
&=3\cdot \sum_{x\in f(z)} \ceil*{\frac{S(\mathcal{R}_x,t)} {g_x} } r_x\\
&< 3\cdot (\frac{1}{3} r_z )=r_z, \text{~by~equation~(\ref{eq:property_off_upperb_case2.2}),} \\
\end{aligned}
\end{equation*}
where the first inequality has used inductive hypothesis on the trees rooted at $x$ for each $x\in f(z)$.

Eventually, we have proved the inductive case.
\end{proof}

\begin{proof} [Proof of property \ref{property:off_lowerb}]
$K_z (t)\neq \emptyset$ and $\mathcal{R}_z^h=\emptyset$ implies that $c_{z,t}=\sum_{x\in f(z)} \ceil*{ \frac{S(F_{x,t})} {g_x}} r_x \geq \frac{1}{3} r_z$. Meanwhile, the condition $K_i(t)= \emptyset$ for each $i\in P(z)\setminus \{z\}$ says that $R_z (t) = \mathcal{R}_z (t)$. Therefore, $\{F_{x,t} : x\in f(z)\}$ is not only the partitioning of $\mathcal{R}_z (t)$ but also $R_z(t)$. Therefore, by the definition of $F_{x,t}$ and $R_x(t)$, we have $F_{x,t} = R_x (t)$ for each $x\in f(z)$. Consequently,
\begin{equation*}
\begin{aligned}
\sum_{x\in f(z)} S(R_x,t)  \frac{r_x} {g_x}
&= \sum_{x\in f(z)} S(F_{x,t})  \frac{r_x} {g_x}\\
&\geq \sum_{x\in f(z)} \left( \ceil*{\frac{S(F_{x,t})} {g_x}} - 1 \right) r_x\\
&= \sum_{x\in f(z)} \ceil*{\frac{S(F_{x,t})} {g_x}} r_x - \sum_{x\in f(z)} r_x\\
&\geq \frac{1}{3} r_z - \frac{1}{7} r_z = \frac{4}{21} r_z. \\
\end{aligned}
\end{equation*}
\end{proof}

\subsection{An $O(1)$ approximation of optimal one-shot scheduling} \label{subsec_app:one_shot}

\begin{proof} [Proof of proposition \ref{pro:cn&r'_total}]
By the definition of the alternative machine configuration $\cn_{z^*}$, for each $z\in T(z^*)\setminus \{z^*\}$,
\begin{equation} \label{eq:cn_total_1}
S(H_z) \geq \cn_{z^*}(z)\cdot g_z
\end{equation}
By proposition \ref{pro:T_ratio}, for each $z\in T(z^*)$,
\begin{equation} \label{eq:cn_total_2}
\frac{r_z}{g_z} \leq \frac{r_{z^*}}{g_{z^*}}
\end{equation}

There are two cases when $z_1=z^*$ or $z_1 < z^*$.

Case 1: $z_1=z^*$.

On the one hand, it can be observed that for each $z_0\in T(z^*)$, the initial part of the alternative machine configuration $\cn_{z^*} (z)$ with $z\in T(z^*) \land z\geq z_0$ is actually able to contain all the jobs in $\cup_{z\in T(z^*) \land z\geq z_0} H_z$. Therefore, $0\leq \sum_{z\in T(z^*) \land z \geq z_0} \cn_{z^*} (z) \cdot g_z - S(H_z)$ holds. So, we have
\begin{equation*}
\begin{aligned}
&\quad \sum_{z\in T(z^*) \land z\geq z_0} \cn_{z^*} (z) \cdot r_z - S(H_z) \cdot \frac{r_z}{g_z}\\
&= \sum_{z\in T(z^*) \land z\geq z_0} (\cn_{z^*} (z) g_z - S(H_z)) \cdot \frac{r_z} {g_z}\\
&\geq \sum_{z\in T(z^*) \land z\geq z_0} (\cn_{z^*} (z) g_z - S(H_z)) \cdot \frac{r_{z^*}} {g_{z^*}}, \text{~by~equation~(\ref{eq:cn_total_1})~and~(\ref{eq:cn_total_2})}\\
&= \left( \sum_{z\in T(z^*) \land z\geq z_0} \cn_{z^*} (z) g_z - S(H_z) \right)\cdot \frac{r_{z^*}} {g_{z^*}}\\
&\geq 0. \\
\end{aligned}
\end{equation*}

On the other hand, we have
\begin{equation*}
\begin{aligned}
&\quad \sum_{z\in T(z^*) \land z\geq z_0} \cn_{z^*} (z) \cdot r_z - S(H_z) \cdot \frac{r_z}{g_z}\\
&= \sum_{z\in T(z^*) \land z\geq z_0} \left( \cn_{z^*} (z) \cdot g_z - S(H_z) \right) \cdot \frac{r_z}{g_z}\\
&\leq \left(\cn_{z^*} (z^*) \cdot g_{z^*} - S(H_{z^*}) \right)\cdot \frac{r_{z^*}}{g_{z^*}}, \text{~by~equation~(\ref{eq:cn_total_1})}\\
&= \left( \ceil*{\frac{S(H_{z^*})}{g_{z^*}}}\cdot g_{z^*} - S(H_{z^*})\right)\cdot \frac{r_{z^*}}{g_{z^*}}\\
&\leq r_{z^*}\\
\end{aligned}
\end{equation*}

Case 2: $z_1 < z^*$.

By equation (\ref{eq:cn_total_1}), $\sum_{z\in T(z^*) \land z_0 \leq z \leq z_1} \cn_{z^*} (z) \cdot r_z - S(H_z) \cdot \frac{r_z}{g_z} = \sum_{z\in T(z^*) \land z_0 \leq z \leq z_1} \left( \cn_{z^*} (z) \cdot g_z - S(H_z) \right) \cdot \frac{r_z}{g_z} \leq 0$. On the other hand, by the definition of the alternative machine configuration, \\
$\sum_{z\in T(z^*) \land z_0 \leq z \leq z_1} S(H_z) - \cn_{z^*}(z)\cdot g_z < g_{z^*}$. Therefore, we have
\begin{equation*}
\begin{aligned}
&\quad \sum_{z\in T(z^*) \land z_0 \leq z \leq z_1} S(H_z) \cdot \frac{r_z}{g_z} - \cn_{z^*} (z) \cdot r_z\\
&= \sum_{z\in T(z^*) \land z_0 \leq z \leq z_1} \left( S(H_z) - \cn_{z^*} (z) \cdot g_z \right) \cdot \frac{r_z}{g_z}\\
&\leq \sum_{z\in T(z^*) \land z_0 \leq z \leq z_1} \left( S(H_z) - \cn_{z^*} (z) \cdot g_z \right) \cdot \frac{r_{z^*}}{g_{z^*}}, \text{~by~equation~(\ref{eq:cn_total_1}),~(\ref{eq:cn_total_2})}\\
&= \left( \sum_{z\in T(z^*) \land z_0 \leq z \leq z_1} S(H_z) - \cn_{z^*} (z) \cdot g_z  \right) \cdot \frac{r_{z^*}}{g_{z^*}}\\
&\leq r_{z^*}\\
\end{aligned}
\end{equation*}

\end{proof}

\begin{proof} [Proof of proposition \ref{pro:cn&r'_head}]

By Proposition \ref{pro:T_ratio}, we have $\frac{r_z}{g_z}\leq \frac{r_{z^*}}{g_{z^*}}$ for each $z\in T(z^*)$.
Suppose $i^*\in T(z^*)\setminus{} \{z^*\}$ is the boundary machine type satisfying equations (\ref{eq:boundary_machine_type1}) and (\ref{eq:boundary_machine_type2}). Either $z_0>i^*$ or $z_0\leq{}i^*$.

If $z_0>i^*$, by equation (\ref{eq:boundary_machine_type1}), we have
\begin{eqnarray}
\sum_{z\in{}T(z^*) \land z\geq{}z_0} S(H_z)\cdot{}\frac{r_z}{g_z} & \leq & \bigg(\sum_{z\in{}T(z^*) \land z\geq{}z_0} S(H_z)\bigg) \cdot \frac{r_{z^*}}{g_{z^*}} \nonumber \\
& \leq & \cn_{z^*}(z^*) \cdot g_{z^*} \cdot \frac{r_{z^*}}{g_{z^*}} = \cn_{z^*}(z^*) \cdot r_{z^*}. \nonumber
\end{eqnarray}

If $z_0 \leq i^*$, by the definition of $\cn_{z^*}$, we have
\begin{eqnarray}
\lefteqn{\sum_{z\in{}T(z^*) \land z\geq{}z_0} S(H_z)\cdot{}\frac{r_z}{g_z}} \nonumber \\
& = & \sum_{z\in{}T(z^*) \land z>i^*} S(H_z)\cdot{}\frac{r_z}{g_z} + \sum_{z\in{}T(z^*) \land i^*>z \geq z_0} S(H_z)\cdot{}\frac{r_z}{g_z} \nonumber \\
& & + \bigg(\sum_{z\in{}T(z^*) \land z\geq{}i^*} S(H_z) - \cn_{z^*}(z^*)\cdot{}g_{z^*} \bigg)\cdot \frac{r_{i^*}}{g_{i^*}} \nonumber\\
& & + \bigg(\cn_{z^*}(z^*)\cdot{}g_{z^*} - \sum_{z\in{}T(z^*) \land z>i^*} S(H_z) \bigg) \cdot{}\frac{r_{i^*}}{g_{i^*}} \nonumber \\
& = & \sum_{z\in{}T(z^*) \land z>i^*} S(H_z) \cdot \frac{r_{z}}{g_{z}} + \sum_{z\in{}T(z^*) \land i^*>z \geq z_0} \cn_{z^*}(z)\cdot{}g_z\cdot{}\frac{r_z}{g_z} \nonumber \\
& & + \cn_{z^*}(i^*)\cdot{}g_{i^*}\cdot \frac{r_{i^*}}{g_{i^*}} + \bigg(\cn_{z^*}(z^*)\cdot{}g_{z^*} - \sum_{z\in{}T(z^*) \land z>i^*} S(H_z) \bigg) \cdot{}\frac{r_{i^*}}{g_{i^*}} \nonumber \\
& \leq & \sum_{z\in{}T(z^*) \land z>i^*} S(H_z) \cdot \frac{r_{z^*}}{g_{z^*}} + \sum_{z\in{}T(z^*) \land i^*>z \geq z_0} \cn_{z^*}(z)\cdot{}r_z \nonumber \\
& & + \cn_{z^*}(i^*) \cdot r_{i^*} + \bigg(\cn_{z^*}(z^*)\cdot{}g_{z^*} - \sum_{z\in{}T(z^*) \land z>i^*} S(H_z) \bigg) \cdot{}\frac{r_{z^*}}{g_{z^*}} \nonumber \\
& = & \cn_{z^*}(z^*) \cdot r_{z^*} + \cn_{z^*}(i^*) \cdot r_{i^*} + \sum_{z\in{}T(z^*) \land i^*>z \geq z_0} \cn_{z^*}(z)\cdot{}r_z \nonumber \\
& = & \sum_{z\in{}T(z^*) \land z\geq{}z_0} \cn_{z^*}(z) \cdot r_z. \nonumber
\end{eqnarray}

\end{proof}

\begin{proof} [Proof of the left inequality in theorem \ref{thm:cn_app}]
To show all the details of the proof, we need a lemma first.

\begin{lemma} \label{lem:feasibility}
Suppose that $(a_j)^n_{j=1}, (b_j)^n_{j=1},(c_j)^n_{j=1}$ are three sequences of non-negative values. If $\sum_{j=x}^{n} a_j \leq \sum_{j=x}^{n} b_j$ for each $x=1,2,\ldots,n$ and $(c_j)^n_{j=1}$ is an non-decreasing sequence, then we have $a_1c_1+a_2c_2+\ldots+a_nc_n \leq b_1c_1+b_2c_2+\ldots+b_nc_n$.
\end{lemma}

\begin{proof} [Proof of lemma \ref{lem:feasibility}]
We prove the lemma by using induction on $n$. It suffices to consider the inductive case. Assume the lemma holds for $n-1\geq 1$, and we show that the lemma holds for $n \geq 2$.By the conditions given and the induction assumption, we have
\begin{equation} \label{eq:lem_feasibility_1}
\begin{aligned}
&\quad (\sum_{x=2}^{n} b_x - \sum_{x=3}^{n} a_x) c_2 + a_3c_3+\ldots+a_{n}c_{n} \\
&\leq b_2c_2 + b_3c_3 + \ldots + b_{n}c_{n}, \\
\end{aligned}
\end{equation}
because $\sum_{x=x_0}^{n} a_x \leq \sum_{x=x_0}^{n} b_x$ for each $x_0=3,4,\ldots,n$, $\sum_{x=2}^{n} b_x - \sum_{x=3}^{n} a_x \geq 0$, and $\sum_{x=2}^{n} b_x - \sum_{x=3}^{n} a_x + a_3 +\ldots + a_n = b_2 + b_3 + \ldots + b_n$.

Therefore,
\begin{equation*}
    \begin{aligned}
    &\quad a_1c_1+a_2c_2+\ldots+a_nc_n\\
    &\leq (\sum_{x=1}^n a_x - \sum_{x=2}^{n} b_x)c_1 + (\sum_{x=2}^{n} b_x - \sum_{x=3}^{n} a_x) c_2 + a_3c_3+\ldots+a_{n}c_{n}\\
    &\leq (\sum_{x=1}^n a_x - \sum_{x=2}^{n} b_x)c_1 + b_2c_2 + b_3c_3 + \ldots + b_{n}c_{n}, \text{~by~equation~(\ref{eq:lem_feasibility_1})}\\
    &\leq b_1c_1+b_2c_2+\ldots+b_n c_n,\\
    \end{aligned}
\end{equation*}
where
the first inequality is because $(\sum_{x=1}^n a_x - \sum_{x=2}^{n} b_x)c_1 + (\sum_{x=2}^{n} b_x - \sum_{x=3}^{n} a_x) c_2 - (a_1c_1+a_2c_2) = (\sum_{x=2}^n a_x - \sum_{x=2}^{n} b_x)c_1 + (\sum_{x=2}^{n} b_x - \sum_{x=2}^{n} a_x) c_2 = (\sum_{x=2}^{n} b_x - \sum_{x=2}^{n} a_x) (c_2-c_1) \geq 0$, and
the last inequality is due to $b_1 \geq \sum_{x=1}^n a_x - \sum_{x=2}^{n} b_x$ and $c_1 \geq 0$.
\end{proof}

Case 1:  $z^\diamond=k_{opt}$.

By Theorem \ref{thm:optmachine}, $w^*(k_{opt})\geq{}1$ is either $\floor*{\frac{S(H_{k_{opt}})}{g_{k_{opt}}}}$ or $\ceil*{\frac{S(H_{k_{opt}})}{g_{k_{opt}}}}$.

Case 1.1: $w^*(k_{opt})=\floor*{\frac{S(H_{k_{opt}})}{g_{k_{opt}}}}$.

By the definition of $z^\diamond$, we have
\begin{equation} \label{eq:cn1.1part1}
\begin{aligned}
\cn_{z^\diamond} (z^\diamond) \cdot r_{z^\diamond}
&= \ceil*{\frac{S(H_{k_{opt}})}{g_{k_{opt}}}} \cdot r_{k_{opt}}\\
&\leq{} 2\cdot{} \left(\floor*{\frac{S(H_{k_{opt}})}{g_{k_{opt}}}} \cdot  r_{k_{opt}} \right)\\
&= 2\cdot{} w^* (k_{opt})\cdot r_{k_{opt}}.\\
\end{aligned}
\end{equation}

On the other hand, we can also show that
\begin{equation} \label{eq:cn1.1part2}
\sum_{z\in{}T(z^\diamond)\setminus{}\{z^\diamond\}} \cn_{z^\diamond}(z) r_z \leq \sum_{z=1,2,\ldots,k_{opt}-1} w^*(z) r_z
\end{equation}

By the definition of $\cn_{z^\diamond}$, for each $z\in{}T(z^\diamond)\setminus{}\{z^\diamond\}$, $\cn_{z^\diamond}(z) r_z \leq{} S(H_z) \frac{r_z}{g_z}$. Thus,  $\sum_{z\in{}T(z^\diamond)\setminus{}\{z^\diamond\}} \cn_{z^\diamond}(z) r_z \leq \sum_{z\in{}T(z^\diamond )\setminus{}\{z^\diamond\}} S(H_z) \frac{r_z}{g_z}$. Therefore, it suffices to show that\\ $\sum_{z\in{}T(z^\diamond)\setminus{}\{z^\diamond\}} S(H_z) \frac{r_z}{g_z} \leq \sum_{z=1,2,\ldots,k_{opt}-1} w^*(z) r_z$.
Suppose that $T(z^\diamond)\setminus\{z^\diamond\}=\{i_1,i_2,\ldots,i_{n-1}\}$ where $i_1<i_2<\ldots<i_{n-1}$. By the feasibility of $w^*$, $\left(\sum_{j=x,x+1,\ldots,n-1} S(H_{i_{j}}) \right)+S(H_{z^\diamond}) \leq \sum_{z\geq{}v(i_{x})} w^*(z) g_z$ for each $x=1,2,\ldots,n-1$. Since $w^*(k_{opt})\cdot g_{k_{opt}} = \floor*{\frac{S(H_{k_{opt}})}{g_{k_{opt}}}}\cdot g_{k_{opt}} < S(H_{k_{opt}})$, we have $\sum_{j=x,x+1,\ldots,n-1} S(H_{i_{j}}) \leq \sum_{v(i_{x})\leq z \leq k_{opt}-1} w^*(z) g_z$. By Proposition $\ref{pro:T_ratio}$, $(\frac{r_{i_{j}}}{g_{i_{j}}})_{j=1}^{n-1}$ is a non-decreasing sequence. By Lemma \ref{lem:feasibility}, we have $\sum_{z\in{}T(z^\diamond)\setminus{}\{z^\diamond\}} S(H_z) \frac{r_z}{g_z}
\leq \left(\sum_{v(i_1)\leq z <v(i_2)} w^*(z) g_z\right)\cdot \frac{r_{i_1}}{g_{i_1}} + \left(\sum_{v(i_2)\leq z <v(i_3)} w^*(z) g_z \right) \cdot \frac{r_{i_2}}{g_{i_2}}+ \ldots + \left(\sum_{v(i_{n-1})\leq z \leq k_{opt}-1} w^*(z) g_z\right) \cdot \frac{r_{i_{n-1}}}{g_{i_{n-1}}}$. 
By the definition of the cost-per-capacity graph,\\ $\left(\sum_{v(i_1)\leq z <v(i_2)} w^*(z) g_z\right)\cdot \frac{r_{i_1}}{g_{i_1}} + \left(\sum_{v(i_2)\leq z <v(i_3)} w^*(z) g_z \right) \cdot \frac{r_{i_2}}{g_{i_2}}+ \ldots + \left(\sum_{v(i_{n-1})\leq z \leq k_{opt}-1} w^*(z) g_z\right) \cdot \frac{r_{i_{n-1}}}{g_{i_{n-1}}}
\leq \sum_{z=1,2,\ldots,k_{opt}-1} w^*(z) r_z$. Therefore, equation (\ref{eq:cn1.1part2}) has been proved.

In summary,
\begin{equation*}
\begin{aligned}
\sum_{z\in T(z^\diamond)} \cn_{z^\diamond}(z) r_z &= \cn_{z^\diamond}(z^\diamond)  r_{z^\diamond}+\sum_{z\in{}T(z^\diamond)\setminus{}\{z^\diamond\}} \cn_{z^\diamond}(z) r_z\\
&\leq{} 2\cdot{}w^*(k_{opt})\cdot r_{k_{opt}}+\sum_{z=1,2,\ldots{},k_{opt}-1}w^* (z) r_z\\ &\quad\text{~(by~equations~(\ref{eq:cn1.1part1}),(\ref{eq:cn1.1part2}))}\\
&\leq{} 2\cdot{} \sum_{z=1,2,\ldots{},k_{opt}} w^* (z) r_z.\\
\end{aligned}
\end{equation*}

Case 1.2: $w^* (k_{opt})=\ceil*{\frac{S(H_{k_{opt}})}{g_{k_{opt}}}}$.

By the definition of $\cn_{z^\diamond}$, $\cn_{z^\diamond}(z^\diamond)=w^*(k_{opt})$ and also, for any $z_0\in T(z^\diamond)$, $\sum_{z\in T(z^\diamond) \land z\geq{} z_0} S(H_z)=\sum_{z\in T(z^\diamond) \land z\geq{} z_0} \cn_{z^\diamond}(z) g_z$ if\\ $\sum_{z\in T(z^\diamond) \land z\geq{} z_0} \cn_{z^\diamond}(z) g_z > \cn_{z^\diamond}(z^\diamond) g_{z^\diamond}$. Therefore, due to the feasibility of $w^*$, we have $\sum_{z\in T(z^\diamond) \land z\geq{} z_0} \cn_{z^\diamond}(z) g_z \leq \sum_{z\geq v(z_0)} w^*(z) g_z$ for each $z_0\in T(z^\diamond)$. Suppose that $T(z^\diamond)\setminus{}\{z^\diamond\}=\{i_1,i_2,\ldots,i_{n-1}\}$ where $i_1<i_2<\ldots<i_{n-1}$. Similar to the discussion in Case 1.1, by Proposition \ref{pro:T_ratio}, Lemma \ref{lem:feasibility} and the definition of the cost-per-capacity graph, we have $\sum_{z\in T(z^\diamond)} \cn_{z^\diamond}(z) r_z \leq \left(\sum_{v(i_1)\leq z <v(i_2)} w^*(z) g_z\right) \cdot \frac{r_{i_1}}{g_{i_1}} + \left(\sum_{v(i_2)\leq z <v(i_3)} w^*(z) g_z \right) \cdot \frac{r_{i_2}}{g_{i_2}}+ \ldots + \left(\sum_{v(i_{n-1})\leq z < v(z^\diamond)} w^*(z) g_z\right) \cdot \frac{r_{i_{n-1}}}{g_{i_{n-1}}} + \left(\sum_{v(z^\diamond)\leq z \leq z^\diamond} w^*(z) g_z\right) \cdot \frac{r_{z^\diamond}}{g_{z^\diamond}}
\leq \sum_{z=1,2,\ldots,k_{opt}} w^*(z) r_z$.

Case 2: $z^\diamond<k_{opt}$.

By claim (1) of Proposition \ref{pro:alternative},
\begin{equation} \label{eq:cn2part1}
\sum_{z\in{}T(z^\diamond)\cap A_{k_{opt}}} \cn_{z^\diamond}(z) r_z \leq{} r_{k_{opt}}.
\end{equation}

Next, we show that
\begin{equation} \label{eq:cn2part2}
\sum_{z\in{}T(z^\diamond)\setminus{}A_{k_{opt}}} \cn_{z^\diamond}(z) r_z \leq{} \sum_{z=1,2,\ldots{},m}w^*(z) r_z.
\end{equation}

Similar to the discussion in Case 1.1, we have $\sum_{z\in{}T(z^\diamond )\setminus{}A_{k_{opt}}} \cn_{z^\diamond}(z) r_z
\leq{} \sum_{z\in{}T(z^\diamond )\setminus{}A_{k_{opt}}} S(H_z) \frac{r_z}{g_z}$. Suppose that $T(z^\diamond )\setminus{}A_{k_{opt}}=\{i_1,i_2,\ldots,i_a\}$ where $i_1<i_2<\ldots<i_a$. By the feasibility of $w^*$, Proposition \ref{pro:T_ratio}, Lemma \ref{lem:feasibility} and the definition of the cost-per-capacity graph, we have $\sum_{z\in T(z^\diamond)} \cn_{z^\diamond}(z) r_z \leq \left(\sum_{v(i_1)\leq z <v(i_2)} w^*(z) g_z\right)\cdot \frac{r_{i_1}}{g_{i_1}} + \left(\sum_{v(i_2)\leq z <v(i_3)} w^*(z) g_z\right)\cdot \frac{r_{i_2}}{g_{i_2}} + \ldots + \left(\sum_{v(i_a)\leq z < v(k_{opt})} w^*(z) g_z\right) \cdot \frac{r_{i_a}}{g_{i_a}} + \left(\sum_{v(k_{opt})\leq z \leq k_{opt}} w^*(z) g_z\right) \cdot \frac{r_{k_{opt}}}{g_{k_{opt}}}
\leq \sum_{z=1,2,\ldots,k_{opt}} w^*(z) r_z$.

In summary,
\begin{equation*}
\begin{aligned}
\sum_{z\in{}T(z^\diamond)} \cn_{z^\diamond}(z) r_z &= \sum_{z\in{}T(z^\diamond)\setminus{}A_{k_{opt}}} \cn_{z^\diamond}(z) r_z + \sum_{z\in{}T(z^\diamond)\cap{}A_{k_{opt}}} \cn_{z^\diamond}(z) r_z\\
&\leq{} \sum_{z=1,2,\ldots{},k_{opt}} w^* (z) r_z+r_{k_{opt}} \quad \text{(by~equations~(\ref{eq:cn2part1}),(\ref{eq:cn2part2}))}\\
&\leq{} 2\cdot{}\sum_{z=1,2,\ldots{},k_{opt}} w^* (z) r_z.\\
\end{aligned}
\end{equation*}

Case 3: $z^\diamond>k_{opt}$.

By Theorem \ref{thm:optmachine}, $\sum_{z\in{}A_{z^\diamond}}{w^*(z)r_z} < r_{z^\diamond}$. Consequently,\\ $\sum_{z\in{}A_{z^\diamond}}{w^*(z)g_z}
= \sum_{z\in{}A_{z^\diamond}}{w^*(z)r_z\cdot{}\frac{g_z}{r_z}}
\leq{} \sum_{z\in{}A_{z^\diamond}}{w^*(z)r_z\cdot{}\frac{g_{z^\diamond}}{r_{z^\diamond}}}
< r_{z^\diamond} \cdot \frac{g_{z^\diamond}}{r_{z^\diamond}} = g_{z^\diamond}$.
By the feasibility of $w^*$, $S(H_{z^\diamond})=
S(\{J \in \mathcal{J}^{1d}:m(J)\geq{}v(z^\diamond) \})
\leq{}\sum_{z\in{}A_{z^\diamond}}{w^*(z)g_z}$. Therefore, $S(H_{z^\diamond}) < g_{z^\diamond}$. By the definition of $\cn_{z^\diamond}$, we have
\begin{equation} \label{eq:cn3part1}
 \cn_{z^\diamond}(z^\diamond)=1.
\end{equation}

Consider $z':=\max \{z\in P(k_{opt}): z<z^\diamond\}$ which clearly exists. By claim (2) of Proposition \ref{pro:alternative}, $\sum_{z\in{}T(z')\cap{}A_{z^\triangle}}{\cn_{z'}(z) r_z} > r_{z^\triangle} \geq{} r_{z^\diamond}$ for some $z^\triangle\in P(z')\setminus\{z'\}$.
By Proposition \ref{pro:T_two} and
Proposition \ref{pro:cn&r'_total}, $\sum_{z\in{}T(z')\cap{}A_{z^\triangle}}{S(H_z) \frac{r_z}{g_z}} + r_{z'} \geq{} \sum_{z\in{}T(z')\cap{}A_{z^\triangle}}{\cn_{z'}(z) r_z} \geq{} r_{z^\diamond}$.
This implies that $\sum_{z\in{}T(z') \cap{} A_{z^\triangle}}{S(H_z) \frac{r_z}{g_z}} \geq{} r_{z^\diamond} - r_{z'} \geq{} \frac{7}{8} r_{z^\diamond}$.
Meanwhile, by the feasibility of $w^*$, Proposition \ref{pro:T_ratio}, Lemma \ref{lem:feasibility} and the definition of the cost-per-capacity graph, $\sum_{z=1,2,\ldots{},k_{opt}} {w^*(z) r_z} \geq{} \sum_{z\in{}T(z')\cap{}A_{z^\triangle}} {S(H_z)\frac{r_z}{g_z}}$.
Therefore,
\begin{equation} \label{eq:cn3part2}
\frac{7}{8} r_{z^\diamond} \leq{} \sum_{z=1,2,\ldots{},k_{opt}} w^*(z)r_z.
\end{equation}

On the other hand, note that by Proposition \ref{pro:T_two}, $T(z^\diamond)\setminus\{z^\diamond\}\subset T(k_{opt})$. Similar to the arguments for equation (\ref{eq:cn2part2}), we also have
\begin{equation} \label{eq:cn3part3}
\sum_{z\in{}T(z^\diamond)\setminus\{z^\diamond\}} \cn_{z^\diamond}(z) r_z \leq{} \sum_{z=1,2,\ldots{},k_{opt}} w^*(z)r_z.
\end{equation}

In summary,
\begin{equation*}
\begin{aligned}
\sum_{z\in T(z^\diamond)} \cn_{z^\diamond}(z) r_z
&= r_{z^\diamond} + \sum_{z\in{}T(z^\diamond)\setminus{}\{z^\diamond\}} \cn_{z^\diamond}(z) r_z \quad \text{(by~equation~(\ref{eq:cn3part1}))}\\
&\leq{} \frac{8}{7}\cdot{}\sum_{z=1,2,\ldots{},k_{opt}} w^*(z) r_z + \sum_{z=1,2,\ldots{},k_{opt}} w^*(z) r_z\\
&\quad \text{~(by~equations~(\ref{eq:cn3part2}),(\ref{eq:cn3part3}))}\\
&= \frac{15}{7}\cdot{}\sum_{z=1,2,\ldots,k_{opt}} w^*(z) r_z.\\
\end{aligned}
\end{equation*}

\end{proof}

\section{The online setting}

\subsection {About one-shot job scheduling} \label{sec:one_shot}

In this subsection, we keep denoting the given set of jobs for the optimization problem (\ref{opt:1d}) by $\mathcal{J}^{1d}$. In the following, we point out one special explicit way of scheduling the jobs in $\mathcal{J}^{1d}$ into any feasible machine configuration. Subsequently, for each job in $\mathcal{J}^{1d}$, its placement actually indicates a cost for each individual job in $\mathcal{J}^{1d}$.

Suppose that $w$ represents a feasible machine configuration. Sort all the machines in $w$ in the descending order of capacity. Also, sort the jobs in $\mathcal{J}^{1d}$ in the descending order of job size. Schedule the (divisible) jobs into the first available machine one by one until the machine is full or all parts of the jobs in $\mathcal{J}^{1d}$ are scheduled. In the latter case, the scheduling is done. If the former case happens, continue to schedule the rest parts of jobs into the next available machine. The feasibility of $w$ guarantees that there always exists the next available machine. Observe that each job can be scheduled into at most two distinct machines. For each $J\in{}\mathcal{J}^{1d}$, define $r^* (\mathcal{J}^{1d},w)(J):=s_1 \cdot \frac{r_{z_1}}{g_{z_1}} + s_2 \cdot \frac{r_{z_2}}{g_{z_2}}$, where $s(J)=s_1+s_2$, where $s_1$ part of $J$ is scheduled into some machine in type-$z_1$ and that $s_2$ part of $J$ is scheduled into some machine in type-$z_2$. The feasibility of $w$ also guarantees that $s(J)\leq{}g_{z_1}$ and $s(J)\leq{}g_{z_2}$.

\begin{remark} \label{rem:optone}
(1). Suppose that $w$ is any machine configuration feasible to (\ref{opt:1d}) and $k_{opt}:=\max \{z:w(z)>0\}$. For any job $J\in{}H_z=\{J\in \mathcal{J}^{1d}:m(J)\in A_z\}$, where $z$ is any machine type from $T(k_{opt})$, we show that $r^* (\mathcal{J}^{1d},w)(J)\geq{}s(J)\cdot \frac{r_z}{g_z}$.

Suppose $J$ is divided into two parts which are scheduled into some type $z_1$ machine and some $z_2$ machine, without loss of generality. By the feasibility of $w$, $z_1 \geq m(J) \geq v(z)$ and $z_2 \geq m(J) \geq v(z)$. On the other hand, by proposition \ref{pro:T_union} and \ref{pro:T_ratio}, for each $i\in \{v(z),v(z)+1,\ldots,k_{opt}\}$, $\frac{r_i}{g_i} \geq \frac{r_z}{g_z}$. Therefore, $r^* (\mathcal{J}^{1d},w)(J):=s_1\cdot  \frac{r_{z_1}}{g_{z_1}} + s_2\cdot \frac{r_{z_2}}{g_{z_2}} \geq s(J)\cdot \frac{r_z}{g_z}$.

(2). $\sum_{J\in \mathcal{J}^{1d}} r^*(\mathcal{J}^{1d}, w)(J) \leq \sum_{z\in \mathcal{M}} w(z)\cdot r_z$.
\end{remark}

\subsection {Proof of Theorem \ref{thm:N_leq_opt}} \label{sec:N_leq_opt}

\begin{lemma} \label{lem:online_t}
For each machine type $z\in \mathcal{M}$ and for each time $t\in \sspan(\mathcal{J})$, we have $\sum_{i\in A_z\setminus \{z\}} N(i,t) r_i < r_z$.
\end{lemma}

\begin{proof} [Proof of lemma \ref{lem:online_t}]
This lemma can be easily proven by contradiction. Assume that at some time $t$ there exists some machine type $z_0$ such that $\sum_{i\in A_{z_0}\setminus \{z_0\}} N(i,t) r_i \geq r_{z_0}$.
Choose the last opened machine among all the machines being open at time $t$ and whose machine type is from $A_{z_0}\setminus \{z_0\}$, say machine $m_1$ whose machine type is $z_1\in A_{z_0}\setminus \{z_0\}$. (By assumption of the online setting, there is only one such machine) Call the job arriving at the same time as the starting time of the machine $m_1$ by $J_1$. (by the same assumption, there is only one such job. )
We have $I(J_1 )^- \leq t$. By the choice of $m_1$, we know that $\sum_{i\in A_{z_0}\setminus \{z_0\}} N(i,I(J_1 )^- ) r_i \geq \sum_{i\in A_{z_0 }\setminus \{z_0\}} N(i,t) r_i \geq r_{z_0}$, where the first inequality is because every machine being open at time t from the machine types $A_{z_0} \setminus \{z_0\}$ must also be open at $I(J_1 )^-$.
Also note that since $J_1$ is scheduled into a new machine $m_1$ of type $z_1\in A_{z_0}\setminus \{z_0\}$, we have $N(i,I(J)^-)=n_i$ for each $i\in A_{z_0} \setminus \{z_0,z_1\}$, and $N(z_1,I(J)^-)=n_{z_1}+1$, where $n_{x}$ for each $x\in A_{z_0}\setminus \{z_0\}$ is the number of type $x$ machines being open when $J_1$ is arriving but has not been scheduled yet. Therefore, $\sum_{i\in A_{z_0}\setminus \{z_0\}} n_i\cdot r_i \geq r_{z_0} - r_{z_1}$ which contradicts to the condition defined in the online algorithm $ALG_{online}$. (line 6)

\end{proof}

\begin{proof} [Proof of theorem \ref{thm:N_leq_opt}]

Take any time $t\in \sspan(\mathcal{J})$.
Denote $\max \{m(J) : J\in \mathcal{J}(t)\}$ by $k_0$.

In the following we briefly show that $k(t)\in P(k_0)$, because the arguments will be identical to the proof of property \ref{property:off_head}.
By the definition of $ALG_{online}$, for each $J\in \mathcal{J}(t)$, the machine type $J$ is scheduled into must be in $P(m(J))$. Therefore, all the possible machine types $J$ is scheduled into are $\{z\in \mathcal{M}: \exists i, \text{~s.t.~}, 1\leq i \leq k_0 \land z\in P(i)\}=\{1,2,\ldots,k_0\}\cup (P(k_0)\setminus \{k_0\})$. Since the jobs in $\{J\in \mathcal{J}(t): m(J)=k_0\}$ must be scheduled into some machine of type from $P(k_0)$, $k(t)$ must be in $P(k_0)$.

Denote by $w^* (z)$ with $z\in \mathcal{M}$ some optimal machine configuration for the one-shot scheduling of the jobs in $\mathcal{R}(t)\cup \mathcal{J}(t)$ chosen by theorem \ref{thm:optmachine}. Subsequently, let $k_{opt}:=\max \{z : w^* (z) > 0\}$ denote the highest machine type used by the chosen optimal machine configuration. We have $k_{opt}\in P(k_0)$.

Case 1: $k_{opt} \geq k(t)$.

By proposition \ref{pro:T_union}, $\{1,2,\ldots,k(t)\}\subset \{1,2,\ldots,k_{opt}\}= \cup_{z\in T(k_{opt})} A_z$. Note that if $N(k_{opt},t)>1$, then $k_{opt}=k(t)$.
For each $z\in I:=\{z\in T(k_{opt}) : N(z,t)>1\}$, denote the union of $K_z (t)$ and $\{J\in \mathcal{R}(t) : m(J)\in A_z \}$ by $L_z$. Clearly all such $L_z$'s are mutually disjoint.

We have
\begin{equation} \label{eq:N_leq_opt_case1_1}
\begin{aligned}
&\quad \optone (\mathcal{R}(t)\cup \mathcal{J}(t))\\
&= \sum_{z\in \mathcal{M}} w^*(z) \cdot r_z\\
&\geq \sum_{J\in \cup_{z\in I} L_z} r^*(\mathcal{R}(t)\cup \mathcal{J}(t),w^*) (J), \text{~by~remark~\ref{rem:optone}~(2)}\\
&\geq \sum_{J\in \cup_{z\in I} L_z} s(J)\cdot \frac{r_z}{g_z},\text{~by~remark~\ref{rem:optone}~(1)}\\
&\geq \sum_{z\in I} (N(z,t)-1)\cdot r_z, \text{~by~lemma~\ref{thm:R_case2}}\\
\end{aligned}
\end{equation}

Consider the following sequence of inequalities:
\begin{equation} \label{eq:N_leq_opt_case1_2}
\begin{aligned}
&\quad \sum_{i=1,2,\ldots,k(t)} N(i,t)\cdot r_i\\
&= \sum_{i\in A_{k_{opt}}} N(i,t)\cdot r_i + \sum_{z\in T(k_{opt})\setminus \{k_{opt}\}} \sum_{i\in A_z} N(i,t)\cdot r_i\\
&\leq (N(k_{opt},t) + 1)\cdot r_{k_{opt}} + \sum_{z\in T(k_{opt})\setminus \{k_{opt}\}} (N(z,t) + 1)\cdot r_{z}, \text{~by~lemma~\ref{lem:online_t}}\\
&\leq (N(k_{opt},t) + 1)\cdot r_{k_{opt}} + \frac{2}{7}\cdot r_{k_{opt}} + \sum_{z\in I\setminus \{k_{opt}\}} (N(z,t) - 1)\cdot r_{z}, \\
\end{aligned}
\end{equation}
where we have used the assumption which says that the cost rate of each machine type is a power of $8$.

Note that $r_{k_{opt}} \leq \optone (\mathcal{R}(t) \cup \mathcal{J}(t))$.
No matter $N(k_{opt},t) \leq 1$ or $N(k_{opt} ,t) \geq 2$, the right end of equation (\ref{eq:N_leq_opt_case1_2}) is bounded by $(3+\frac{2}{7})\cdot \optone (\mathcal{R}(t) \cup \mathcal{J}(t))$ by using equation (\ref{eq:N_leq_opt_case1_1}).

Therefore, in case 1, we have $\sum_{i=1,2,\ldots,k(t)} N(i,t)\cdot r_i \leq (3+\frac{2}{7})\cdot \optone (\mathcal{R}(t) \cup \mathcal{J}(t))$. Clearly, $3+\frac{2}{7} \leq 5$.

Case 2: $k_{opt} < k(t)$

Firstly, we show that only the case that $N(k(t),t)=1$ and $\forall J\in K_{k(t)}(t)$, $m(J) < k(t)$ will happen.
Since $k_0 \leq k_{opt} < k(t)$, we have $\forall J\in \mathcal{J}(t)$, $m(J) < k(t)$. For the sake of contradiction, assume that $N(k(t),t) \geq 2$. Lemma \ref{thm:R_case2} and the definition of $ALG_{online}$ implies that
\begin{equation} \label{eq:N_leq_opt_case2_1}
\begin{aligned}
&\quad S(\{J\in{}\mathcal{R}(t)\cup \mathcal{J}(t):m(J)\in{}A_{k(t)}\})\\
&\geq S(\{J\in{}\mathcal{R}(t):m(J)\in{}A_{k(t)}\})+S(K_{k(t)},t)\\
&\geq{} (N(k(t),t) - 1)\cdot g_{k(t)}.\\
\end{aligned}
\end{equation}
However, the choice of the optimal machine configuration says that $\sum_{z\in A_{k(t)}} w^*(z) \cdot r_z < r_{k(t)}$. Since the node $k(t)$ has the lowest cost rate among all the nodes in $A_{k(t)}$, immediately we have $g_{k(t)}>\sum_{z\in A_{k(t)}} w^*(z) \cdot g_z \geq S(\{J\in{}\mathcal{R}(t)\cup \mathcal{J}(t) : m(J)\in{}A_{k(t)}\})$, where the last inequality is due to the feasibility of $w^*$. This clearly contradicts to equation (\ref{eq:N_leq_opt_case2_1}).

Therefore, $N(k(t),t)=1$. By Lemma \ref{thm:R_case1}, take $\hat{J}\in K_{k(t)}(t)$. Since $m(\hat{J}) < k(t)$, $m(\hat{J})\in A_{f_1}$ for some $f_1\in f(k(t))$. By the definition of the online algorithm, if $\hat{J}$ were to be scheduled into some type $f_1$ machine, one new type $f_1$ machine must be opened for processing $\hat{J}$. However, although it is known that $\hat{J}$ is scheduled into some type $k(t)$ machine, it is not known that whether a new type $k(t)$ machine is opened for processing it or not. There are two cases in general:

Case 2.1: $\sum_{z\in A_{k(t)}\setminus \{k(t)\}} N(z,I(\hat{J})^-) \cdot r_z + r_{f_1} \geq r_{k(t)}$.

Case 2.2: there exists some $k^\triangle \in P(k(t))\setminus \{k(t)\}$ such that $\sum_{z\in A_{k^\triangle}\setminus \{k^\triangle, k(t)\}} N(z,I(\hat{J})^-) \cdot r_z + n_{k(t)}\cdot r_{k(t)} + r_{f_1} \geq r_{k^\triangle}$, where $n_{k(t)}$ is either $N(k(t),I(\hat{J})^-)$ or $N(k(t),I(\hat{J})^-) - 1$. (see line 6 of the definition of $ALG_{online}$)

Case 2.1. In this case, we have
\begin{equation} \label{eq:N_leq_opt_case2.1_1}
\begin{aligned}
&\quad \sum_{z\in f(k(t)) \land N(z,I(\hat{J})^-) > 1} (N(z,I(\hat{J})^-) - 1)\cdot r_z\\
&\geq r_{k(t)} - r_{f_1} - \sum_{z\in f(k(t))} \left(r_z + \sum_{i\in A_z\setminus \{z\}} N(i,I(\hat{J})^-)\cdot r_i\right) \\
&\geq r_{k(t)} - r_{f_1} - \sum_{z\in f(k(t))} 2\cdot r_z, \text{~by~lemma~\ref{lem:online_t}} \\
&\geq (1-\frac{1}{8}-\frac{2}{7})\cdot r_{k(t)}, \\
\end{aligned}
\end{equation}
where in the last inequality we have used the assumption that the cost rate for each machine type is a power of $8$.

Consider those $z\in f(k(t))$ such that $N(z,I(\hat{J} )^- ) > 1$ (corresponding to the left-hand side of the above equation (\ref{eq:N_leq_opt_case2.1_1})). For each $J\in L_z:=K_z (t)\cup \{\mathcal{R}(t) : m(J)\in A_z \}$, we have
\begin{equation} \label{eq:N_leq_opt_case2.1_2}
\quad r^* (\mathcal{R}(t)\cup \mathcal{J}(t),w^* )(J) \geq s(J)\cdot \frac{r_{z_0}}{g_{z_0}} \geq s(J)\cdot \frac{r_z}{g_z},
\end{equation}
where $z_0$ is the node in $T(k_{opt})$ such that $m(J)\in A_{z_0}$. The first inequality is by remark \ref{rem:optone} (1), and the second is because $k_{opt}<k(t)$ implies that $z_0\in A_z$.

We have the following inequalities hold:
\begin{equation*}
\begin{aligned}
&\quad \optone(\mathcal{R}(t)\cup \mathcal{J}(t)) \\
&\geq \sum_{J\in \cup_{z\in f(k(t)) \land N(z,I(\hat{J})^- ) > 1} L_z} r^* (\mathcal{R}(t)\cup \mathcal{J}(t) ,w^* )(J),\\ &\quad \text{~by~remark~\ref{rem:optone}~(2)}\\
&\geq \sum_{z\in f(k(t)) \land N(z,I(\hat{J})^- ) > 1} S(L_z ) \cdot \frac{ r_z}{g_z}, \text{~by~equation~(\ref{eq:N_leq_opt_case2.1_2})}\\
&\geq \sum_{z\in f(k(t)) \land N(z,I(\hat{J})^- )>1} (N(z,I(\hat{J})^- )-1)\cdot r_z,\text{~by~lemma~\ref{thm:R_case2}}\\
&\geq (1-\frac{1}{8}-\frac{2}{7})\cdot r_{k(t)}, \text{~by~equation~(\ref{eq:N_leq_opt_case2.1_1})}\\
\end{aligned}
\end{equation*}

Eventually, we have $\optone(\mathcal{R}(t)\cup \mathcal{J}(t)) \geq (1-\frac{1}{8}-\frac{2}{7})\cdot r_{k(t)}$.

Case 2.2.
At first, we show that $\hat{J}$ was scheduled into some type $k(t)$ machine which has been opened before $\hat{J}$ arrives, i.e., $n_{k(t)} = N(k(t),I(\hat{J} )^- )$. For the sake of contradiction, assume that a new type $k(t)$ machine was opened for scheduling $\hat{J}$ at time $I(\hat{J} )^-$. In this case, we have $\sum_{z\in A_{k^\triangle}\setminus\{k^\triangle\}} N(z,I(\hat{J} )^- ) r_z - r_{k(t)} + r_{f_1} \geq r_{k^\triangle}$. This implies that $\sum_{z\in A_{k^\triangle}\setminus\{k^\triangle\}} N(z,I(\hat{J} )^- ) r_z \geq r_{k^\triangle}$ which contradicts to lemma \ref{lem:online_t}.

It is easy to see that $k(t)-1=\max f(k(t))$. Denote $T(k(t)-1)\cap A_{k^\triangle}$ by $H$. By proposition \ref{pro:consecutive} and proposition \ref{pro:T_union}, we have $\cup_{z\in H} A_z = \{v(k^\triangle),v(k^\triangle)+1,\ldots,k(t)-1\}$. Now, the condition of the case 2.2 says that $\sum_{z\in A_{k^\triangle}\setminus \{k^\triangle\}} N(z,I(\hat{J})^-) \cdot r_z + r_{f_1} \geq r_{k^\triangle}$ which implies that
\begin{equation} \label{eq:N_leq_opt_case2.2_1}
\begin{aligned}
&\quad \sum_{z\in H \land N(z,I(\hat{J})^-)>1} (N(z,I(\hat{J})^-)-1)\cdot r_z \\
&\quad + \sum_{z\in H} \left(r_z + \sum_{i\in A_z\setminus \{z\}} N(i,I(\hat{J})^-)\cdot r_i \right) + r_{f_1}\\
&\geq r_{k^\triangle} - \sum_{z=k(t),k(t)+1,\ldots,k^\triangle-1} N(z,I(\hat{J})^-)\cdot r_z.\\
\end{aligned}
\end{equation}
For the right-hand side of the above equation (\ref{eq:N_leq_opt_case2.2_1}), by lemma \ref{lem:online_t}, $r_{k^\triangle} - \sum_{z=k(t),k(t)+1,\ldots,k^\triangle-1} N(z,I(\hat{J})^-)\cdot r_z$ must be positive. Furthermore, after dividing it by $r_{k(t)}$, the result is a positive integer. Therefore, $r_{k^\triangle} - \sum_{z=k(t),k(t)+1,\ldots,k^\triangle-1} N(z,I(\hat{J})^-)\cdot r_z$ must be at least $1\cdot r_{k(t)}$. By equation (\ref{eq:N_leq_opt_case2.2_1}), we have the following
\begin{equation*}
\begin{aligned}
&\quad \sum_{z\in H \land N(z,I(\hat{J})^-)>1} (N(z,I(\hat{J})^-)-1)\cdot r_z \\
&\quad + \sum_{z\in H} \left(r_z + \sum_{i\in A_z\setminus \{z\}} N(i,I(\hat{J})^-)\cdot r_i \right) + r_{f_1} \\
&\geq r_{k(t)}, \\
\end{aligned}
\end{equation*}
and consequently,
\begin{equation} \label{eq:N_leq_opt_case2.2_2}
\begin{aligned}
&\quad \sum_{z\in H \land N(z,I(\hat{J})^-)>1} (N(z,I(\hat{J})^-)-1)\cdot r_z\\
&\geq r_{k(t)} - r_{f_1} - \sum_{z\in H} \left(r_z + \sum_{i\in A_z\setminus \{z\}} N(i,I(\hat{J})^-)\cdot r_i \right) \\
&\geq r_{k(t)} - \frac{1}{8} \cdot r_{k(t)} - \sum_{z\in H} 2\cdot r_z\\
&\geq (1-\frac{1}{8} - \frac{2}{7})\cdot r_{k(t)}. \\
\end{aligned}
\end{equation}

Similarly as case 2.1, let $L_z:=K_z (t)\cup \{\mathcal{R}(t) : m(J)\in A_z \}$ for each $z\in H$ such that $N(z,I(\hat{J})^-) > 1$. we have the inequalities:
\begin{equation*}
\begin{aligned}
&\quad \optone (\mathcal{R}(t)\cup \mathcal{J}(t)) \\
&\geq \sum_{J\in \cup_{z\in H \land N(z,I(\hat{J} )^- )>1} L_z} r^* (\mathcal{R}(t)\cup \mathcal{J}(t),w^* )(J)\\
&\geq \sum_{z\in H \land N(z,I(\hat{J} )^- )>1} S(L_z) \cdot \frac{r_z}{g_z}, \text{~since~} k_{opt} < k(t) \\
&\geq \sum_{z\in H \land N(z,I(\hat{J} )^- ) > 1} (N(z,I(\hat{J} )^- )-1)\cdot r_z, \text{~by~lemmas~\ref{thm:R_case2}~and~\ref{thm:R_case1}} \\
&\geq  (1-\frac{1}{8} - \frac{2}{7})\cdot r_{k(t)}, \text{~by~equation~(\ref{eq:N_leq_opt_case2.2_2})}. \\
\end{aligned}
\end{equation*}

In summary, all the above discussion for the two subcases 2.1 and 2.2 has shown that
\begin{equation} \label{eq:N_leq_opt_case2_2}
\optone (\mathcal{R}(t)\cup \mathcal{J}(t)) \geq (1-\frac{1}{8} - \frac{2}{7})\cdot r_{k(t)}
\end{equation}

Also let $L_{z}:=K_{z} (t)\cup \{J\in \mathcal{R}(t) : m(J)\in A_{z} \}$ for each $z\in \{z\in T(k(t)) : N(z,t)>1\}$. Since $k_{opt} < k(t)$, we have $\cup_{z\in T(k_{opt})} A_z = \{1,2,\ldots,k_{opt}\} \subset \{1,2,\ldots,k(t)\} = \cup_{z\in T(k(t))} A_z$. For each $z\in \{z\in T(k(t)) : N(z,t)>1\}$, for each job $J\in L_z$, $r^* (\mathcal{R}(t)\cup \mathcal{J}(t),w^* )(J) \geq s(J)\cdot \frac{r_{z}}{g_{z}}$. Therefore,
\begin{equation} \label{eq:N_leq_opt_case2_3}
\begin{aligned}
&\quad \optone (\mathcal{R}(t)\cup \mathcal{J}(t))\\
&\geq \sum_{J\in \cup_{z\in T(k(t)) \land N(z,t )>1} L_z} r^* (\mathcal{R}(t)\cup \mathcal{J}(t),w^* )(J)\\
&\geq \sum_{z\in T(k(t)) \land N(z,t) > 1} S(L_z)  \cdot \frac{r_z}{g_z}\\
&\geq \sum_{z\in T(k(t)) \land N(z,t) > 1} (N(z,t) - 1)\cdot r_{z}.\\
\end{aligned}
\end{equation}

At last, we have
\begin{equation*}
\begin{aligned}
&\quad \sum_{z=1,2,\ldots,k(t)} N(z,t)\cdot r_z \\
&\leq \sum_{z\in T(k(t)) \land N(z,t) > 1} (N(z,t) - 1)\cdot r_z \\
&+ \sum_{z\in T(k(t))} \left(r_z + \sum_{i\in A_z\setminus \{z\}} N(i,t)\cdot r_i \right) \\
&\leq \optone (\mathcal{R}(t)\cup \mathcal{J}(t)) + \sum_{z\in T(k(t))} 2\cdot r_z,\\
&\quad \text{~by~equation~(\ref{eq:N_leq_opt_case2_3})~and~lemma~\ref{lem:online_t}}\\
&\leq \optone (\mathcal{R}(t)\cup \mathcal{J}(t)) + (2+\frac{2}{7}) \cdot r_{k(t)}\\
&\leq \left( 1 + \frac{2 + \frac{2}{7}} {1 - \frac{1}{8} - \frac{2}{7}} \right)\cdot \optone (\mathcal{R}(t)\cup \mathcal{J}(t)), \text{~by~equation~(\ref{eq:N_leq_opt_case2_2})}. \\
\end{aligned}
\end{equation*}

Clearly, $1 + \frac{2 + \frac{2}{7}} {1 - \frac{1}{8} - \frac{2}{7}} \leq 1 + \frac{2 + \frac{2}{7}} {1 - \frac{3}{7}} = 5$. We have finished the discussion of case 2.

\end{proof}

\subsection {Proof of Theorem \ref{thm:tilder_1}} \label{sec:tilder_1}

For ease of reference,
Figure \ref{fig:rtilde} summarizes the costs  $\tilde{r}(\mathcal{J}^{1d})(J)$ charged on individual jobs in different cases.

\begin{figure}[h]
  \centering
  \includegraphics[width=18cm]{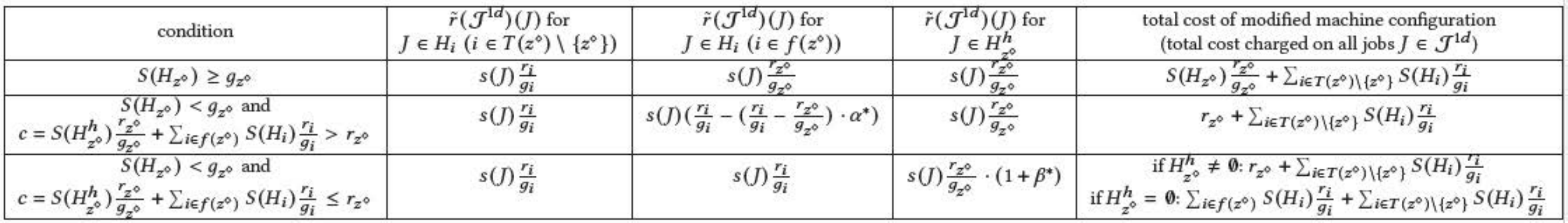}
  \caption{Definition of $\tilde{r}(\mathcal{J}^{1d})(J)$}
  \label{fig:rtilde}
\end{figure}

In this section, we prove Theorem \ref{thm:tilder_1}. Suppose $\mathcal{X}$ and $\mathcal{Y}$ are two sets of jobs such that $\mathcal{X}\subset \mathcal{Y}$. Let $h_0:=\max \{m(J): J\in \mathcal{X}\}$ and $k_0:=\max \{m(J):J\in \mathcal{Y}\}$. Furthermore, let $\cn_{z^\diamond_1}^{\mathcal{X}}$ and $\cn_{z^\diamond_2}^{\mathcal{Y}}$ denote the alternative machine configurations for the one-shot scheduling $\mathcal{X}$ and $\mathcal{Y}$ respectively. By definition, $z^\diamond_1\in P(h_0)$ and $z^\diamond_2\in P(k_0)$.

\begin{lemma} \label{lem:compare_0}

Take any $k_1\in P(k_0)$ arbitrarily. There exists $h_1\in T(k_1)$ such that $h_0\in A_{h_1}$. Furthermore, for each $z_0\in T(h_1)$, we have $\sum_{z\in T(h_1) \land z\geq z_0} \cn_{h_1}^{\mathcal{X}} (z) \cdot r_z \leq \sum_{z\in T(k_1) \land z\geq z_0} \cn_{k_1}^{\mathcal{Y}} (z) \cdot r_z$.

\end{lemma}

\begin{proof}

For the first part, $\mathcal{X}\subset \mathcal{Y}$ implies that $h_0 \leq k_0 \leq k_1$. By proposition \ref{pro:T_union}, $h_0\in A_{h_1}$ for some $h_1\in T(k_1)$.
For the second part, it suffices to prove the special case of lemma \ref{lem:compare_0} when $\mathcal{Y}\setminus \mathcal{X}$ is singleton. Hence, assume that $\mathcal{Y}\setminus \mathcal{X} = \{J_0\}$.

Case 1: $\ceil*{\frac{S({J\in \mathcal{X} : m(J)\in A_{k_1} })}{g_{k_1}}} = \ceil*{\frac{S({J\in \mathcal{X}\cup\{J_0\} : m(J)\in A_{k_1} })}{g_{k_1}}}$.

In this case, clearly $h_1 = k_1$. No matter where $m(J_0)$ is within $\cup_{z\in T(k_1)} A_z$,  $\cn_{k_1}^{\mathcal{X}} (z) \leq \cn_{k_1}^\mathcal{Y} (z)$ for each $z\in T(k_1 )\setminus \{k_1\}$. Therefore, we are done.

Case 2:  $\ceil*{\frac{S({J\in \mathcal{X} : m(J)\in A_{k_1} })}{g_{k_1}}} + 1 = \ceil*{\frac{S({J\in \mathcal{X}\cup\{J_0\} : m(J)\in A_{k_1} })}{g_{k_1}}}$.

There are two subcases about $\{J\in \mathcal{X} : m(J)\in A_{k_1} \}$.

Case 2.1: $\ceil*{\frac{S(\{J\in \mathcal{X} : m(J)\in A_{k_1} \})}{g_{k_1}}}=0$, and
Case2.2: $\ceil*{\frac{S(\{J\in \mathcal{X} : m(J)\in A_{k_1} \})}{g_{k_1}}} \geq 1$. For simplicity of notations, for each $z\in T(k_1)\setminus \{k_1\}$, let $H_z$ denote the set of jobs $\{J\in \mathcal{X}: m(J)\in A_z\}$.

In case 2.1. $\ceil*{\frac{S(\{J\in \mathcal{X} : m(J)\in A_{k_1} \})}{g_{k_1}}}=0$ implies that $h_1<k_1$ must hold. Furthermore, $\{J_0\} = \{J\in \mathcal{X}\cup \{J_0\}: m(J)\in A_{k_1}\}$. We have the following inequalities:
\begin{equation*}
\begin{aligned}
&\quad \sum_{z\in T(h_1) \land z\geq z_0} \cn_{h_1}^{\mathcal{X}} (z)\cdot r_z\\
&\leq r_{h_1} + \sum_{z\in T(h_1) \land z\geq z_0} S(H_z)\cdot \frac{r_z}{g_z}, \text{~by~proposition~\ref{pro:cn&r'_total}}\\
&\leq \frac{s(J_0)}{g_{h_1}}\cdot r_{h_1} +  \sum_{z\in T(h_1) \land z\geq z_0} S(H_z)\cdot \frac{r_z}{g_z}, \text{~because~} m(J_0)\in A_{k_1}\\
&\leq s(J_0) \cdot \frac{r_{k_1}}{g_{k_1}} + \sum_{z\in T(h_1) \land z\geq z_0} S(H_z)\cdot \frac{r_z}{g_z}, \text{~by~proposition~\ref{pro:T_ratio}}\\
&\leq  s(J_0) \cdot \frac{r_{k_1}}{g_{k_1}} + \sum_{z\in T(k_1) \land z_0 \leq z < k_1} S(H_z)\cdot \frac{r_z}{g_z}\\
&\leq \sum_{z\in T(k_1) \land z\geq z_0} \cn_{k_1}^\mathcal{Y} (z)\cdot r_z, \text{~by~proposition~\ref{pro:cn&r'_head}}\\
\end{aligned}
\end{equation*}

In case 2.2, clearly $h_1=k_1$. And $m(J_0)$ must be in $A_{k_1}$. We have the following inequalities:

\begin{equation*}
\begin{aligned}
&\quad \sum_{z\in T(h_1) \land z\geq z_0} \cn_{h_1}^\mathcal{X} (z)\cdot r_z\\
&= \cn_{k_1}^\mathcal{X} (k_1) \cdot r_{k_1} + \sum_{z\in T(k_1) \land z_0 \leq z < k_1} \cn_{k_1}^\mathcal{X} (z)\cdot r_z\\
&\leq \cn_{k_1}^\mathcal{X} (k_1) \cdot r_{k_1} + \sum_{z\in T(k_1) \land z_0 \leq z < k_1} S(H_z)\cdot \frac{r_z}{g_z}, \\
&\quad \text{~since~} \cn_{k_1}^\mathcal{X} (z)\cdot g_z \leq S(H_z), \text{~for~each~} z\in T(k_1)\setminus \{k_1\}\\
&\leq (\cn_{k_1}^\mathcal{X} (k_1) + 1) \cdot r_{k_1} + \sum_{z\in T(k_1) \land z_0 \leq z < k_1} \cn_{k_1}^\mathcal{Y} (z) \cdot r_{z},\\
&\quad \text{~by~proposition~\ref{pro:cn&r'_total}}\\
&= \sum_{z\in T(k_1) \land z \geq z_0} \cn_{k_1}^\mathcal{Y} (z) \cdot r_{z}\\
\end{aligned}
\end{equation*}

\end{proof}

\begin{lemma} \label{lem:compare}
We have $z^\diamond_1 \in{} A_i$ for some $i\in{}T(z^\diamond_2)$. \textit{Equivalently}, we show that the highest machine type used by the alternative machine configuration for one-shot scheduling $\mathcal{X}$ is less than or equal to the highest machine type used by the alternative machine configuration for one-shot scheduling $\mathcal{Y}$. (see Proposition \ref{pro:T_union})
\end{lemma}

\begin{proof}

Recall that $h_0:=\max \{m(J): J\in \mathcal{X}\}$ and $k_0:=\max \{m(J):J\in \mathcal{Y}\}$. By lemma \ref{lem:compare_0}, $h_0\in A_{h_1}$ for some $h_1\in T(z^\diamond_2)$. To show that $z^\diamond_1\in A_{h_1}$, it suffices to show that $\sum_{z\in T(h_1)\cap A_{h_2}} \cn^\mathcal{X}_{h_1} (z)\cdot r_z \leq r_{h_2}$ for each $h_2\in P(h_1)\setminus \{h_1\}$.

For each $h_2\in P(h_1)\setminus \{h_1\}$, let $z_0=\min T(h_1)\cap A_{h_2}$. Indeed, we have
\begin{equation*}
\begin{aligned}
&\quad \sum_{z\in T(h_1)\cap A_{h_2}} \cn^\mathcal{X}_{h_1} (z)\cdot r_z \\
&= \sum_{z\in T(h_1) \land z\geq z_0} \cn^\mathcal{X}_{h_1} (z)\cdot r_z, \text{~by~Proposition~\ref{pro:T_two}} \\
&\leq \sum_{z\in T(z^\diamond _2) \land z\geq z_0} \cn^\mathcal{Y}_{z^\diamond_2}(z)\cdot r_z, \text{~by~lemma~\ref{lem:compare_0}} \\
&= \sum_{z\in T(z^\diamond _2)\cap A_{h_2}} \cn^{\mathcal{Y}}_{z^\diamond_2}(z)\cdot r_z, \text{~by~Proposition~\ref{pro:T_two}}\\
&\leq r_{h_2}, \text{~by~Proposition~\ref{pro:alternative}~(1)}.\\
\end{aligned}
\end{equation*}

\end{proof}

\begin{proof} [Proof of Theorem \ref{thm:tilder_1}]

By lemma \ref{lem:compare}, $z^\diamond_1\in{}A_{h_1}$ for some $h_1\in{}T(z^\diamond_2)$ with $P(h_1)\subset P(z^\diamond_1)\subset P(h_0)$.
Take any job $J\in{}\mathcal{X}\setminus{}\{J\in{}\mathcal{X}:m(J)\in{}A_{h_1}\}$.
Since $h_0\in{}A_{h_1}$, clearly $J\in{}H_z$ for some $z\in{}T(h_1)\setminus{}\{h_1\}$.
By definition, $T(h_1)\setminus{}\{h_1\}\subset{}T(z^\diamond_1)\setminus{}\{z^\diamond_1\}$. $\tilde{r}(\mathcal{X})(J)=s(J)\cdot \frac{r_z}{g_z}$ by definition.
On the other hand, again by definition, since $h_1\in T(z^\diamond_2)$, $T(h_1 )\setminus{}\{h_1\}\subset{}T(z^\diamond_2)\setminus{}\{z^\diamond_2\}$, which implies that $\tilde{r}(\mathcal{Y})(J)=s(J)\cdot \frac{r_z}{g_z}$.
Therefore, it suffices to consider the set of remaining jobs, i.e., $\{J\in{}\mathcal{X}:m(J)\in{}A_{h_1}\}$.

We show that it suffices to assume that $h_1=z^\diamond_2$.
Otherwise, $h_1\in{}T(z^\diamond_2)\setminus{}\{z^\diamond_2\}$. We have $\{1,2,\ldots{},h_0\}\cap{}A_{h_1}\subset{}\{1,2,\ldots{},z^\diamond_1\}\cap{}A_{h_1}=\cup_{z\in{}T(z^\diamond_1)\cap{}A_{h_1}} A_z$, where the inclusion is due to $h_0\leq{}z^\diamond_1$ and the equality is due to proposition \ref{pro:T_union}.
Therefore, $\{J\in{}\mathcal{X}:m(J)\in{}A_{h_1}\} = \{J\in{}\mathcal{X}:m(J)\in{}\cup{}_{z\in{}T(z^\diamond_1)\cap{}A_{h_1}} A_z \}$.
For each $i\in{}T(z^\diamond_1 )\cap{}A_{h_1}$, for each $J\in\{J\in{}\mathcal{X}:m(J)\in{}A_i\}$, we have $\tilde{r}(\mathcal{X})(J) \geq{} s(J)\cdot \frac{r_i}{g_i} \geq{} s(J)\cdot\frac{r_{h_1}}{g_{h_1}}$.
On the other hand, note that $h_1\in{}T(z^\diamond_2 )\setminus{}\{z^\diamond_2\}$.
Therefore, for each $J\in{}\{J\in{}\mathcal{X}:m(J)\in{}A_{h_1}\}$, $\tilde{r}(\mathcal{Y})(J)=s(J)\cdot\frac{r_{h_1}}{g_{h_1}}$.
Therefore, we have $\tilde{r}(\mathcal{X})(J) \geq{} \tilde{r}(\mathcal{Y})(J)$ for each $J\in\{J\in{}\mathcal{X}:m(J)\in{} A_{h_1}\}$.

Furthermore, it suffices to assume $z^\diamond_1=z^\diamond_2$.
This is because otherwise $z^\diamond_2\in P(z^\diamond_1)\setminus \{z^\diamond_1\}$ and hence $\{J\in{}\mathcal{X}:m(J)\in{}A_{z^\diamond_2}\}=\{J\in{}\mathcal{X}:m(J)\in{}A_{z^\diamond_2}\setminus{}\{z^\diamond_2\}\}=\cup{}_{i\in{}f(z^\diamond_2)} \{J\in{}\mathcal{X}:m(J)\in{}A_i\}$.
For each $i\in{}f(z^\diamond_2)$, for each $J\in{}\{J\in{}\mathcal{X}:m(J)\in{}A_i \}$, $\tilde{r}(\mathcal{X})(J) \geq{} s(J)\cdot\frac{r_i}{g_i}$ by definition.
On the other hand, for each $i\in{}f(z^\diamond_2)$, for each $J\in{}\{J\in{}\mathcal{X}:m(J)\in{}A_i\}$, $\tilde{r}(\mathcal{Y})(J) \leq{} s(J)\cdot\frac{r_i}{g_i}$.
Clearly, we must have $\tilde{r}(\mathcal{X})(J) \geq{} \tilde{r}(\mathcal{Y})(J)$ for each $J\in{}\{J\in{}\mathcal{X}:m(J)\in{}A_{z^\diamond_2}\}$.

Let $z^*=z^\diamond_2$.

Let $h:=S({J\in{}\mathcal{Y}:m(J)=z^* })\cdot  \frac{r_{z^*}}{g_{z^*}}$.

Let $c_1:=\sum_{i\in{}f(z^*)} S(\{J\in{}\mathcal{Y}:m(J)\in{}A_i\})\cdot \frac{r_i}{g_i}$.

Let  $c_2:=S(\{J\in{}\mathcal{Y}:m(J)\in{}A_{z^*}\setminus{}\{z^*\}\})\cdot \frac{r_{z^*}}{g_{z^*}}$.

Let $h':= S(\{J\in{}\mathcal{X}:m(J)=z^*\})\cdot \frac{r_{z^*}}{g_{z^*}}$.

Let $c'_1:=\sum_{i\in{}f(z^*)} S(\{J\in{}\mathcal{X}:m(J)\in{}A_i\})\cdot \frac{r_i}{g_i}$.

Let $c'_2:=S(\{J\in{}\mathcal{X}:m(J)\in{}A_{z^*}\setminus\{z^*\}\})\cdot \frac{r_{z^*}}{g_{z^*}}$.

Clearly, we have $h' \leq{} h$, $c'_1 \leq{} c_1$ and $c'_2 \leq{} c_2$.
In the following, we show that $\tilde{r}(\mathcal{X})(J) \geq{} \tilde{r}(\mathcal{Y})(J)$ for each $J\in{}\{J\in{}\mathcal{X}:m(J)\in{}A_{z^*}\}$ case by case.

Case 1: $h+c_1 \leq{} r_{z^*}$.

Consequently, $h'+c'_1 \leq{} r_{z^*}$.
By definition, for each $i\in{}f(z^*)$, for each $J\in{}\{J\in{}\mathcal{X}:m(J)\in{}A_i\}$, we have $\tilde{r}(\mathcal{X})(J)=\tilde{r}(\mathcal{Y})(J)=s(J)\cdot\frac{r_i}{g_i}$.
If $h' > 0$, which implies $h > 0$, for each $J\in{}\{J\in{}\mathcal{X}:m(J)=z^*\}$, $\tilde{r}(\mathcal{X})(J)=s(J)\cdot \frac{r_{z^*}}{g_{z^*}} \cdot{} \frac{r_{z^*}-c'_1}{h'} \geq{} \frac{r_{z^*}}{g_{z^*}} \cdot{} \frac{r_{z^*} - c'_1}{h} \geq{} \frac{r_{z^*}}{g_{z^*}}\cdot{} \frac{r_{z^*} - c_1}{h} = \tilde{r}(\mathcal{Y})(J)$.

Case 2: $h+c_1>r_{z^*}$ and $h+c_2<r_{z^*}$.

Consequently, we must have $h'+c'_2 < r_{z^*}$.

Case 2.1: $h' + c'_1 \leq{} r_{z^*}$ and $h' + c'_2 < r_{z^*}$.

For each $i\in{}f(z^*)$, for each $J\in{}\{J\in{}\mathcal{X}:m(J)\in{}A_i\}$, we have $\tilde{r}(\mathcal{X})(J) = s(J)\cdot \frac{r_i}{g_i} >s(J)\cdot \left(\frac{r_i}{g_i} +(\frac{r_{z^*}}{g_{z^*}} -\frac{r_i}{g_i}) \cdot{} \alpha^*\right)=\tilde{r}(\mathcal{Y})(J)$. For each $J\in{}\{J\in{}\mathcal{X}:m(J)=z^* \}$, $\tilde{r}(\mathcal{X})(J) \geq{} s(J)\cdot \frac{r_{z^*}}{g_{z^*}} =\tilde{r}(\mathcal{Y})(J)$ no matter whether $h'>0$ or not.

Case 2.2: $h'+c'_1 > r_{z^*}$ and $h'+c'_2 < r_{z^*}$.

For each $J\in{}\{J\in{}\mathcal{X}:m(J)=z^*\}$, $\tilde{r}(\mathcal{X})(J)=s(J)\cdot \frac{r_{z^*}}{g_{z^*}} =\tilde{r}(\mathcal{Y})(J)$ by definition.
It suffices to check $\{J\in{}\mathcal{X}:m(J)\in{}A_i\}$ for each $i\in{}f(z^* )$.
Define the linear function
\begin{equation}
l(\alpha)=\sum_{i\in{}f(z^*)} \sum_{J\in{}\{J\in{}\mathcal{Y}:m(J)\in{}A_i\}}\\
 s(J)\cdot{} \left( \frac{r_i}{g_i} +(\frac{r_{z^*}}{g_{z^*}} - \frac{r_i}{g_i} )\cdot{}\alpha \right)
\end{equation}
\begin{equation}
    l'(\beta)=\sum_{i\in{}f(z^*)} \sum_{J\in{}\{J\in{}\mathcal{X}:m(J)\in{}A_i\}} s(J)\cdot{} \left( \frac{r_i}{g_i} +(\frac{r_{z^*}}{g_{z^*}} - \frac{r_i}{g_i} )\cdot{}\beta \right)
\end{equation}
such that $l(\alpha^* )=r_{z^*} - h$ and $l' (\beta^*)=r_{z^*}-h'$.
For each $i\in{}f(z^*)$, for each $J\in{}\{J\in{}\mathcal{X}:m(J)\in{}A_i\}$, in order to show that $\tilde{r}(\mathcal{X})(J)= s(J)\cdot \left(\frac{r_i}{g_i} + (\frac{r_{z^*}}{g_{z^*}} - \frac{r_i}{g_i} )\cdot{}\beta^* \right) \geq{} s(J)\cdot \left( \frac{r_i}{g_i} +(\frac{r_{z^*}}{g_{z^*}} - \frac{r_i}{g_i} )\cdot{}\alpha^*\right) = \tilde{r}(\mathcal{Y})(J)$, it suffices to show that $\beta^* \leq{} \alpha^*$.

Look at
\begin{equation*}
\begin{aligned}
l'(\alpha^*) &= \sum_{i\in{}f(z^*)} \sum_{J\in{}\{J\in{}\mathcal{X}:m(J)\in{}A_i\}} s(J)\cdot{} \left( \frac{r_i}{g_i} + (\frac{r_{z^*}}{g_{z^*}} - \frac{r_i}{g_i} )\cdot{}\alpha^* \right) \\
&\leq{} \sum_{i\in{}f(z^* )} \sum_{J\in{}\{J\in{}\mathcal{Y}:m(J)\in{}A_i\}} s(J)\cdot{}\left( \frac{r_i}{g_i} + (\frac{r_{z^*}}{g_{z^*}} - \frac{r_i}{g_i} )\cdot{}\alpha^* \right)\\
&= l(\alpha^* )=r_{z^*} - h \leq{} r_{z^*} - h' = l' (\beta^*).\\
\end{aligned}
\end{equation*}
Since $l'$ is always decreasing, we must have $\beta^* \leq{} \alpha^*$.

Case 3: $h+c_2 \geq{} r_{z^*}$.

By definition of $\tilde{r}(\mathcal{Y})$, for each $J\in{}\{J\in{}\mathcal{X}:m(J)\in{}A_{z^*}\}$, $\tilde{r}(\mathcal{Y})(J) = s(J)\cdot \frac{r_{z^*}}{g_{z^*}}$.
For each $J\in{}\{J\in{}\mathcal{X}:m(J)=z^*\}$, $\tilde{r}(\mathcal{X})(J) \geq{} s(J)\cdot \frac{r_{z^*}}{g_{z^*}}$ in general.
For each $i\in{}f(z^*)$, for each $J\in{}\{J\in{}\mathcal{X}:m(J)\in{}A_i\}$, $\tilde{r}(\mathcal{X})(J) \geq{} s(J)\cdot \frac{r_{z^*}}{g_{z^*}}$ in general, again.

\end{proof}

\subsection {Proof of Theorem \ref{thm:tilder_2}}

\begin{proof} [Proof of theorem \ref{thm:tilder_2}]

Recall the $z^\diamond$ is the largest machine type used by the alternative machine configuration with the input set of jobs $\mathcal{J}^{1d}$. For each machine type $z\in \mathcal{M}$, let $H_z$ denote the set of jobs $\{J\in \mathcal{J}^{1d}: m(J)\in A_z\}$ for simplicity.
Let $h:= S(H_{z^\diamond}^h)\cdot  \frac{r_{z^\diamond}}{g_{z^\diamond}}$.
Let $c_1 := \sum_{i\in{}f(z^\diamond)} S(H_i)\cdot \frac{r_i}{g_i}$.
Let $c_2 := S(H_{z^\diamond}\setminus{}H_{z^\diamond}^h)\cdot \frac{r_{z^\diamond}}{g_{z^\diamond}}$, where $H_{z^\diamond}^h=\{J\in \mathcal{J}^{1d}: m(J)=z^\diamond\}$.

Case 1: $h+c_1 \leq r_{z^\diamond}$.
Clearly, in this case, $\cn_{z^\diamond} (z^\diamond)=1$.

Case 1.1: $h>0$.

We have the following results immediately:
\begin{equation*}
\begin{aligned}
&\quad \sum_{J\in \mathcal{J}^{1d}}  \tilde{r}(\mathcal{J}^{1d})(J) \\
&= r_{z^\diamond} + \sum_{z\in T(z^\diamond) \setminus \{z^\diamond\}} S(H_z )\cdot  \frac{r_z}{g_z}\\
&\leq 2\cdot r_{z^\diamond} + \sum_{z\in T(z^\diamond )\setminus \{z^\diamond \}} \cn_{z^\diamond} (z)\cdot r_z, \text{~by~proposition~\ref{pro:cn&r'_total}}\\
&\leq 2\cdot \sum_{z\in T(z^\diamond)} \cn_{z^\diamond} (z)\cdot r_z. \\
\end{aligned}
\end{equation*}
and
\begin{equation*}
\begin{aligned}
&\quad \sum_{z\in T(z^\diamond)} \cn_{z^\diamond} (z)\cdot r_z\\
&= r_{z^\diamond} + \sum_{z\in T(z^\diamond)\setminus \{z^\diamond\}} \cn_{z^\diamond} (z)\cdot r_z\\
&\leq r_{z^\diamond} + \sum_{z\in T(z^\diamond)\setminus \{z^\diamond\}} S(H_z)\cdot \frac{r_z}{g_z},\\
&\quad \text{~since~} \cn_{z^\diamond}(z)\cdot g_z \leq S(H_z), \text{~for~each~} z\in T(z^\diamond)\setminus \{z^\diamond\}\\
&= \sum_{J\in \mathcal{J}^{1d}} \tilde{r}(\mathcal{J}^{1d})(J). \\
\end{aligned}
\end{equation*}

Case 1.2: $h=0$.

In this case, we have $c_1 \leq r_{z^\diamond}$. Immediately,
\begin{equation*}
\begin{aligned}
&\quad \sum_{J\in \mathcal{J}^{1d}} \tilde{r} (\mathcal{J}^{1d} )(J) \\
&= c_1 + \sum_{z\in T(z^\diamond)\setminus \{z^\diamond \}} S(H_z ) \cdot \frac{r_z}{g_z}\\
&\leq r_{z^\diamond} + \sum_{z\in T(z^\diamond)\setminus \{z^\diamond\}} S(H_z ) \cdot \frac{r_z}{g_z}\\
&\leq 2\cdot r_{z^\diamond} + \sum_{z\in T(z^\diamond)\setminus \{z^\diamond\}} \cn_{z^\diamond}(z)\cdot r_z, \text{~by~proposition~\ref{pro:cn&r'_total}} \\
&\leq 2\cdot \sum_{z\in T(z^\diamond)} \cn_{z^\diamond}(z)\cdot r_z. \\
\end{aligned}
\end{equation*}

For the other inequality,
let $k_0:=\max \{m(J) : J\in \mathcal{J}^{1d}\}$. By definition of $z^\diamond$ and $h=0$, $z^\diamond \in P(k_0)\setminus \{k_0\}$. Suppose $z'\in P(k_0)$ such that $p(z') = z^\diamond$. Clearly, $\{1,2,\dots,k_0\}\subset \{1,2,\ldots,z'\}=\cup_{z\in T(z')} A_z$.
Also by the choice of $z^\diamond$, $\sum_{z\in T(z') \cap A_{z^\triangle}} \cn_{z^\diamond} (z)\cdot r_z > r_{z^\triangle}$ for some $z^\triangle \in P(z^\diamond)$. We have
\begin{equation} \label{eq:tilder_case1.2_1}
\begin{aligned}
&\quad \sum_{J\in \mathcal{J}^{1d}} \tilde{r}(\mathcal{J}^{1d})(J) \\
&= \sum_{z\in T(z')} S(H_z) \cdot \frac{r_z}{g_z}\\
&\geq \left(\sum_{z\in T(z')} \cn_{z'}(z)\cdot r_z \right)- r_{z'}, \text{~by~proposition~\ref{pro:cn&r'_total}}\\
&\geq \left(\sum_{z\in T(z') \cap A_{z^\triangle}} \cn_{z^\diamond} (z)\cdot r_z \right)- r_{z'}\\
&> r_{z^\triangle} - r_{z'} \geq \frac{7}{8} \cdot r_{z^\diamond}. \\
\end{aligned}
\end{equation}
Consequently,
\begin{equation*}
\begin{aligned}
&\quad  \sum_{z\in T(z^\diamond)} \cn_{z^\diamond}(z)\cdot r_z\\
&= r_{z^\diamond} + \sum_{z\in T(z^\diamond)\setminus \{z^\diamond\}} \cn_{z^\diamond}(z)\cdot r_z\\
&\leq (\frac{8}{7} + 1)\cdot \sum_{J\in \mathcal{J}^{1d}} \tilde{r}(\mathcal{J}^{1d})(J), \text{~by~equation~(\ref{eq:tilder_case1.2_1})}\\
&= \frac{15}{7} \cdot \sum_{J\in \mathcal{J}^{1d}} \tilde{r}(\mathcal{J}^{1d})(J)\\
\end{aligned}
\end{equation*}

Case2: $h+c_1 > r_{z^\diamond}$ and $h+c_2 < r_{z^\diamond}$.

Again, in this case, we have $\cn_{z^\diamond}(z^\diamond) = 1$. The arguments in this case are identical to case 1.1 and we still have $\sum_{J\in \mathcal{J}^{1d}}  \tilde{r}(\mathcal{J}^{1d})(J) \leq 2\cdot \sum_{z\in T(z^\diamond)} \cn_{z^\diamond} (z)\cdot r_z$ and $\sum_{z\in T(z^\diamond)} \cn_{z^\diamond} (z)\cdot r_z \leq \sum_{J\in \mathcal{J}^{1d}} \tilde{r}(\mathcal{J}^{1d})(J)$.

Case 3: $h+c_1 > r _{z^\diamond}$ and $h+c_2 \geq r_{z^\diamond}$.

In this case, we have
\begin{equation*}
\begin{aligned}
&\sum_{J\in \mathcal{J}^{1d}} \tilde{r}(\mathcal{J}^{1d})(J) \\
&= \sum_{z\in T(z^\diamond)} S(H_z)\cdot \frac{r_z}{g_z}, \text{~by~the~definition~of~$\tilde{r}$}\\
&\leq \sum_{z\in T(z^\diamond)} \cn_{z^\diamond} (z)\cdot r_z, \text{~by~proposition~\ref{pro:cn&r'_head}}. \\
\end{aligned}
\end{equation*}
and
\begin{equation*}
\begin{aligned}
&\quad \sum_{z\in T(z^\diamond)} \cn_{z^\diamond} (z)\cdot r_z\\
&= \cn_{z^\diamond} (z^\diamond)\cdot r_{z^\diamond}  + \sum_{z\in T(z^\diamond)\setminus \{z^\diamond\}} \cn_{z^\diamond} (z)\cdot r_z\\
&\leq 2\cdot \floor*{\frac{S(H_{z^\diamond})}{g_{z^\diamond}}} \cdot r_{z^\diamond} + \sum_{z\in T(z^\diamond)\setminus \{z^\diamond\}} S(H_z)\cdot \frac{r_z}{g_z}\\
&\leq 2\cdot S(H_{z^\diamond})\cdot \frac{r_{z^\diamond}}{g_{z^\diamond}} + \sum_{z\in T(z^\diamond)\setminus \{z^\diamond\}} S(H_z)\cdot \frac{r_z}{g_z}\\
&\leq 2\cdot \sum_{z\in T(z^\diamond)} S(H_z)\cdot \frac{r_z}{g_z}\\
&= 2\cdot \sum_{J\in \mathcal{J}^{1d}} \tilde{r}(\mathcal{J}^{1d})(J).\\
\end{aligned}
\end{equation*}

\end{proof}


\begin{thebibliography}{10}

\bibitem{alicherry2003line}
Mansoor Alicherry and Randeep Bhatia.
\newblock Line system design and a generalized coloring problem.
\newblock In {\em Proceedings of the 11th Annual European Symposium on
  Algorithms (ESA)}, pages 19--30, 2003.

\bibitem{amazonec2}
Amazon.
\newblock Amazon {EC2}, 2021. \url{http://aws.amazon.com/ec2/}

\bibitem{azar2017tight}
Yossi Azar and Danny Vainstein.
\newblock Tight bounds for clairvoyant dynamic bin packing.
\newblock In {\em Proceedings of the 29th ACM Symposium on Parallelism in
  Algorithms and Architectures (SPAA)}, pages 77--86, 2017.

\bibitem{borodin1998}
Allan Borodin and Ran El-Yaniv.
\newblock {\em Online computation and competitive analysis}, volume~53.
\newblock Cambridge University Press Cambridge, 1998.

\bibitem{buchbinder2021prediction}
Niv Buchbinder, Yaron Fairstein, Konstantina Mellou, Ishai Menache, and
  Joseph~(Seffi) Naor.
\newblock Online virtual machine allocation with lifetime and load predictions.
\newblock In {\em ACM International Conference on Measurement and Modeling of
  Computer Systems (SIGMETRICS)}, 2021.

\bibitem{chang2017lp}
Jessica Chang, Samir Khuller, and Koyel Mukherjee.
\newblock {LP} rounding and combinatorial algorithms for minimizing active and
  busy time.
\newblock {\em Journal of Scheduling}, 20(6):657--680, 2017.

\bibitem{flammini2010minimizing}
Michele Flammini, Gianpiero Monaco, Luca Moscardelli, Hadas Shachnai, Mordechai
  Shalom, Tami Tamir, and Shmuel Zaks.
\newblock Minimizing total busy time in parallel scheduling with application to
  optical networks.
\newblock {\em Theoretical Computer Science}, 411(40-42):3553--3562, 2010.

\bibitem{googlecloud2}
Google.
\newblock Google {Cloud}, 2021. \url{https://cloud.google.com/}

\bibitem{khandekar2015real}
Rohit Khandekar, Baruch Schieber, Hadas Shachnai, and Tami Tamir.
\newblock Real-time scheduling to minimize machine busy times.
\newblock {\em Journal of Scheduling}, 18(6):561--573, 2015.

\bibitem{kumar2005approximation}
Vijay Kumar and Atri Rudra.
\newblock Approximation algorithms for wavelength assignment.
\newblock In {\em Proceedings of the 25th International Conference on
  Foundations of Software Technology and Theoretical Computer Science
  (FSTTCS)}, pages 152--163, 2005.

\bibitem{spaa2014}
Yusen Li, Xueyan Tang, and Wentong Cai.
\newblock On dynamic bin packing for resource allocation in the cloud.
\newblock In {\em Proceedings of the 26th ACM Symposium on Parallelism in
  Algorithms and Architectures (SPAA)}, pages 2--11, 2014.

\bibitem{tpds2016}
Yusen Li, Xueyan Tang, and Wentong Cai.
\newblock Dynamic bin packing for on-demand cloud resource allocation.
\newblock {\em IEEE Transactions on Parallel and Distributed Systems},
  27(1):157--170, 2016.

\bibitem{mertzios2015optimizing}
George~B. Mertzios, Mordechai Shalom, Ariella Voloshin, Prudence~W.H. Wong, and
  Shmuel Zaks.
\newblock Optimizing busy time on parallel machines.
\newblock {\em Theoretical Computer Science}, 562:524--541, 2015.

\bibitem{microsoftazure2}
Microsoft.
\newblock Microsoft {Azure}, 2021. \url{https://azure.microsoft.com/}

\bibitem{ren2016clairvoyant}
Runtian Ren and Xueyan Tang.
\newblock Clairvoyant dynamic bin packing for job scheduling with minimum
  server usage time.
\newblock In {\em Proceedings of the 28th ACM Symposium on Parallelism in
  Algorithms and Architectures (SPAA)}, pages 227--237, 2016.

\bibitem{ren2020heterogeneous}
Runtian Ren and Xueyan Tang.
\newblock Busy-time scheduling on heterogeneous machines.
\newblock In {\em Proceedings of the 34th IEEE International Parallel and
  Distributed Processing Symposium (IPDPS)}, pages 306--315, 2020.

\bibitem{ren2016competitiveness}
Runtian Ren, Xueyan Tang, Yusen Li, and Wentong Cai.
\newblock Competitiveness of dynamic bin packing for online cloud server
  allocation.
\newblock {\em IEEE/ACM Transactions on Networking}, 25(3):1324--1331, 2017.

\bibitem{tpds2020}
Runtian Ren, Yuqing Zhu, Chuanyou Li, and Xueyan Tang.
\newblock Interval job scheduling with machine launch cost.
\newblock {\em IEEE Transactions on Parallel and Distributed Systems},
  31(12):2776--2788, 2020.

\bibitem{shalom2014online}
Mordechai Shalom, Ariella Voloshin, Prudence~W.H. Wong, Fencol~C.C. Yung, and
  Shmuel Zaks.
\newblock Online optimization of busy time on parallel machines.
\newblock {\em Theoretical Computer Science}, 560:190--206, 2014.

\bibitem{tpds2019}
Ming~Ming Tan, Runtian Ren, and Xueyan Tang.
\newblock Cloud scheduling with discrete charging units.
\newblock {\em IEEE Transactions on Parallel and Distributed Systems},
  30(7):1541--1551, 2019.

\bibitem{ipdps2016}
Xueyan Tang, Yusen Li, Runtian Ren, and Wentong Cai.
\newblock On first fit bin packing for online cloud server allocation.
\newblock In {\em Proceedings of the 30th IEEE International Parallel and
  Distributed Processing Symposium (IPDPS)}, pages 323--332, 2016.

\bibitem{vazirani2013approximation}
Vijay~V Vazirani.
\newblock {\em Approximation algorithms}.
\newblock Springer Science \& Business Media, 2013.

\bibitem{winkler2003wavelength}
Peter Winkler and Lisa Zhang.
\newblock Wavelength assignment and generalized interval graph coloring.
\newblock In {\em Proceedings of the 14th Annual ACM-SIAM Symposium on Discrete
  Algorithms (SODA)}, pages 830--831, 2003.

\end{thebibliography}
\end{document}